\documentclass[11pt, oneside]{article}   	
\usepackage{geometry}                		
\usepackage{graphicx}				
\usepackage{amssymb}

\usepackage{amsmath}
\usepackage{amsfonts}
\usepackage{amssymb}
\usepackage{amsthm}
\usepackage{mathtools}
\usepackage{algorithm}
\usepackage{algcompatible}
\usepackage{tikz}
\usepackage[round,authoryear]{natbib}
\usepackage{thmtools,thm-restate}
\usepackage{authblk}

\usepackage{comment}
\usepackage{ifpdf}
\usepackage{bm}
\usepackage{color}
\usepackage{url}

\newtheorem{thm}{Theorem}[section]

\newtheorem{lem}{Lemma}

\newtheorem{prop}{Proposition}[section]

\newtheorem{dfn}{Definition}

\newtheorem{rk}{Remark}

\newcommand{\vect}[1]{\boldsymbol{#1}}
\newcommand{\tp}[1]{{#1}^{\mathsf T}}
\renewcommand{\bar}{\overline}
\newcommand{\eps}{\epsilon}
\newcommand{\pa}{\partial}

\renewcommand{\eps}{\varepsilon}
\renewcommand{\epsilon}{\varepsilon}
\renewcommand{\Sigma}{\varSigma}
\newcommand{\Rho}{\mathrm{P}}

\newcommand{\E}{\mathrm E}

\newcommand{\Cor}{\mathrm{Cor}}
\newcommand{\tr}{\mathrm{tr}}
\newcommand{\VEC}{\mathrm{vec}}
\newcommand{\vech}{\mathrm{vech}}
\DeclareMathAlphabet\mathbfcal{OMS}{cmsy}{b}{n}
\DeclareMathOperator{\diag}{diag}
\newcommand{\half}{\frac12}

\newcommand{\unif}{\mathrm{Unif}}
\newcommand{\Dir}{\mathrm{Dir}}
\newcommand{\sqDir}{\mathrm{Dir}^2}
\newcommand{\vmf}{\mathrm{vMF}}
\newcommand{\bing}{\mathrm{Bing}}

\newcommand{\bq}{{\bf q}}
\newcommand{\bv}{{\bf v}}
\newcommand{\bg}{{\bf g}}
\newcommand{\bG}{{\bf G}}
\newcommand{\bGamma}{\vect\Gamma}

\newcommand{\keywords}[1]{\textbf{Keywords:} #1}

\title{Flexible Bayesian Dynamic Modeling of Correlation and Covariance Matrices}
\author[1]{Shiwei Lan\thanks{shiwei@illinois.edu}}
\author[2]{Andrew Holbrook}
\author[3]{Gabriel A. Elias}
\author[4]{Norbert J. Fortin}
\author[5]{Hernando Ombao}
\author[6]{Babak Shahbaba}
\affil[1]{Department of Statistics, University of Illinois Urbana-Champaign, Champaign, IL 61820}
\affil[2,6]{Department of Statistics, University of California-Irvine, Irvine, CA 92697}
\affil[3,4]{Center for the Neurobiology of Learning and Memory, Department of Neurobiology and Behavior, University of California-Irvine, Irvine, CA 92697}
\affil[5]{Statistics Program, King Abdullah University of Science and Technology, Thuwal 23955, Saudi Arabia}

\date{}							

\begin{document}
\maketitle

\begin{abstract}
Modeling correlation (and covariance) matrices can be challenging due to the positive-definiteness constraint and potential high-dimensionality.
Our approach is to decompose the covariance matrix into the correlation and variance matrices and 
propose a novel Bayesian framework based on modeling the correlations as products of unit vectors. 
By specifying a wide range of distributions on a sphere (e.g. the squared-Dirichlet distribution),
the proposed approach induces flexible prior distributions for covariance matrices (that go beyond the commonly used inverse-Wishart prior).
For modeling real-life spatio-temporal processes with complex dependence structures, we extend our method to dynamic cases and introduce unit-vector Gaussian process priors in order to capture the evolution of correlation among components of a multivariate time series. 
To handle the intractability of the resulting posterior, we introduce the adaptive $\Delta$-Spherical Hamiltonian Monte Carlo. 
We demonstrate the validity and flexibility of our proposed framework in a simulation study of periodic processes and an analysis of rat's local field potential activity in a complex sequence memory task.
\end{abstract}


\section{Introduction}
\label{sec:intro}
Modeling covariance matrices---or more broadly, positive definite (PD) matrices---is one of the most fundamental problems in statistics. In general, the task is difficult because the number of parameters grows quadratically with the dimension of the matrices. The complexity of the challenge increases substantially if we allow dependencies to vary over time (or space) in order to account for the dynamic (non-stationary) nature of the underlying probability model. In this paper, we propose a novel solution to the problem by developing a flexible and yet computationally efficient Bayesian inferential framework for both static and dynamic covariance matrices. 

This work is motivated by modeling the dynamic brain connectivity (i.e., associations between brain activity at different regions).  
In light of recent technical advances that allow the collection of large, multidimensional neural activity datasets, brain connectivity analyses are emerging as critical tools in neuroscience research. Specifically, the development of such analytical tools will help elucidate fundamental mechanisms underlying cognitive processes such as learning and memory, and identify potential biomarkers for early detection of neurological disorders.
There are a number of new methods that have been developed \citep{cribben12,fiecas16,lindquist14,ting15,prado13} but the main limitation of these methods (especially the ones that have a frequentist approach) is a lack of natural framework for inference. Moreover, parametric approaches (e.g. vector auto-regressive models) need to be tested for adequacy for modeling complex brain processes and often have high dimensional parameter spaces (especially with a large number of channels and high lag order).
This work provides both a nonparametric Bayesian model and an efficient inferential method for modeling the complex dynamic dependence among multiple stochastic processes that is common in the study of brain connectivity.

Within the Bayesian framework, it is common to use an inverse-Wishart prior on the covariance matrix for computational convenience \citep{mardia80,anderson03}.
This choice of prior however is very restrictive (e.g. common degrees of freedom for all components of variance) \citep{barnard00,tokuda11}.
\cite{daniels1999a, daniels2001} propose uniform shrinkage priors. \cite{daniels1999b} discuss three hierarchical priors to generalize the inverse-Wishart prior. 
Alternatively, one may use decomposition strategies for more flexible modeling choices (see \cite{barnard00} for more details). For instance, \cite{banfield93}, \cite{yang94}, \cite{celeux95}, \cite{leonard92}, \cite{chiu96}, and \cite{bensmail97} propose methods based on the spectral decomposition of the covariance matrix. Another strategy is to use the Cholesky decomposition of the covariance matrix or its inverse, e.g., \cite{pourahmadi99,pourahmadi00,liu93,pinheiro96}. 
There are other approaches directly related to correlation, including the constrained model based on truncated distributions \citep{liechty2004}, the Cholesky decomposition of correlation matrix using an angular parametrization \citep{POURAHMADI2015}, and methods based on partial autocorrelation and parameterizations using angles \citep{rapisarda2007}.
In general, these methods fail to yield full flexibility and generality; and often sacrifice statistical interpretability.

While our proposed method in this paper is also based on the separation strategy \citep{barnard00} and the Cholesky decomposition, the main distinction from the existing methods is that it represents each entry of the correlation matrix as a product of unit vectors. 
This in turn provides a flexible framework for modeling covariance matrices without sacrificing interpretability. 
Additionally, this framework can be easily extended to dynamic settings in order to model real-life spatio-temporal processes with complex dependence structures that evolve over the course of the experiment.

To address the constraint for correlation processes (positive definite matrix at each time having unit diagonals and off-diagonal entries with magnitudes no greater than 1), we introduce unit-vector Gaussian process priors. 
There are other related works, e.g. generalized Wishart process \citep{wilson10}, and latent factor process \citep{fox15}, that explore the product of vector Gaussian processes. In general they do not grant full flexibility in simultaneously modeling the mean, variance and correlation processes.
For example, latent factor based models link the mean and covariance processes through a loading matrix, which is restrictive and undesirable if the linear link is not appropriate, and thus are outperformed by our proposed flexible framework (See more details in Section \ref{sec:flexibility}).
Other approaches to model non-stationary processes use a representation in terms of a basis such as wavelets \citep{nason00,park14,cho15} and the SLEX \citep{ombao05}, which are actually inspired by the Fourier representations in the Dahlhaus locally stationary processes \cite{dahlhaus00,priestley65}. These approaches are frequentist and do not easily provide a framework for inference (e.g., obtaining confidence intervals).
The class of time-domain parametric models allows for the ARMA parameters to evolve over time \citep[see, e.g.][]{rao70} or via parametric latent signals \citep{west99,prado01}. A restriction for this class of parametric models is that some processes might not be adequately modeled by them.


This main contributions of this paper are:
(a.) a sphere-product representation of correlation/covariance matrix is introduced to induce flexible priors for correlation/covariance matrices and processes;
(b.) a general and flexible framework is proposed for modeling mean, variance, and correlation processes separately;
(c.) an efficient algorithm is introduced to infer correlation matrices and processes;
(d.) the posterior contraction of modeling covariance (correlation) functions with Gaussian process prior is studied for the first time.

The rest of the paper is organized as follows. In the next section, we present a geometric view of covariance matrices and extend this view to allow covariance matrices to change over time. In Section \ref{sec:postinfer}, we use this geometrical perspective to develop an effective and computationally efficient inferential method for modeling static and dynamic covariance matrices. Using simulated data, we will evaluate our method in Section \ref{sec:numerics}. In Section \ref{sec:lfp}, we apply our proposed method to local field potential (LFP) activity data recorded from the hippocampus of rats performing a complex sequence memory task \citep{allen14,allen16,ng17}. In the final section, we conclude with discussions on the limitations of the current work and future extensions.

\section{Structured Bayesian Modeling of the Covariance (Correlation) Matrices}\label{sec:strct_model}
To derive flexible models for covariance and correlation matrices, we start with the Cholesky decomposition, form a sphere-product representation, and finally obtain the separation decomposition in \cite{barnard00} with correlations represented as products of unit vectors. 
The sphere-product representation is amenable for the inferential algorithm to handle the resulting intractability, and hence lays the foundation for full flexibility in choosing priors.

Any covariance matrix $\vect\Sigma=[\sigma_{ij}]>0$ is symmetric positive definite, 
and hence has a unique Cholesky decomposition $\vect\Sigma = {\bf L} \tp{\bf L}$
where the Cholesky factor ${\bf L}=[l_{ij}]$ is a lower triangular matrix such that $\sigma_{ij} = \sum_{k=1}^{\min\{i,j\}} l_{ik} l_{jk}$.
We denote the variance vector as $\vect\sigma^2:=\tp{[\sigma^2_1,\cdots,\sigma^2_D]}$,
then each variance component, $\sigma_i^2:=\sigma_{ii}$, can be written in terms of the corresponding row ${\bf l}_i$ of ${\bf L}$ as follows:
\begin{equation}\label{eq:cov_constr}
\sigma_i^2 = \sum_{k=1}^{i} l_{ik}^2 = \Vert {\bf l}_i\Vert^2, \quad {\bf l}_i := [l_{i1},l_{i2},\cdots, l_{ii}] 
\end{equation}
For $\vect\Sigma$ to be positive definite,
it is equivalent to require all the leading principal minors $\{M_i\}$ to be positive,
\begin{equation}\label{eq:minor_constr}
M_i = \prod_{k=1}^i l_{kk}^2 >0,\; i=1,\cdots, D\; \Longleftrightarrow \; l_{ii} \neq 0,\; i=1,\cdots, D 
\end{equation}
Based on \eqref{eq:cov_constr} and \eqref{eq:minor_constr}, for $i\in\{1,\cdots, D\}$, ${\bf l}_i$ can be viewed as a point on a sphere with radius $\sigma_i$ excluding the equator, denoted as $\mathcal S_0^{i-1}(\sigma_i):=\{{\bf l}\in \mathbb R^i | \Vert {\bf l}\Vert_2=\sigma_i, l_{ii}\neq 0\}$. Therefore the space of the Cholesky factor in terms of its rows can be written as a product of spheres and we require 
\begin{equation}\label{eq:cholcov_manifold}
({\bf l}_1,{\bf l}_2,\cdots,{\bf l}_D) \in \mathcal S_0^0(\sigma_1)\times\mathcal S_0^1(\sigma_2)\cdots\times\mathcal S_0^{D-1}(\sigma_D)
\end{equation}
Note that \eqref{eq:cholcov_manifold} is the sufficient and necessary condition
for the matrix $\vect\Sigma={\bf L} \tp{\bf L}$ to be a covariance matrix. 

We present probabilistic models involving covariance matrices in the following generic form:
\begin{equation}\label{eq:strt_cov}
\begin{aligned}
{\bf y}|\vect\Sigma(\vect\sigma, {\bf L}) &\sim \ell ({\bf y} ; \vect\Sigma(\vect\sigma, {\bf L}) ), \quad \vect\Sigma(\vect\sigma, {\bf L}) = {\bf L} \tp{\bf L} \\
\vect\sigma &\sim p(\vect\sigma) \\
{\bf L} | \vect\sigma &\sim p({\bf L}; \vect\sigma),\quad \tp\vech({\bf L})\in \prod_{i=1}^D \mathcal S_0^{i-1}(\sigma_i)
\end{aligned}
\end{equation}
where $\vect\sigma:=\tp{[\sigma_1,\cdots,\sigma_D]}$, and the half-vectorization in row order, $\tp\vech$, transforms the lower triangular matrix ${\bf L}$ into a vector $({\bf l}_1,{\bf l}_2,\cdots,{\bf l}_D)$. 
The total dimension of $(\vect\sigma, {\bf L})$ is $\frac{D(D+1)}{2}$.\footnote{For each $i\in \{1,\cdots, D\}$, given $\sigma_i$, there are only $(i-1)$ free parameters on $\mathcal S_0^{i-1}(\sigma_i)$, so there are totally $\frac{D(D-1)}{2}+D$ free parameters.}

Alternatively, if we separate variances from covariance, then we have a unique Cholesky decomposition for the correlation matrix $\vect\Rho=[\rho_{ij}] = {\bf L}^* \tp{({\bf L}^*)}$,
where the Cholesky factor ${\bf L}^*=\diag(\vect\sigma^{-1}){\bf L}$ can be obtained by normalizing each row of ${\bf L}$.
The magnitude requirements for correlations 
are immediately satisfied by the Cauchy-Schwarz inequality: 
$|\rho_{ij}| = \frac{|\sigma_{ij}|}{\sigma_i \sigma_j}= \frac{|\langle {\bf l}_i, {\bf l}_j \rangle|}{\Vert {\bf l}_i\Vert_2 \Vert {\bf l}_j\Vert_2} \leq 1$.
Thus we require
\begin{equation}\label{eq:cholcorr_manifold}
({\bf l}^*_1,{\bf l}^*_2,\cdots,{\bf l}^*_D) \in \mathcal S_0^0\times\mathcal S_0^1\cdots\times\mathcal S_0^{D-1}
\end{equation}
where $\mathcal S_0^{i-1}:=\mathcal S_0^{i-1}(1)$. 
Similarly, \eqref{eq:cholcorr_manifold} is the sufficient and necessary condition for $\vect\Rho={\bf L}^* \tp{({\bf L}^*)}$ to be a correlation matrix.
Then we have the following alternatively structured model for covariance $\vect\Sigma$ that involves correlation $\vect\Rho$ explicitly
\begin{equation}\label{eq:strt_corr}
\begin{aligned}
{\bf y}|\vect\Sigma(\vect\sigma, {\bf L}^*) &\sim \ell ({\bf y} ; \vect\Sigma(\vect\sigma, {\bf L}^*) ), \quad  \vect\Sigma(\vect\sigma, {\bf L}^*) = \diag(\vect\sigma) \vect\Rho \diag(\vect\sigma), \quad \vect\Rho = {\bf L}^* \tp{({\bf L}^*)} \\
\vect\sigma &\sim p(\vect\sigma) \\
{\bf L}^* &\sim p({\bf L}^*),\quad \tp\vech({\bf L}^*)\in \prod_{i=1}^D \mathcal S_0^{i-1}
\end{aligned}
\end{equation}
Note, this direct decomposition $\vect\Sigma = \diag(\vect\sigma) \vect\Rho \diag(\vect\sigma)$ as a \emph{separation strategy} is motivated by statistical thinking in terms of standard deviations and correlations \citep{barnard00}.
This setting is especially relevant if the statistical quantity of interest is correlation matrix $\vect\Rho$ itself, and we can then skip inference of the standard deviation $\vect\sigma$ by fixing it to a data-derived point estimate.

In what follows, we will show that the above framework includes the inverse-Wishart prior as a special case, but it can be easily generalized to a broader range of priors for additional flexibility. Such flexibility enables us to better express prior knowledge, control the model complexity and speed up computation in modeling real-life phenomena. This is crucial in modeling spatio-temporal processes with complex structures.

\subsection{Connection to the Inverse-Wishart Prior}
There are some interesting connections between the spherical product representations \eqref{eq:cholcov_manifold} \eqref{eq:cholcorr_manifold} and the early development of the Wishart distribution \citep{wishart28}.
The original Wishart distribution was derived by orthogonalizing multivariate Gaussian random variables leading to a lower triangular matrix whose elements $\{t_{ij}^*| i\geq j\}$ (analogous to $l_{ij}$ or $l^*_{ij}$) were called \emph{rectangular coordinates}.
This way, the probability density has a geometric interpretation as a product of volumes and approximate densities on a series of spherical shells with radius $\{t_{ii}^*\}$ \citep[See more details in][]{sverdrup47, anderson03}. 
Now we demonstrate that the proposed schemes \eqref{eq:strt_cov} \eqref{eq:strt_corr} include the commonly used inverse-Wishart prior as a special case in modeling covariances.

Suppose $\vect\Sigma$ is a random sample from the inverse-Wishart distribution $\mathcal W^{-1}_D(\vect\Psi,\nu)$ with the scale matrix $\vect\Psi>0$ and the degree of freedom $\nu\ge D$. Therefore, 
$\vect\Sigma^{-1} \sim \mathcal W_D(\vect\Psi^{-1},\nu)$.
Denote ${\bf C}$ as the Cholesky factor of $\vect\Psi^{-1}$, i.e. $\vect\Psi^{-1} = {\bf C} \tp{\bf C}$. 
Then $\vect\Sigma^{-1}$ has the following Bartlett decomposition  \citep{anderson03,smith72}
\begin{equation}\label{eq:bartlett}
\vect\Sigma^{-1} = {\bf T}\, \tp{\bf T}, \quad {\bf T}:= {\bf C}{\bf T}^*, \quad t^*_{ij} \sim 
\begin{cases} \chi_{D-i+1}, & i=j\\
\mathcal N(0,1), & i>j\\\
\delta_0, & i<j
\end{cases}
\end{equation}
where the lower triangular matrix ${\bf T}$, named \emph{Bartlett factor}, has the following density \citep[Theorem 7.2.1 of][]{anderson03}
\begin{equation*}
p({\bf T}) = \frac{|\vect\Psi|^{\nu/2}}{2^{D(\nu-2)/2} \Gamma_D(\nu/2)} \prod_{i=1}^D |t_{ii}|^{\nu-i} \exp\left(-\half\tr (\vect\Psi {\bf T}\tp{\bf T}) \right)
\end{equation*}
with multivariate gamma function defined as $\Gamma_D(x):= \pi^{D(D-1)/4} \prod_{i=1}^D\Gamma[x+(1-i)/2]$.

Now taking the inverse of the first equation in \eqref{eq:bartlett} yields the following \emph{reversed Cholesky decomposition}\footnote{This can be achieved through the \emph{exchange matrix} (a.k.a. reversal matrix, backward identity, or standard involutory permutation) ${\bf E}$ with 1's on the anti-diagonal and 0's elsewhere. Note that ${\bf E}$ is both involutory and orthogonal, i.e. ${\bf E}={\bf E}^{-1}=\tp{\bf E}$.
Let ${\bf E}\vect\Sigma{\bf E}={\bf L}\tp{\bf L}$ be the usual Cholesky decomposition. Then $\vect\Sigma = ({\bf E}{\bf L}{\bf E}) \tp{({\bf E}{\bf L}{\bf E})}={\bf U} \tp{\bf U}$ and define ${\bf U}:={\bf E}{\bf L}\tp{\bf E}$.}
\begin{equation*}\label{eq:rev_cholcov}
\vect\Sigma = {\bf U} \tp{\bf U}, \qquad \sigma_{ij} = \sum_{k=\max\{i,j\}}^D u_{ik} u_{jk}, \quad \vech(\tp{\bf U}) \in \prod_{i=1}^D \mathcal S_0^{D-i}(\sigma_i)
\end{equation*}
where ${\bf U} := {\bf T}^{-\mathsf T}$ is an upper triangular matrix.
%
The following proposition describes the density of the reversed Cholesky factor ${\bf U}$ of $\vect\Sigma$, which enables us to treat the inverse-Wishart distribution as a special instance of strategy \eqref{eq:strt_cov} or \eqref{eq:strt_corr}.
\begin{prop}
\label{prop:conn_iWishart}
Assume $\vect\Sigma\sim \mathcal W^{-1}_D(\vect\Psi,\nu)$. Then its reversed Cholesky factor ${\bf U}$ has the following density
\begin{equation*}\label{eq:pdf_bartlett}
p({\bf U}) = \frac{|\vect\Psi|^{\nu/2}}{2^{D(\nu-2)/2} \Gamma_D(\nu/2)} |{\bf U}|^{-(\nu+D+1)} \prod_{i=1}^D u_{ii}^i \exp\left(-\half\tr (\vect\Psi {\bf U}^{-\mathsf T} {{\bf U}^{-1}}) \right)
\end{equation*}
\end{prop}
\begin{proof}
See Section 
A
in the supplementary file.
\end{proof}
If we normalize each row of ${\bf U}$ and write
\begin{equation*}
{\bf U} = \diag(\vect\sigma) {\bf U}^*, \quad \sigma_i = \sqrt{\sigma_{ii}}=\Vert {\bf u}_i\Vert, \quad u^*_{ij} = u_{ij}/\sigma_i \, ,
\end{equation*}
then the following joint prior of $(\vect\sigma,{\bf U}^*)$ is inseparable in general:
\begin{equation}\label{eq:jtpri4iwishart}
p(\vect\sigma,{\bf U}^*) \propto \prod_{i=1}^D |\sigma_i u^*_{ii}|^{i-(\nu+D+1)} \exp\left\{-\half\tr (\vect\Psi \diag(\vect\sigma^{-1})({\bf U}^*)^{-\mathsf T} {({\bf U}^*)^{-1}}\diag(\vect\sigma^{-1})) \right\}
\end{equation}
With this result, we can conditionally model variance and correlation factor as $p(\vect\sigma|{\bf U}^*)$ and $p({\bf U}^*|\vect\sigma)$ respectively, similarly as in our proposed scheme \eqref{eq:strt_cov} or \eqref{eq:strt_corr}.
It is also used to verify the validity of our proposed method \eqref{eq:strt_corr} (see more details in Section \ref{sec:validity}).
A similar result exists for the Wishart prior distribution regarding the Cholesky factor.
This representation facilitates the construction of a broader class of more flexible prior distributions for covariance matrix detailed below.

\subsection{More Flexible Priors}\label{sec:priors}
Within the above framework, the only constraint on ${\bf U}$ or ${\bf L}$ is that it resides on the product of spheres with increasing dimensions. Using this fact, we can develop a broader class of priors on covariance matrices and thus be able to model processes with more complicated dependence in covariance structures. Since $\vect\sigma$ and ${\bf L}^*$ have independent priors in \eqref{eq:strt_corr}, in what follows we focus on the scheme \eqref{eq:strt_corr}, and for simplicity, we denote the normalized Cholesky factor as ${\bf L}$. Also, following \cite{barnard00}, we assume a log-Normal prior on $\sigma$:
\begin{equation*}\label{eq:pri_var}
\log(\vect\sigma) \sim \mathcal N(\vect\xi,\vect\Lambda)
\end{equation*}

We now discuss priors on ${\bf L}$ that properly reflect the prior knowledge regarding the covariance structure among variables. If two variables, $y_i$ and $y_j$ (assuming $i<j$) are known to be uncorrelated a priori, i.e. $0=\rho_{ij}=\langle {\bf l}_i, {\bf l}_j\rangle$, then we can choose a prior that encourages ${\bf l}_i \perp {\bf l}_j$, e.g. $l_{jk} \approx 0$ for $k\leq i$. In contrast, if we believe a priori that there is a strong correlation between the two variables, we can specify that ${\bf l}_i$ and ${\bf l}_j$ be linearly dependent, e.g., by setting $[l_{jk}]_{k\leq i} \approx \pm {\bf l}_i\,$.
When there is no prior information, we might assume that components are uncorrelated and specify priors for ${\bf l}_i$ that concentrate on the (two) poles of $\mathcal S_0^{i-1}$,
\begin{equation}\label{eq:pri_cor}
p({\bf l}_i) \propto |l_{ii}|, \quad i=2,\cdots, D
\end{equation}
Putting more prior probability on the diagonal elements of ${\bf L}$ renders fewer non-zero off-diagonal elements, which in turn leads to a larger number of perpendicular variables; that is, such a prior favors zeros in the correlation matrix $\vect\Rho$.
More generally, one can map a probability distribution defined on the simplex onto the sphere
and consider the following \emph{squared-Dirichlet} distribution.
\begin{dfn}[Squared-Dirichlet distribution]
A random vector ${\bf l}_i\in\mathcal S^{i-1}$ is said to have a \emph{squared-Dirichlet} distribution with parameter $\vect\alpha_i:=(\alpha_{i1},\alpha_{i2},\cdots,\alpha_{ii})$ if
\begin{equation*}\label{eq:sqDir_def}
{\bf l}_i^2 := (l_{i1}^2,l_{i2}^2,\cdots, l_{ii}^2) \sim \Dir (\vect\alpha_i)
\end{equation*}
Denote ${\bf l}_i \sim \sqDir (\vect\alpha_i)$. Then ${\bf l}_i$ has the following density
\begin{equation}\label{eq:sqDir_density}
p({\bf l}_i) = p({\bf l}_i^2) | 2{\bf l}_i | \propto ({\bf l}_i^2)^{\vect\alpha_i-1} |{\bf l}_i| = |{\bf l}_i|^{2\vect\alpha_i-1} := \prod_{k=1}^i |l_{ik}|^{2\alpha_{ik}-1}
\end{equation}
\end{dfn}
\begin{rk}
This definition includes a large class of flexible prior distributions on the unit sphere that specify different concentrations of probability density through the parameter $\vect\alpha_i$.
For example, the above prior \eqref{eq:pri_cor} corresponds to $\vect\alpha_i=(\half,\cdots,\half,1)$.
\end{rk}

\begin{figure}[t] 
   \centering
   \includegraphics[width=1\textwidth,height=.35\textwidth]{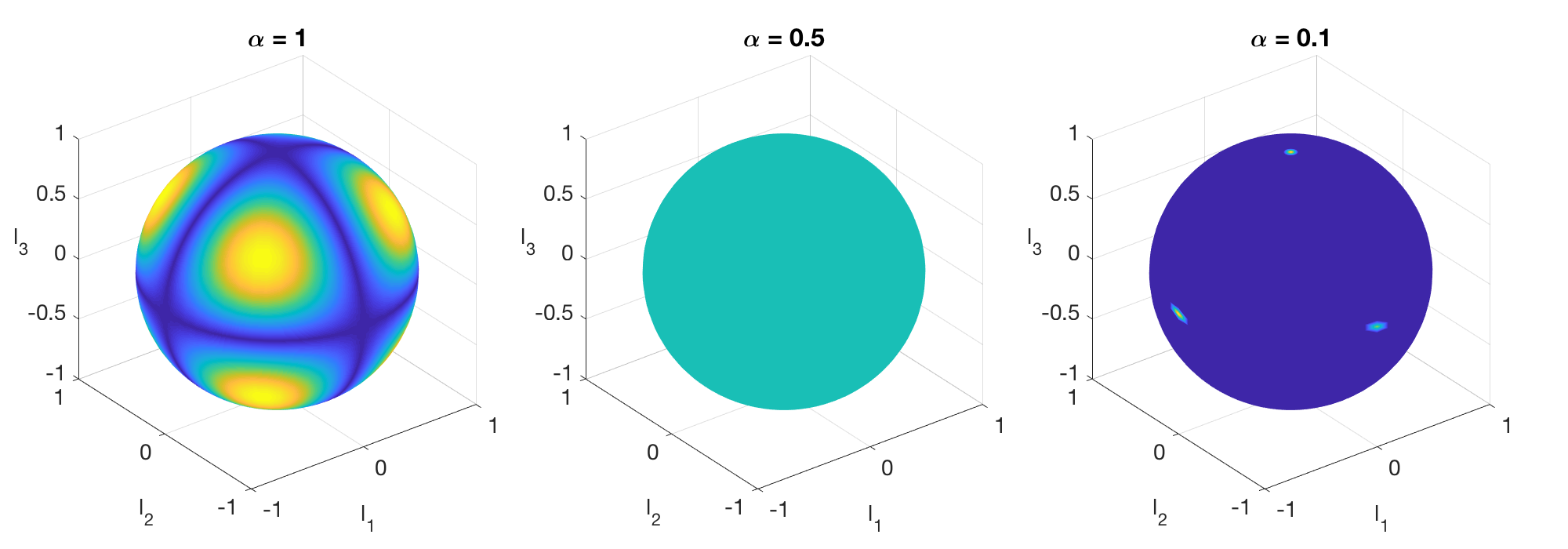} 
   \caption{Symmetric squared-Dirichlet distributions $\sqDir (\vect\alpha)$ defined on the 2-sphere with different settings for concentration parameter $\vect\alpha=\alpha{\bm 1}$.
                The uniform distribution on the simplex, $\Dir({\bm 1})$, becomes non-uniform on the sphere due to the stretch of geometry (left);
                the symmetric Dirichlet distribution $\Dir(\half {\bm 1})$ becomes uniform on the sphere (middle);
                with $\alpha$ closer to 0, the induced distribution becomes more concentrated on the polar points (right).}
   \label{fig:sqdirichlet}
\end{figure}

To induce a prior distribution for the correlation matrix $\vect\Rho = {\bf L} \tp{\bf L}$, one can specify priors on row vectors of ${\bf L}$, ${\bf l}_i \sim \sqDir (\vect\alpha_i)$ for $i=2,\cdots, D$.
To encourage small correlation, we choose the concentration parameter $\vect\alpha_i$ so that the probability density concentrates around the (two) poles of $\mathcal S_0^{i-1}$, e.g. $0<\alpha_{ik}\ll \alpha_{ii}$ for $k < i$.
Figure \ref{fig:sqdirichlet} illustrates the density heat maps of some symmetric squared-Dirichlet distributions $\sqDir (\alpha {\bm 1})$ on $\mathcal S^2$.
It is interesting that the squared-Dirichlet distribution induces two important uniform prior distributions over correlation matrices from \cite{barnard00} in an effort to provide flexible priors for covariance matrices, as stated in the following theorem.
\begin{thm}[Uniform distributions]
\label{thm:conn_unif}
Let $\vect\Rho={\bf L}\tp{\bf L}$.
Suppose ${\bf l}_i \sim \sqDir (\vect\alpha_i)$, for $i=2,\cdots,D$, are independent, where ${\bf l}_i$ is the $i$-th row of ${\bf L}$.
We have the following
\begin{enumerate}
\item If $\vect\alpha_i = ( \half \tp{\bm 1}_{i-1}, \alpha_{ii} ), \; \alpha_{ii} = \frac{(i-2)D-1}{2}$,
then $\vect\Rho$ follows a \emph{marginally} uniform distribution, that is, 
$\rho_{ij} \sim \unif(-1,1),\quad i\neq j$.
\item If $\vect\alpha_i = ( \half \tp{\bm 1}_{i-1}, \alpha_{ii} ), \; \alpha_{ii} = \frac{D-i}{2}+1$,
then $\vect\Rho$ follows a \emph{jointly} uniform distribution, that is, 
$p(\vect\Rho) \propto 1$.
\end{enumerate}
\end{thm}
\begin{proof}
See Section 
A
in the supplementary file.
\end{proof}

Another natural spherical prior can be obtained by constraining a multivariate Gaussian random vector to have unit norm.
This is later generalized to a vector Gaussian process constrained to a sphere that serves as a suitable prior for modeling correlation processes.
Now we consider the following \emph{unit-vector Gaussian} distribution:
\begin{dfn}[Unit-vector Gaussian distribution]
A random vector ${\bf l}_i\in\mathcal S^{i-1}$ is said to have a \emph{unit-vector Gaussian} distribution with mean $\vect\mu$ and covariance $\vect\Sigma$ if
\begin{equation*}\label{eq:uvG_def}
{\bf l}_i \sim \mathcal N_i(\vect\mu, \vect\Sigma), \quad with \; \Vert {\bf l}_i\Vert_2 =1
\end{equation*}
Then we denote ${\bf l}_i \sim \mathcal N_i^{\mathcal S}(\vect\mu, \vect\Sigma)$ and ${\bf l}_i$ has the following (conditional) density
\begin{equation*}\label{eq:uvG_density}
p({\bf l}_i|\,\Vert {\bf l}_i\Vert_2=1) 
= \frac{1}{(2\pi)^{\frac{i}{2}}|\vect\Sigma|^{\half}} \exp\left\{-\half \tp{({\bf l}_i-\vect\mu)} \vect\Sigma^{-1} ({\bf l}_i-\vect\mu) \right\}, \quad \Vert {\bf l}_i\Vert_2 =1
\end{equation*}
\end{dfn}
\begin{rk}
This conditional density essentially defines the following \emph{Fisher-Bingham} distribution \citep[a.k.a. \emph{generalized Kent} distribution,][]{kent82,mardia09}. 
If $\vect\Sigma={\bf I}$, then the above distribution reduces to the \emph{von Mises-Fisher} distribution \citep{fisher53,mardia09} as a special case. 
If in addition $\vect\mu={\bf 0}$, then the above density becomes a constant; that is, the corresponding distribution is uniform on the sphere $\mathcal S_0^{i-1}$.
See more details in Section 
E.1
of the supplementary file.
\end{rk}

\subsection{Dynamically Modeling the Covariance Matrices}\label{sec:dynamic}
We can generalize the proposed framework for modeling covariance/correlation matrices to the dynamic setting
by adding subscript $t$ to variables in the model \eqref{eq:strt_cov} and the model \eqref{eq:strt_corr}, thus called \emph{dynamic covariance} and \emph{dynamic correlation} models respectively.
We focus the latter in this section.
One can model the components of $\vect\sigma_t$ as independent dynamic processes using, e.g. ARMA, GARCH, or log-Gaussian process. For ${\bf L}_t$, we use vector processes. Since each row of ${\bf L}_t$ has to be on a sphere of certain dimension, we require the unit norm constraint for the dynamic process over time.
We refer to any multivariate process ${\bf l}_i(x)$ satisfying $\Vert {\bf l}_i(x) \Vert \equiv 1, \; \forall x \in \mathcal X$ as \emph{unit-vector process (uvP)}.
A unit-vector process can be obtained by constraining an existing multivariate process, e.g. the \emph{vector Gaussian process (vGP)},
as defined below.
\begin{dfn}[Vector Gaussian process]
A $D$-dimensional \emph{vector Gaussian process} ${\bf Z}(x):=(Z_1(x),\cdots, Z_D(x))$, with vector mean function $\vect\mu(x)=(\mu_1(x),\cdots,\mu_D(x))$, covariance function $\mathcal C$ and ($D$-dimensional) cross covariance ${\bf V}_{D\times D}$,
\begin{equation*}
{\bf Z}(x) \sim \mathcal{GP}_D(\vect\mu,\mathcal C, {\bf V}_{D\times D})
\end{equation*}
is a collection of $D$-dimensional random vectors, indexed by $x\in \mathcal X$, such that for any finite set of indices $\{x_1,\cdots, x_N\}$, 
the random matrix $\widetilde{\bf Z}_{N\times D}:=\tp{({\bf Z}(x_1),\cdots,{\bf Z}(x_N))}$ has the following matrix normal distribution
\begin{equation*}
\widetilde{\bf Z}_{N\times D} \sim \mathcal{MN}_{N\times D} ({\bf M}_{N\times D}, {\bf K}_{N\times N}, {\bf V}_{D\times D})
\end{equation*}
where ${\bf M}_{N\times D}:=({\bf m}_1,\cdots,{\bf m}_D)$, and ${\bf m}_k=\tp{(\mu_k(x_1),\cdots,\mu_k(x_N))}$, and ${\bf K}$ is the kernel matrix with elements $K_{ij}=\mathcal C(x_i,x_j)$.
\end{dfn}
\begin{rk}
Note for each $k=1,\cdots D$, we have the following marginal GP
\begin{equation*}
Z_k(x) \sim \mathcal{GP}(\vect\mu_k, \mathcal C)
\end{equation*}
In the above definition, we require a common kernel $\mathcal C$ for all the marginal GPs, whose dependence
is characterized by the cross covariance ${\bf V}_{D\times D}$.
On the other hand, for any fixed $x^*\in \mathcal X$, we have
\begin{equation*}
{\bf Z}(x^*) \sim \mathcal N_D(\vect\mu(x^*), {\bf V}_{D\times D})
\end{equation*}
For simplicity, we often consider $\mu\equiv 0$ and ${\bf V}_{D\times D}={\bf I}_D$. That is, $Z_k(x)\overset{iid}{\sim} \mathcal{GP}(0,\mathcal C)$ for $k=1,\cdots,D$.
\end{rk}

Restricting vGP ${\bf Z}(\cdot)$ to sphere yields a \emph{unit-vector Gaussian process (uvGP)} ${\bf Z}^*(\cdot):={\bf Z}(\cdot) |\, \{\Vert {\bf Z}(\cdot)\Vert_2\equiv 1\}$, denoted as ${\bf Z}^*(\cdot) \sim \mathcal{GP}_D^{\mathcal S}(\vect\mu,\mathcal C, {\bf V})$.
Note for any fixed $x^*\in\mathcal X$, ${\bf Z}^*(x^*) \sim \mathcal N_D^{\mathcal S}(\vect\mu, {\bf V})$. 
Setting $\mu\equiv 0$, ${\bf V}={\bf I}$, and conditioned on the length $\ell_n$ of each row of $\widetilde{\bf Z}$, we have
\begin{equation*}
p(\widetilde{\bf Z}^*|\,\{\Vert z_{n\cdot}\Vert=\ell_n\}) 
= \frac{\prod_{n=1}^N \ell_n^D}{(2\pi)^{\frac{ND}{2}}|{\bf K}|^{\frac{D}{2}}} \exp\left\{-\half\tr \left[ \tp{(\widetilde{\bf Z}^*)} \diag(\{\ell_n\}){\bf K}^{-1}\diag(\{\ell_n\}) \widetilde{\bf Z}^* \right]\right\}
\end{equation*}
This conditional density is preserved by the inference algorithm in Section \ref{sec:postinfer} and used for defining priors for correlations with all $\ell_n=1$.
For each marginal GP, we select the following exponential function as the common kernel
\begin{equation*}
\mathcal C(x,x') = \gamma \exp(-0.5\Vert x-x'\Vert^s/\rho^s)
\end{equation*}
where $s$ controls the smoothness, the scale parameter $\gamma$ is given an inverse-Gamma prior, and the correlation length parameter $\rho$ is given a log-normal prior.
Figure \ref{fig:sph_proc} shows a realization of vector GP ${\bf Z}_t$, unit-vector GP (forming rows of) ${\bf L}_t$ and the induced correlation process $\vect\Rho_t$ respectively. 

\begin{figure}[t] 
   \centering
   \includegraphics[width=1\textwidth,height=.28\textwidth]{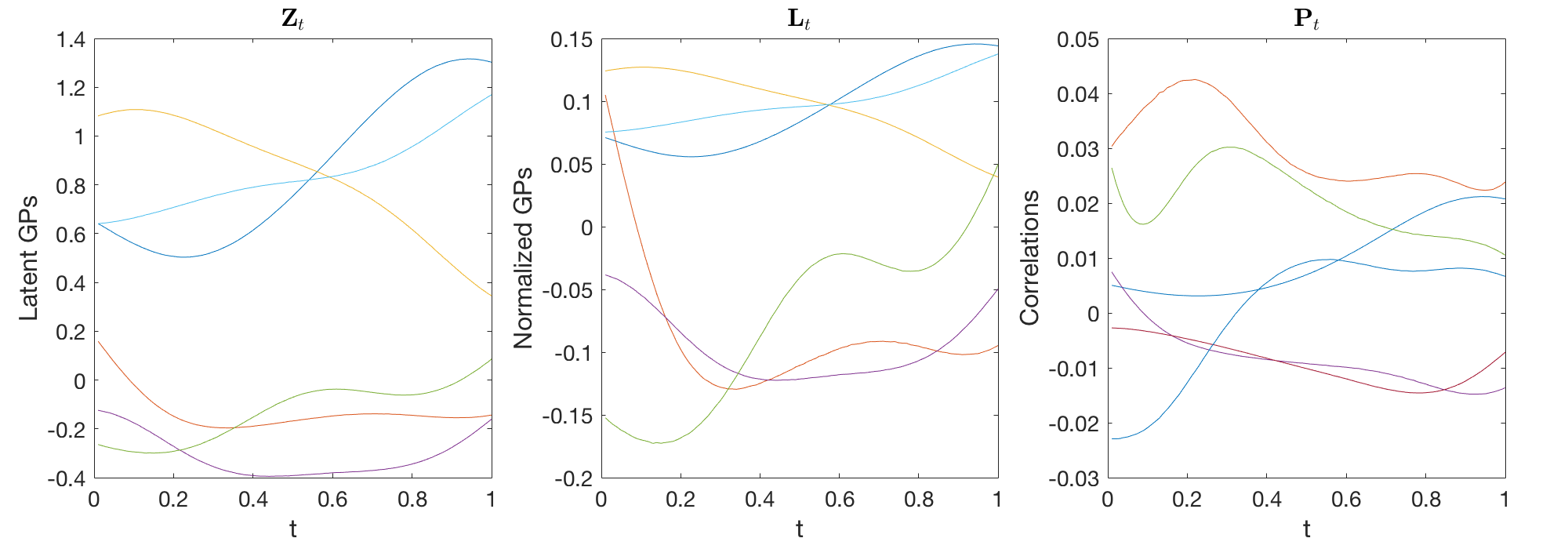} 
   \caption{A realization of vector GP ${\bf Z}_t$ (left), unit-vector GP (forming rows of) ${\bf L}_t$ (middle) and the induced correlation process $\vect\Rho_t$ (right).}
   \label{fig:sph_proc}
\end{figure}


In what follows, we focus on multivariate time series; therefore, we use the one dimensional time index $t\in \mathcal X=\mathbb R^+$. The overall dynamic correlation model can be summarized as follows:
\begin{equation}\label{eq:strt_dyn_model}
\begin{aligned}
{\bf y}_t \sim \mathcal N(\vect\mu_t,\vect\Sigma_t), &\quad \vect\Sigma_t = \diag(\vect\sigma_t) {\bf L}_t \tp{{\bf L}_t} \diag(\vect\sigma_t) \\
\vect\mu_t \sim \mathcal{GP}_D(0,\mathcal C_{\mu},{\bf I}), &\quad \mathcal C_\mu(t,t') = \gamma_{\mu} \exp(-0.5\Vert t-t'\Vert^s/\rho_{\mu}^s) \\
\log \vect\sigma_t \sim \mathcal{GP}_D(0,\mathcal C_{\sigma},{\bf I}), &\quad \mathcal C_\sigma(t,t') = \gamma_{\sigma} \exp(-0.5\Vert t-t'\Vert^s/\rho_{\sigma}^s) \\
{\bf l}_i(t) \sim \mathcal{GP}_i^{\mathcal S}({\bf n}_i,\mathcal C_L,{\bf I}), &\quad \mathcal C_L(t,t') = \gamma_L \exp(-0.5\Vert t-t'\Vert^s/\rho_L^s) \\
\gamma_* \sim \Gamma^{-1}(a_*,b_*), &\quad \log\rho_* \sim \mathcal N(m_*,V_*), \; * = \mu, \sigma, \,\textrm{or}\, L
\end{aligned}
\end{equation}
where a constant mean function ${\bf n}_i=(0,\cdots,0,1)$ is used in the uvGP prior for ${\bf l}_i(t)$, with mean matrix ${\bf M}={\bm 1}_N\otimes \tp{\bf n}_i$ for the realization $\tilde{\bf l}_i$.
This model \eqref{eq:strt_dyn_model} captures the spatial dependence in the matrix $\vect\Sigma_t$, which evolves along the time; while the temporal correlation is characterized by various GPs.
The induced covariance process $\vect\Sigma_t$ is not a generalized Wishart process \citep{wilson10}, which only models Cholesky factor of covariance using GP.
Though with GP, dynamic covariance model may work similarly as the dynamic correlation model \eqref{eq:strt_dyn_model}, yet the latter provides extra flexibility in modeling the evolution of variances and correlations separately.
In general such flexibility could be useful in handling constraints for processes, e.g. modeling the dynamic probability for binary time series.

With this structured model \eqref{eq:strt_dyn_model}, one can naturally model the evolution of variances and correlations separately in order to obtain more flexibility. 
If the focus is on modeling the correlation among multiple time series, then one can substitute $\vect\sigma_t$ with a point estimate $\widehat{\vect\sigma}$ from one trial and assume a steady variance vector. Alternatively, if sufficient trials are present, one can obtain an empirical estimate, $\widehat{\vect\sigma}_t$, from multiple trials at each time point.
In the following, we study the posterior contraction of GP modeling in this setting.

\subsection{Posterior Contraction Theorem}
We now provide a theorem on the posterior contraction of the dynamic covariance model 
before we conclude this section.
Because the posterior contraction for mean regression using Gaussian process has been vastly investigated in the literature \citep{vanderVaart08a,vandervaart09,vanderVaart11,yang16},
we only investigate the posterior contraction for the covariance regression and set $\vect\mu_t\equiv 0$.
We leave the posterior contraction of the dynamic correlation model \eqref{eq:strt_dyn_model} for future work.
Note, the Cholesky decomposition of covariance matrix $\vect\Sigma={\bf L}\tp{\bf L}$ is unique if all the diagonal entries of ${\bf L}$ are positive.
Therefore in the remaining of this section, we identify Cholesky factors up to a column-wise sign, i.e. ${\bf L} \sim {\bf L} \diag(-\sum_{j\in J} e_j)$ for $J\subset\{1,\cdots,D\}$ where $e_j$ is the $j$-th column of identity matrix ${\bf I}_D$.

In most cases, Gaussian process ${\bf L}_t$ can be viewed as a tight Borel measurable map in a Banach space, e.g. a space of continuous functions or $L_p$ space.
It is well known that the support of a centered GP is equal to the closure of the \emph{reproducible kernel Hilbert space (RKHS)} $\mathbb H$ associated to this process \citep[Lemma 5.1 of][]{vanderVaart08b}.
Because the posterior distribution necessarily puts all its mass on the support of the prior, the posterior consistency requires the true parameter ${\bf L}_0$ governing the distribution of the data to fall in this support \citep{vanderVaart08a}. 
Following \cite{vanderVaart08a,vanderVaart08b,vanderVaart11}, we express the rate of the posterior contraction in terms of the \emph{concentration} function
\begin{equation}\label{eq:concent_1}
\phi_{{\bf L}_0}(\eps) = \inf_{h\in\mathbb H : \,\Vert h-{\bf L}_0\Vert < \eps} \Vert h\Vert_{\mathbb H}^2 - \log \Pi({\bf L} :\, \Vert {\bf L}\Vert < \eps)
\end{equation}
where $\Vert\cdot\Vert$ is the norm of the Banach space where the GP ${\bf L}$ takes value, $\Pi$ is the GP prior and $\mathbb H$ is the associated RKHS with norm $\Vert\cdot\Vert_{\mathbb H}$.
Under certain regularity conditions, the posterior contracts with increasing data expressed in $n$ at the rate $\eps_n\rightarrow 0$ satisfying
\begin{equation}\label{eq:concent_2}
\phi_{{\bf L}_0}(\eps_n) \leq n \eps_n^2
\end{equation}

Define $\Vert{\bf L}(t)\Vert_\infty := \max_{1\leq i,j\leq D} \sup_{t\in\mathcal X} |l_{ij}(t)|$. 
Consider the separable Banach space $L^\infty(\mathcal X)^{D(D+1)/2}:=\{{\bf L}(t): \Vert{\bf L}(t)\Vert_\infty<+\infty\}$.
Let $p$ be a (centered) Gaussian model, which is uniquely determined by the covariance matrix $\vect\Sigma={\bf L}\tp{\bf L}$.
Therefore the model density is parametrized by ${\bf L}$, hence denoted as $p_{\bf L}$.
Denote $P_{\bf L}^{(n)}:=\otimes_{i=1}^n P_{{\bf L},i}$ as the product measure on $\otimes_{i=1}^n(\mathcal X_i,\mathcal B_i,\mu_i)$.
Each $P_{{\bf L},i}$ has a density $p_{{\bf L}_i}$ with respect to the $\sigma$-finite measure $\mu_i$.
Define the average Hellinger distance as $d_n^2({\bf L},{\bf L}')=\frac1n\sum_{i=1}^n\int (\sqrt{p_{{\bf L},i}}-\sqrt{p_{{\bf L}',i}})^2 d\mu_i$.
Denote the observations $Y^{(n)}=\{Y_i\}_{i=1}^n$ with $Y_i={\bf y}(t_i)$. Note they are independent but not identically distributed (inid).
Now we state the main theorem of posterior contraction.
\begin{thm}[Posterior contraction]
\label{thm:post_contr}
Let ${\bf L}-{\bf I}$ be a Borel measurable, zero-mean tight Gaussian random element in $L^\infty(\mathcal X)^{D(D+1)/2}$ and $P_{\bf L}^{(n)}=\otimes_{i=1}^n P_{{\bf L},i}$ be the product measure of $Y^{(n)}$ parametrized by ${\bf L}$.
Let $\phi_{{\bf L}_0}$ be the function in \eqref{eq:concent_1} with the uniform norm $\Vert\cdot\Vert_\infty$.
If ${\bf L}_0$ is contained in the support of ${\bf L}$ and  $\phi_{{\bf L}_0}$ satisfies \eqref{eq:concent_2} with $\eps_n\geq n^{-\half}$,
then $\Pi_n({\bf L} : d_n({\bf L}, {\bf L}_0)> M_n\eps_n | Y^{(n)}) \rightarrow 0$ in $P_{{\bf L}_0}^{(n)}$-probability for every $M\rightarrow \infty$.
\end{thm}
\begin{proof}
See Section 
B
in the supplementary file.
\end{proof}
\begin{rk}
In principle, the smoothness of GP should match the regularity of the true parameter to achieve the optimal rate of contraction \citep{vanderVaart08a,vanderVaart11}.
One can scale GP, e.g. using an inverse-Gamma bandwidth, to get optimal contraction rate for every regularity level so that the resulting estimator is rate adaptive \citep{vandervaart09,vanderVaart11}.
One can refer to Section 3.2 of \citep{vanderVaart11} for posterior contraction rates using squared exponential kernel for GP. 
We leave further investigation on contraction rates in the setting of covariance regression to future work.
\end{rk}
\begin{rk}
Here the GP prior ${\bf L}$ defines a (mostly finite) probability measure on the space of bounded functions. The true parameter function ${\bf L}_0$ is required to be contained in the support of the prior, the RKHS of ${\bf L}$. The contraction rate depends on the position of ${\bf L}_0$ relative to the RKHS and the small-ball probability $\Pi(\Vert {\bf L}\Vert<\eps)$.
\end{rk}

\section{Posterior Inference}\label{sec:postinfer}
Now we obtain the posterior probability of mean $\vect\mu_t$, variance $\vect\sigma_t$, Cholesky factor of correlation ${\bf L}_t$, hyper-parameters $\vect\gamma:=(\gamma_\mu,\gamma_\sigma,\gamma_L)$ and $\vect\rho:=(\rho_\mu,\rho_\sigma,\rho_L)$ in the model \eqref{eq:strt_dyn_model}.
Denote the realization of processes $\vect\mu_t, \vect\sigma_t, {\bf L}_t$ at discrete time points $\{t_n\}_{n=1}^N$ as $\widetilde{\vect\mu}_{N\times D}, \widetilde{\vect\sigma}_{N\times D}, \widetilde{\bf L}_{N\times D\times D}$ respectively.
Transform the parameters $\widetilde{\vect\tau}:=\log(\widetilde{\vect\sigma})$, $\vect\eta:=\log(\vect\rho)$ for the convenience of calculation. 
Denote $\widetilde{\bf Y}_{M\times N\times D}:=\{{\bf Y}_1,\cdots, {\bf Y}_M\}$ for $M$ trials, $({\bf Y}_m)_{N\times D}:=\tp{[{\bf y}_{m1},\cdots, {\bf y}_{mN}]}$ and ${\bf y}_{mn}^*:=({\bf y}_{mn}-\vect\mu_n) \circ e^{-\vect\tau_n}$ where $\circ$ is the Hadamard product (a.k.a. Schur product), i.e. the entry-wise product.
Let ${\bf K}_*(\gamma_*,\eta_*) = \gamma_* {\bf K}_{0*}(\eta_*)$ and $\tilde{\bf l}_i^*:=\tilde{\bf l}_i-{\bm 1}_N\otimes \tp{\bf n}_i$. 

\subsection{Metropolis-within-Gibbs}\label{sec:sample_procedure}
We use a Metropolis-within-Gibbs algorithm and alternate updating the model parameters $\widetilde{\vect\mu}, \widetilde{\vect\tau}, \widetilde{\bf L}, \vect\gamma, \vect\eta$.  We now list the parameters and their respective updates one by one.

\noindent $(\vect\gamma)$. \quad
Note the prior for $\vect\gamma$ is conditionally conjugate given $* = \mu, \tau, \,\textrm{or}\, L$,
\begin{equation*}
\gamma_* | \tilde*, \eta_* \sim \Gamma^{-1}(a'_*,b'_*), \quad a'_* = a_* + \frac{ND}{2} (\frac{D+1}{2}-\frac1D)^{[*=L]},\quad b'_* = b_* + \half \tr (\tp{\tilde*} {\bf K}_{0*}(\eta_*)^{-1} \tilde* )
\end{equation*}
where $[{\rm condition}]$ is 1 with the condition satisfied and 0 otherwise.

\noindent $(\vect\eta)$. \quad
Given $* = \mu, \tau, \,\textrm{or}\, L$, we could sample $\eta_*$ using the slice sampler \citep{neal03}, which only requires log-posterior density and works well for scalar parameters,
\begin{equation*}
\log p(\eta_*| \tilde*, \gamma_*) = -\frac{D(\frac{D+1}{2}-\frac1D)^{[*=L]}}{2} \log|{\bf K}_{0*}(\eta_*)| - \frac{\tr (\tp{\tilde*} {\bf K}_{0*}(\eta_*)^{-1} \tilde* )}{2\gamma_*} - \frac{(\eta_*-m_*)^2}{2V_*}
\end{equation*}

\noindent $(\widetilde{\vect\mu})$. \quad
By the definition of vGP, we have $\widetilde{\vect\mu}| \gamma_\mu,\eta_\mu \sim \mathcal{MN}_{N\times D}({\bf 0},{\bf K}_\mu,{\bf I}_D)$; therefore,
$\VEC(\widetilde{\vect\mu}) |\gamma_\mu,\eta_\mu \sim \mathcal N_{ND}({\bf 0},{\bf I}_D\otimes {\bf K}_\mu)$.
On the other hand, one can write
\begin{equation*}
\begin{aligned}
\sum_{m=1}^M \sum_{n=1}^N \tp{{\bf y}_{mn}^*} \vect\Rho_n^{-1} {\bf y}_{mn}^* &= \sum_{m=1}^M \tp{\VEC(\tp{({\bf Y}_m-\widetilde{\vect\mu})})} \diag(\{\widetilde{\vect\Sigma}_n^{-1}\}) \VEC(\tp{({\bf Y}_m-\widetilde{\vect\mu})}) \\
&= \sum_{m=1}^M \tp{(\VEC({\bf Y}_m)-\VEC(\widetilde{\vect\mu}))} \widetilde{\vect\Sigma}_K^{-1} (\VEC({\bf Y}_m)-\VEC(\widetilde{\vect\mu}))
\end{aligned}
\end{equation*}
where $\widetilde{\vect\Sigma}_K^{-1}:=K_{(D,N)} \diag(\{\widetilde{\vect\Sigma}_n\})^{-1} K_{(N,D)}$, and
$K_{(N,D)}$ is the \emph{commutation matrix} of size $ND\times ND$ such that for any $N\times D$ matrix ${\bf A}$, $K_{(N,D)} \VEC({\bf A})=\VEC(\tp{\bf A})$ \citep{tracy69,magnus79}. Therefore, the prior on $\VEC(\widetilde{\vect\mu})$ is conditionally conjugate, and we have
\begin{align*}
\VEC(\widetilde{\vect\mu}) | \widetilde{\bf Y}, \widetilde{\vect\Sigma}, \gamma_\mu,\eta_\mu &\sim \mathcal N_{ND}(\vect\mu', \vect\Sigma')\\ 
\vect\mu' &= \vect\Sigma' \widetilde{\vect\Sigma}_K^{-1} \sum_{m=1}^M \VEC({\bf Y}_m), \quad \vect\Sigma' = \left({\bf I}_D\otimes{\bf K}_\mu^{-1}+M \widetilde{\vect\Sigma}_K^{-1}\right)^{-1}
\end{align*}

\noindent $(\widetilde{\vect\tau})$. \quad
Using a similar argument by 
matrix Normal prior for $\widetilde{\vect\tau}$, we have
$\VEC(\widetilde{\vect\tau})| \gamma_\tau,\eta_\tau \sim \mathcal N_{ND}({\bf 0},{\bf I}_D\otimes {\bf K}_\tau)$.
Therefore, we could use the elliptic slice sampler \citep[ESS,][]{murray10}, which only requires the log-likelihood
\begin{equation*}
\log p(\widetilde{\vect\tau}; \widetilde{\bf Y},\widetilde{\vect\mu}) 
= - M \tp{\bm 1}_{ND} \VEC(\widetilde{\vect\tau}) - \sum_{m=1}^M \half \tp{\VEC({\bf Y}_m^*)} \widetilde{\vect\Rho}_K^{-1} \VEC({\bf Y}_m^*)
\end{equation*}
where $\widetilde{\vect\Rho}_K^{-1}:=K_{(D,N)} \diag(\{\widetilde{\vect\Rho}_n\})^{-1} K_{(N,D)}$ and ${\bf Y}_m^*:=({\bf Y}_m-\widetilde{\vect\mu}) \circ \exp(-\widetilde{\vect\tau})$. 

\noindent $(\widetilde{\bf L})$. \quad
For each $n\in\{1,\cdots N\}$, we have $\tp\vech({\bf L}_n)\in \prod_{i=1}^D \mathcal S_0^{i-1}$. 
We could sample from its posterior distribution using the \emph{$\Delta$-Spherical Hamiltonian Monte Carlo ($\Delta$-SphHMC)} described below. The log-posterior density of $\widetilde{\bf L}$ is
\begin{equation*}
\log p(\widetilde{\bf L}| \widetilde{\bf Y}, \widetilde{\vect\mu},\widetilde{\vect\tau},\gamma_L,\eta_L) = - \sum_{n=1}^N \left[ M \log |{\bf L}_n| + \sum_{m=1}^M \half \tp{{\bf y}_{mn}^*} \vect\Rho_n^{-1} {\bf y}_{mn}^*\right] - \half \sum_{i=2}^D \tr (\tilde{\bf l}_i^{*\mathsf T} {\bf K}_L^{-1} \tilde{\bf l}_i^*)
\end{equation*}
The derivative of log-likelihood with respect to ${\bf L}_n$ and the derivative of log-prior with respect to $\tilde{\bf l}_i$ can be calculated as
\begin{equation*}
\frac{\pa \log p(\widetilde{\bf L}; \widetilde{\bf Y},\widetilde{\vect\mu},\widetilde{\vect\tau}) }{\pa {\bf L}_n} = - M \frac{{\bf I}_D}{{\bf L}_n} + \sum_{m=1}^M \mathrm{tril} ( \vect\Rho_n^{-1} {\bf y}_{mn}^* \tp{{\bf y}_{mn}^*} {\bf L}_n^{-\mathsf T} ), \;
\frac{\pa \log p(\widetilde{\bf L}| \gamma_L,\eta_L) }{\pa \tilde{\bf l}_i} = - {\bf K}_L^{-1} \tilde{\bf l}_i^*
\end{equation*}

\subsection{Spherical HMC}\label{sec:SphHMC}
We need an efficient algorithm to handle the intractability in the posterior distribution of $\widetilde{\bf L}$ introduced by various flexible priors.
\emph{Spherical Hamiltonian Monte Carlo} \citep[SphHMC,][]{lan14b,lan16} is a Hamiltonian Monte Carlo \citep[HMC,][]{duane87,neal11} algorithm on spheres that
can be viewed as a special case of geodesic Monte Carlo \citep{byrne13}, or manifold Monte Carlo methods \citep{girolami11,lan14a}.
The algorithm was originally proposed to handle norm constraints in sampling so it is natural to use it to sample each row of the Cholesky factor of a correlation matrix with unit 2-norm constraint. The general notation $\bq$ is instantiated as ${\bf l}_i$ in this section.

Assume a probability distribution with density function $f(\bq)$ is defined on a $(D-1)$ dimensional sphere with radius $r$, $\mathcal S^{D-1}(r)$.
Due to the norm constraint, there are $(D-1)$ free parameters $\bq_{-D}:=(q_1,\cdots,q_{D-1})$,
which can be viewed as the Cartesian coordinates for the manifold $\mathcal S_+^{D-1}(r)$.
To induce Hamiltonian dynamics on the sphere, we define the potential energy for position $\bq$ as $U(\bq):= -\log f(\bq)$.
Endowing the canonical spherical metric $\bG(\bq_{-D}) = {\bf I}_{D-1}+ \frac{\bq_{-D}\tp \bq_{-D}}{q_D^2}$ on the Riemannian manifold $\mathcal S^{D-1}(r)$,
we introduce the auxiliary velocity vector $\bv|\bq \sim \mathcal N({\bf 0}, \bG(\bq)^{-1})$ and define the associated kinetic energy as $K(\bv;\bq):= -\log f_{\mathcal N}(\bv|\bq)=-\half \log|\bG(\bq_{-D})|+\half \tp \bv_{-D} \bG(\bq_{-D})\bv_{-D}$ \citep{girolami11}.
Therefore the total energy is defined as
\begin{equation}\label{eq:energy}
E(\bq,\bv) := U(\bq) + K(\bv;\bq) 
= \tilde U(\bq) + K_0(\bv;\bq)
\end{equation}
where we denote $\tilde U(\bq):=U(\bq) -\half \log|\bG(\bq_{-D})|=-\log f(\bq) +\log|q_D|$, and $K_0(\bv;\bq):=\half \tp \bv_{-D} \bG(\bq_{-D})\bv_{-D}=\half \tp \bv \bv$ \citep{lan16}.
Therefore the Lagrangian dynamics with above total energy \eqref{eq:energy} is \citep{lan14a}
\begin{equation}\label{eq:LD}
\begin{aligned}
\dot \bq_{-D} & =  \bv_{-D}\\
\dot \bv_{-D} & = -\tp \bv_{-D} \bGamma(\bq_{-D}) \bv_{-D} - \bG(\bq_{-D})^{-1} \nabla_{\bq_{-D}} \tilde U(\bq)
\end{aligned}
\end{equation}
where $\bGamma(\bq_{-D}) = r^{-2}\bG(\bq_{-D})\otimes \bq_{-D}$ is the Christoffel symbols of second kind \citep[see details in][for $r=1$]{lan16}.
A splitting technique is applied to yield the following geometric integrator \citep{lan14b,lan16}, which also includes the last coordinates $q_D, v_D$:
\begin{equation}\label{eq:sphHMC_prop}
\begin{aligned}
\bv^- &= \bv - \frac{h}{2} \mathcal P(\bq) \bg(\bq)\\
\begin{bmatrix}
\bq' \\ \bv^+
\end{bmatrix}
&= 
\begin{bmatrix}
r & 0\\
0 & \Vert \bv^-\Vert_2
\end{bmatrix}
\begin{bmatrix}
\cos(\Vert \bv^-\Vert_2r^{-1} h) + \sin(\Vert \bv^-\Vert_2r^{-1} h)\\
- \sin(\Vert \bv^-\Vert_2r^{-1} h) + \cos(\Vert \bv^-\Vert_2r^{-1} h)
\end{bmatrix}
\begin{bmatrix}
r^{-1} & 0\\
0 & \Vert \bv^-\Vert_2^{-1}
\end{bmatrix}
\begin{bmatrix}
\bq \\ \bv^-
\end{bmatrix}\\
\bv' &= \bv^+ - \frac{h}{2} \mathcal P(\bq') \bg(\bq')
\end{aligned}
\end{equation}
where $\bg(\bq):=\nabla_{\bq}\tilde U(\bq)$, $\mathcal P(\bq):={\bf I}_D - r^{-2}\bq \tp\bq$.
\eqref{eq:sphHMC_prop} defines a mapping $\mathcal T_h: (\bq, \bv)\mapsto (\bq',\bv')$.
Denote $\Vert {\bf u}\Vert_{\mathcal P(\bq)}^2:= \tp {\bf u} \mathcal P(\bq) {\bf u}$.
After applying such integrator $T$ times, a proposal $(\bq_T, \bv_T)=\mathcal T_h^T(\bq_0, \bv_0)$ is accepted with the following probability
\begin{equation}\label{eq:sphHMC_acpt}
\begin{aligned}
a_{sphHMC} =& 1 \wedge \exp(-\Delta E)\\
\Delta E 
=& \tilde U(\bq_T) - \tilde U(\bq_0) - \frac{h^2}{8} \left[ \Vert \bg(\bq_T) \Vert_{\mathcal P(\bq)}^2 - \Vert \bg(\bq_0) \Vert_{\mathcal P(\bq)}^2 \right] \\
& - \frac{h}{2} \left[ \langle \bv_0, \bg(\bq_0)\rangle + \langle \bv_T, \bg(\bq_T)\rangle \right] - h \sum_{\tau=1}^{T-1} \langle \bv_{\tau}, \bg(\bq_{\tau})\rangle
\end{aligned}
\end{equation}
We can prove the following limiting result \citep{beskos11}.
\begin{thm}
\label{thm:energy_conserv}
Let $h\to 0$ we have the following energy conservation
\begin{equation*}
E(\bq(T), \bv(T)) - E(\bq(0), \bv(0)) = \tilde U(\bq(T)) - \tilde U(\bq(0)) - \int_0^{T} \langle \bv(t), \bg(\bq(t))\rangle dt = 0
\end{equation*}
\end{thm}
\begin{proof}
See Section 
C
in the supplementary file.
\end{proof}

\subsection{Adaptive Spherical HMC}
There are two tuning parameters in HMC and its variants: the step size $h$ and the number of integration (leapfrog) steps $T$. 
Hand tuning heavily relies on domain expertise and could be inefficient. 
Here, we adopt the `No-U-Turn' idea from \cite{hoffman14} and introduce a novel adaptive algorithm that obviates manual tuning of these parameters.

First, for any given step size $h$, we adopt a rule for setting the number of leapfrog steps based on the same philosophy as `No-U-Turn' \citep{hoffman14}.
The idea is to avoid waste of computation occurred (e.g. when the sampler backtracks on its trajectory) without breaking the detailed balance condition for the MCMC transition kernel. 
$\mathcal S^{D-1}(r)$ is a compact manifold where any two points $\bq(0), \bq(t)\in \mathcal \mathcal S^{D-1}(r)$ have bounded geodesic distance $\pi r$.
We adopt the stopping rule for the leapfrog when the sampler exits the orthant of the initial state, that is,
the trajectory measured in geodesic distance is at least $\frac{\pi}{2} r$,
which is equivalent to $\langle \bq(0), \bq(t) \rangle < 0$.
On the other hand, this condition may not be satisfied within reasonable number of iterations because the geometric integrator \eqref{eq:sphHMC_prop} does not follow a geodesic (great circle) in general (only the middle part does), therefore we set some threshold $T_{\max}$ for the number of tests, and adopt the following `Two-Orthants' (as the starting and end points occupy two orthants) rule for the number of leapfrogs:
\begin{equation}\label{eq:sphHMC_2orth}
T_{2orth} = \min_{\tau\in\{0,\cdots,T_{\max}\}} \{\tau: \langle \bq_0, \bq_{\tau} \rangle < 0\}
\end{equation}
Alternatively, 
one can stop the leapfrog steps in a \emph{stochastic} way
based on the geodesic distance travelled:
\begin{equation}\label{eq:sphHMC_stoch}
T_{stoch} = \min_{\tau} \{\tau : Z_{\tau}=0\}, \quad 
Z_{\tau} \sim \mathrm{Bern}(p_{\tau}), \quad p_{\tau} = \frac{r^{-2} \langle \bq_0, \bq_{\tau} \rangle +1}{2}
\end{equation}
These stopping criteria are already time reversible, so the recursive binary tree as in `No-U-Turn' algorithm \citep{hoffman14} is no longer needed.

Lastly, we adopt the \emph{dual averaging} scheme \citep{nesterov09} for the adaptation of step size $h$. 
See \cite{hoffman14} for more details.
We summarize our \emph{Adaptive Spherical Hamiltonian Monte Carlo (adp-SphHMC)} in 
the supplementary file.

To sample ${\bf L}$ (or ${\bf L}_t$), we could update each row vector ${\bf l}_i\in \mathcal S_0^{i-1}$ according to \eqref{eq:sphHMC_prop} (in parallel), and accept/reject $\tp\vech({\bf L})$ (or $\tp\vech({\bf L}_t)$) simultaneously based on \eqref{eq:sphHMC_acpt} in terms of the sum of total energy of all components.
We refer to the resulting algorithm as \emph{$\Delta$-Spherical HMC ($\Delta$-SphHMC)}.

The computational complexity involving GP prior is $\mathcal O(N^3)$, and that of the likelihood evaluation is $\mathcal O(MD^2)$. MCMC updates of $\widetilde{\vect\mu}_{N\times D}, \widetilde{\vect\sigma}_{N\times D}, \widetilde{\bf L}_{N\times D\times D}$ have complexity $\mathcal O(ND)$, $\mathcal O(ND)$ and $\mathcal O(ND^2)$ respectively.
To scale up applications to larger dimension $D$, one could preliminarily classify data into groups, and arrange the corresponding blocks of their covariance/correlation matrix in some `band' along the main diagonal assuming no correlation among groups.
More specifically, we can assume ${\bf L}_t$ is $w$-band lower triangular matrix for each time $t$, i.e. $l_{ij}=0$ for $i<j$ or $i-j\geq w$, then the resulting covariance/correlation matrix will be $(2w-1)$-banded. In this way the complexity of likelihood evaluation and updating $\widetilde{\bf L}$ will be reduced to $\mathcal O(MwD)$ and $\mathcal O(NwD)$ resepctively. Therefore the total computational cost would scale linearly with the dimension $D$.
This technique will be investigated in Section \ref{sec:scalability}.


\section{Simulation Studies}\label{sec:numerics}
In this section, we use simulated examples to illustrate the advantage of our structured models for covariance.
First, we consider the normal-inverse-Wishart problem. Since there is conjugacy and we know the true posterior, we use this to verify our method and investigate flexible priors in Section \ref{sec:priors}.
Then we test our dynamical modeling method in Section \ref{sec:dynamic} on a periodic process model.
Our model manifests full flexibility compared to a state-of-the-art nonparametric covariance regression model based on latent factor process \citep{fox15}.

\subsection{Normal-inverse-Wishart Problem}
Consider the following example involving inverse-Wishart prior
\begin{equation}\label{eq:conj-iwishart}
\begin{aligned}
{\bf y}_n|\vect\Sigma & \sim \mathcal N(\vect\mu_0,\vect\Sigma), \quad n=1,\cdots,N\\
\vect\Sigma & \sim \mathcal W^{-1}_D(\vect\Psi, \nu)
\end{aligned}
\end{equation}
It is known that the posterior of $\vect\Sigma|{\bf Y}$ is still inverse-Wishart distribution:
\begin{equation}\label{eq:post-iwishart}
\vect\Sigma|{\bf Y} \sim \mathcal W^{-1}_D (\vect\Psi + ({\bf Y}-\vect\mu_0)\tp{({\bf Y}-\vect\mu_0)}, \nu+N), \qquad {\bf Y}=\tp{[{\bf y}_1,\cdots, {\bf y}_N]}
\end{equation}
We consider dimension $D=3$ and generate data ${\bf Y}$ with 
$\vect\mu_0={\bf 0}, \quad \vect\Sigma=\vect\Sigma_0=\frac{1}{11} ({\bf I} + {\bm 1}\tp{\bm 1})$
for $N=20$ data points 
so that the prior is not overwhelmed by data.


\begin{figure}[t] 
   \centering
   \includegraphics[width=1\textwidth,height=.4\textwidth]{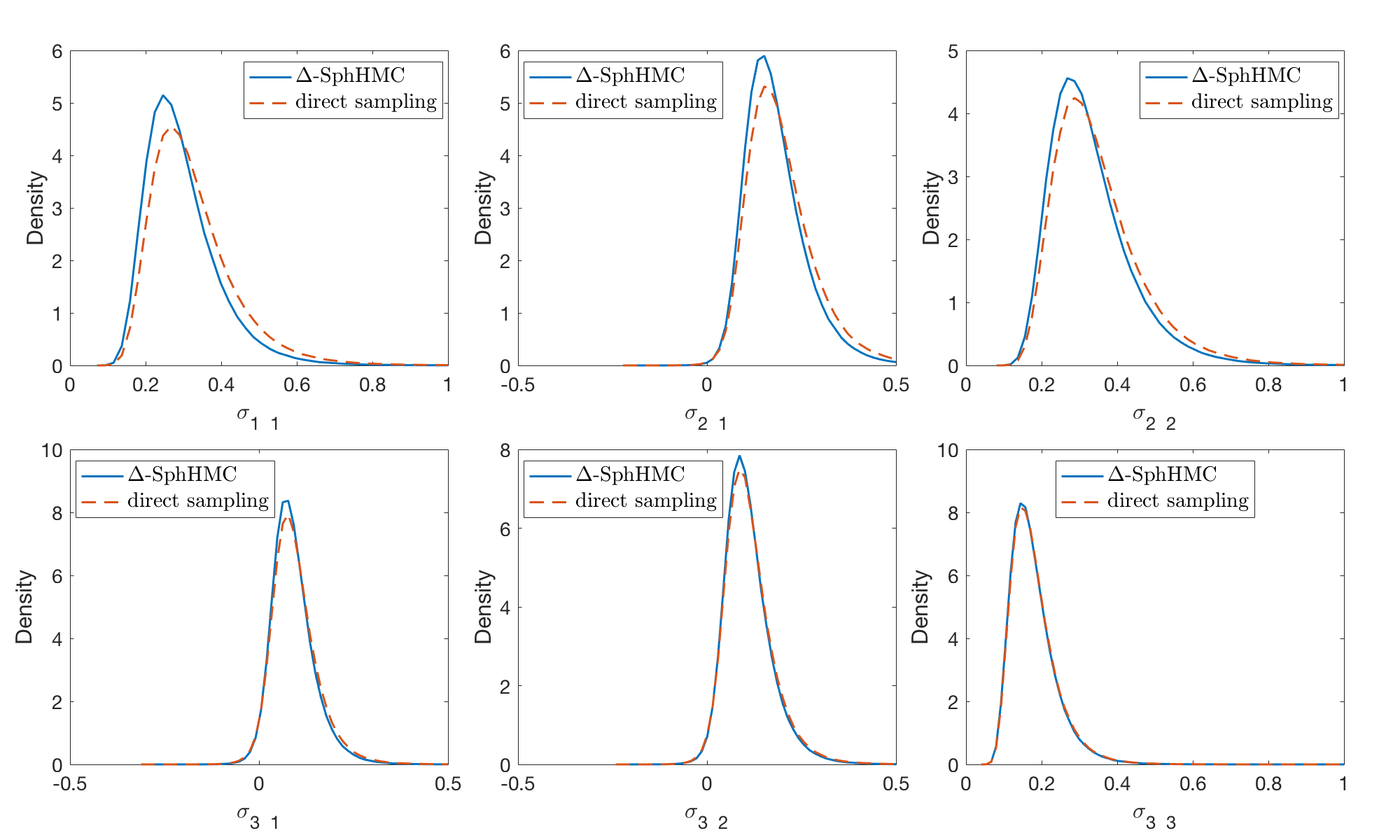} 
   \caption{Marginal posterior densities of $\sigma_{ij}$ in the normal-inverse-Wishart problem. Solid blue lines are estimates by $\Delta$-SphHMC and dashed red lines are estimates by direct sampling. 
   All densities are estimated with $10^6$ samples.}
   \label{fig:post_bartlett}
\end{figure}

\subsubsection{Verification of Validity}\label{sec:validity}
Specifying conditional priors based on \eqref{eq:jtpri4iwishart} in the structured model \eqref{eq:strt_corr},
we want to check the validity of our proposed method by comparing the posterior estimates using $\Delta$-SphHMC agains the truth \eqref{eq:post-iwishart}.

We sample $\vect\tau:=\log(\vect\sigma)$ using standard HMC and ${\bf U}^*$ using $\Delta$-SphHMC. They are updated in Metropolis-Within-Gibbs scheme.
$10^6$ samples are collected after burning the first $10\%$ and subsampling every $1$ of $10$.
For each sample of $\vect\tau$ and $\vech({\bf U}^*)$, we calculate $\vect\Sigma=\diag(e^{\vect\tau}){\bf U}^*\tp{({\bf U}^*)}\diag(e^{\vect\tau})$.
Marginal densities of entries in $\vect\Sigma$ are estimated with these samples and plotted against the results by direct sampling in Figure \ref{fig:post_bartlett}.
Despite of sampling variance, these estimates closely match the results by direct sampling, indicating the validity of our proposed method.

\begin{figure}[t] 
   \centering
   \includegraphics[width=1\textwidth,height=.45\textwidth]{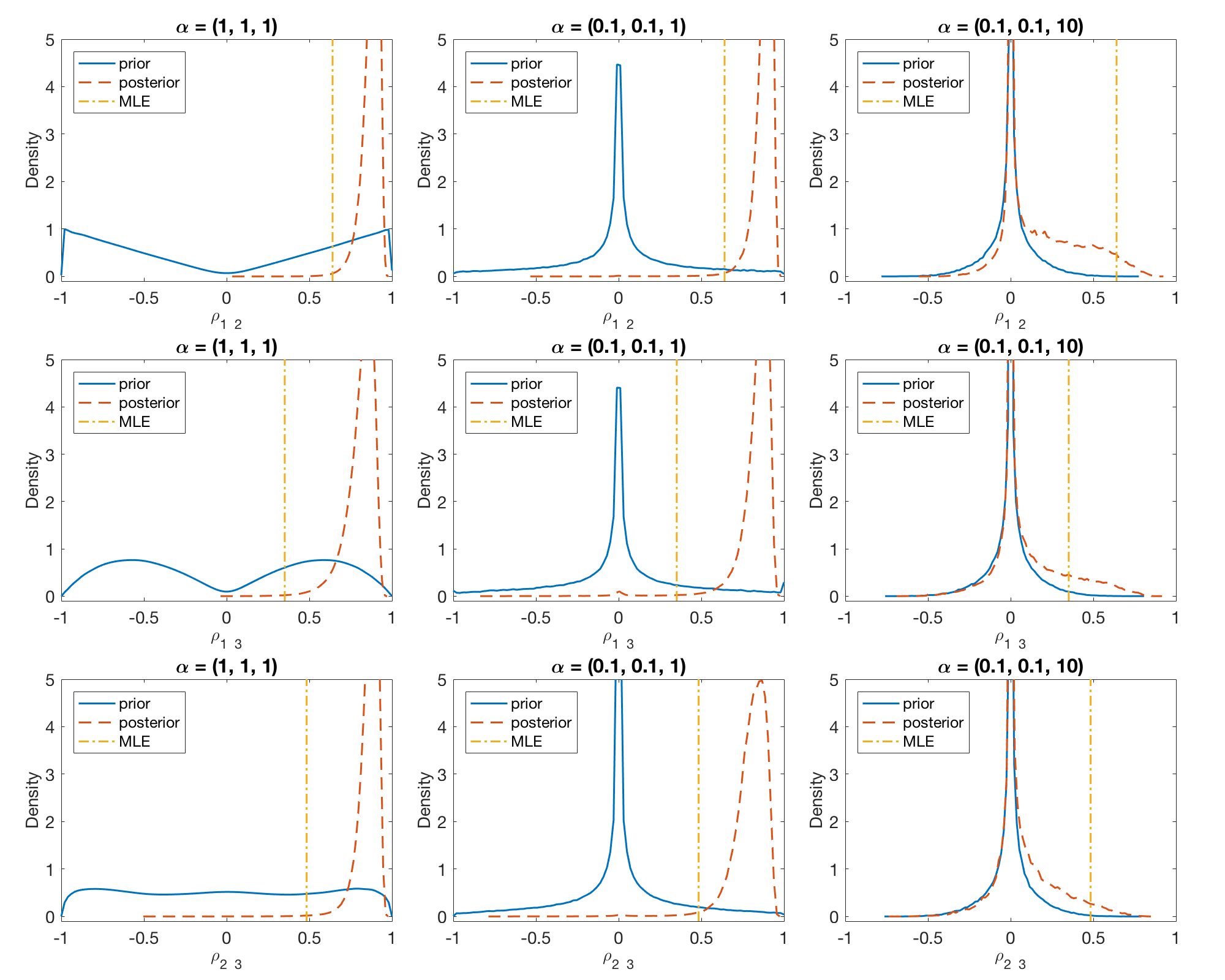} 
   \caption{Marginal posterior, prior (induced from squared-Dirichlet distribution) densities of correlations and MLEs with different settings for concentration parameter $\vect\alpha$, estimated with $10^6$ samples.}
   \label{fig:post_corr_asqdirpri}
\end{figure}

\subsubsection{Examining Flexibility of Priors} \label{sec:prieff}
We have studied several spherical priors for the Cholesky factor of correlation matrix proposed in Section \ref{sec:priors}.
Now we examine the flexibility of these priors in providing prior information for correlation with various parameter settings.


With the same data generated according to \eqref{eq:conj-iwishart}, 
we now consider the squared-Dirichlet prior \eqref{eq:sqDir_density} for ${\bf L}$ in the structured model \eqref{eq:strt_corr} with the following setting
\begin{equation}\label{eq:nonconj-iwishart}
\begin{aligned}
\tau_i=\log(\sigma_i) &\sim \mathcal N(0,0.1^2), \quad i=1,\cdots, D\\
{\bf l}_i &\sim \sqDir (\vect\alpha_i), \quad \vect\alpha_i = (\alpha {\bm 1}_{i-1},\alpha_0), \quad i=2,\cdots, D
\end{aligned}
\end{equation}
where we consider three cases i) $\alpha =1,\,\alpha_0=1$; ii) $\alpha =0.1,\,\alpha_0=1$; iii) $\alpha =0.1,\,\alpha_0=10$.

We generate $10^6$ prior samples (according to \eqref{eq:nonconj-iwishart}) and posterior samples (by $\Delta$-SphHMC) for ${\bf L}$ respectively
and covert them to $\vect\Rho={\bf L}\tp{\bf L}$.
For each entry of $\rho_{ij}$, we estimate the marginal posterior (prior) density based on these posterior (prior) samples.
The posteriors, priors and maximal likelihood estimates (MLEs) of correlations $\rho_{ij}$ are plotted in Figure \ref{fig:post_corr_asqdirpri} for different $\vect\alpha$'s respectively.
In general, the posteriors are compromise between priors and the likelihoods (MLEs).
With more and more weight (through $\vect\alpha$) put around the poles (last component) of each factor sphere, the priors become increasingly dominant that the posteriors (red dash lines) almost fall on priors (blue solid lines) when $\vect\alpha=(0.1,0.1,10)$.
In this extreme case, the squared-Dirichlet distributions induce priors in favor of trivial (zero) correlations.
We have similar conclusion on Bingham prior and von Mises-Fisher prior but results are reported in
Section 
E.1
of the supplementary file.

\subsection{Simulated Periodic Processes}
In this section, we investigate the performance of our dynamic model \eqref{eq:strt_dyn_model} on the following periodic process example
\begin{equation}\label{eq:periodic}
\begin{aligned}
y(t) &\sim \mathcal N_D(\mu(t),\Sigma(t)), \quad \Sigma(t)=L(t)\tp{L(t)} \circ S, \quad t \in [0,2] \\
\mu_i(t) &= \sin(it\pi/D), \quad L_{ij}(t) = (-1)^i\sin(it\pi/D) (-1)^j\cos(jt\pi/D), \quad j\leq i=1,\cdots, D,\\
&\phantom{= \sin(it\pi/D), \;}\quad S_{ij} = (|i-j|+1)^{-1}, \quad i,j=1,\cdots, D
\end{aligned}
\end{equation}
Based on the model \eqref{eq:periodic}, we generate $M$ trials (process realizations) of data $y$ at $N$ evenly spaced points for $t$ in $[0,2]$, and therefore the whole data set $\{y(t)\}$ is an $M\times N\times D$ array.
We first consider $D=2$ to investigate the posterior contraction phenomena and the model flexibility; then we consider $D=100$ over a shorter period $[0,1]$ to show the scalability using the `$w$-band' structure.

\begin{figure}[t] 
   \centering
   \includegraphics[width=.495\textwidth,height=.225\textwidth]{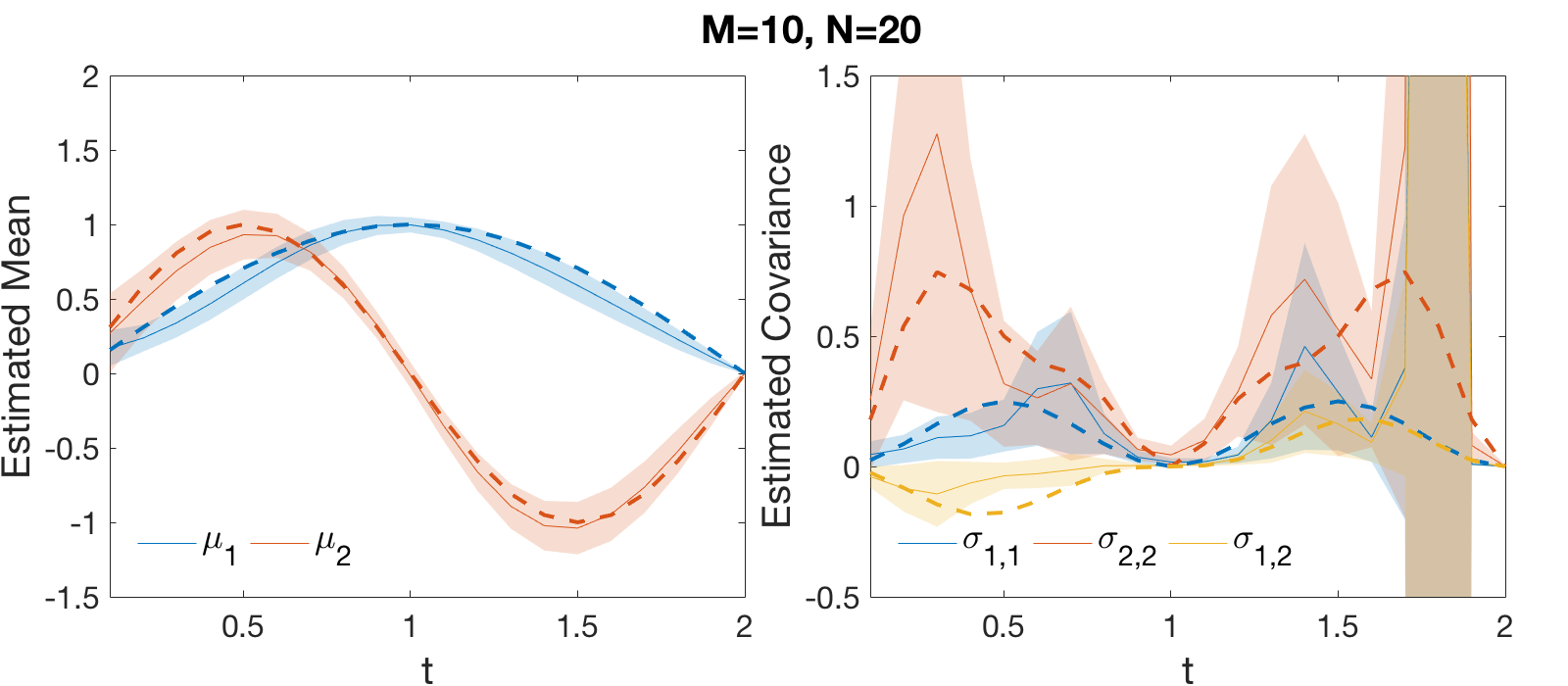}
   \includegraphics[width=.495\textwidth,height=.225\textwidth]{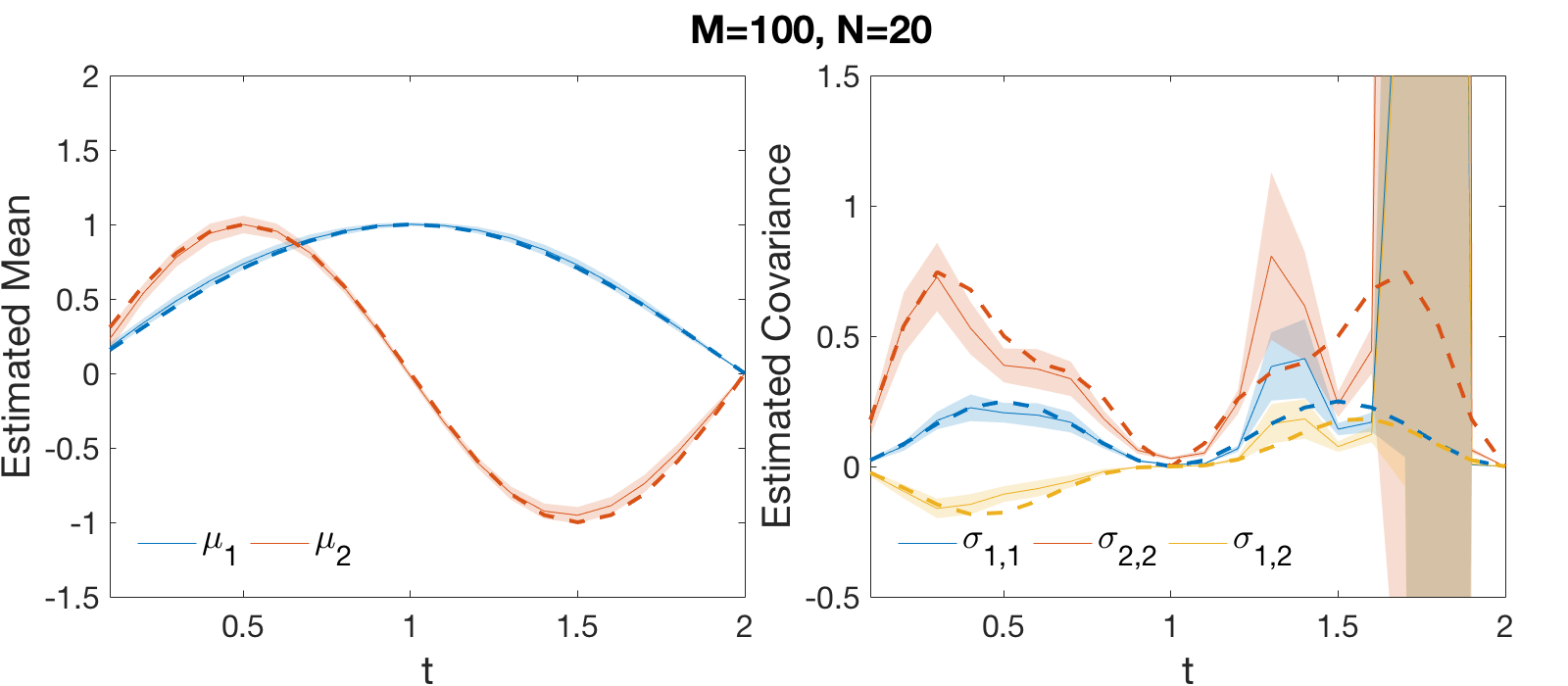}
   \includegraphics[width=.495\textwidth,height=.225\textwidth]{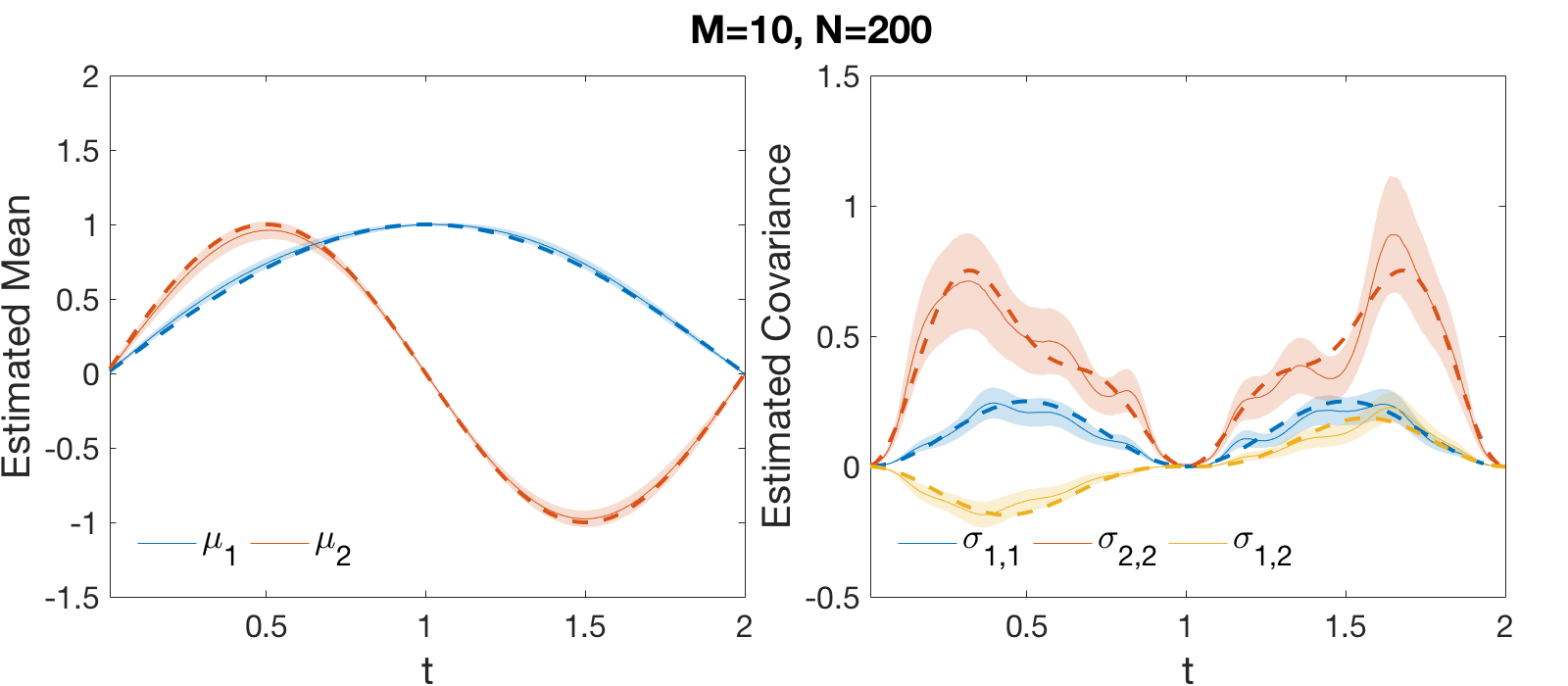}
   \includegraphics[width=.495\textwidth,height=.225\textwidth]{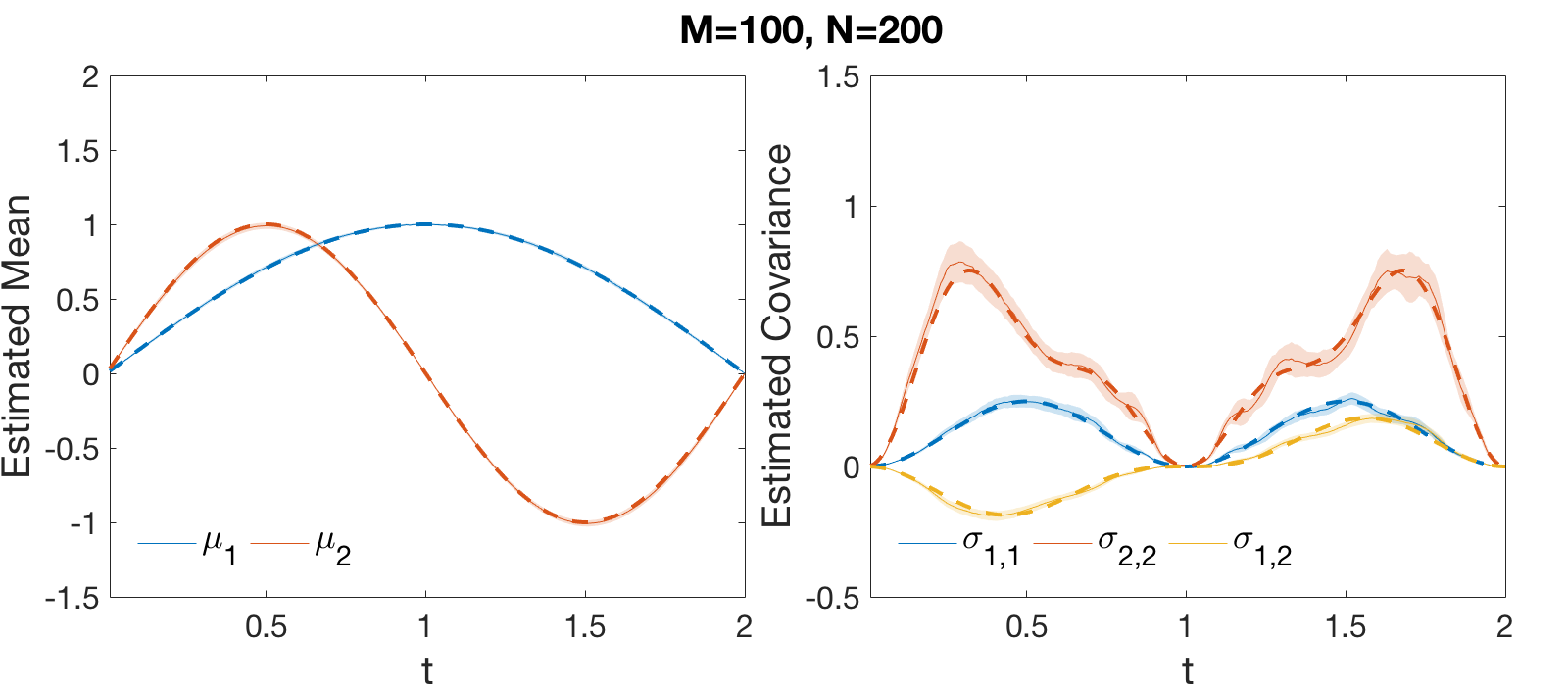}
   \caption{Estimation of the underlying mean functions $\mu_t$ (left in each of 4 subpannels) and covariance functions $\Sigma_t$ (right in each of 4 subpannels) of 2-dimensional periodic processes.
   		$M$ is the number of trials, and $N$ is the number of discretization points.
   		Dashed lines are true values, solid lines are estimates and shaded regions are $95\%$ credible bands.}
   \label{fig:periodic_postcntrt}
\end{figure}

\subsubsection{Posterior Contraction}\label{sec:pc_sim}
Posterior contraction describes the phenomenon that the posterior concentrates on smaller and smaller neighborhood of the true parameter (function) given more and more data \citep{vanderVaart08a}.
We investigate such phenomena in both mean functions and covariance functions in our model \eqref{eq:strt_dyn_model} using the following settings
$i) M=10, N=20; \; ii) M=100, N=20; \; iii) M=10, N=200; \; iv) M=100, N=200$.

To fit the data using the model \eqref{eq:strt_dyn_model}, we set $s=2$, $a=(1,1,1)$, $b=(0.1,10^{-3},0.2)$, $m=(0,0,0)$ for all settings, $V=(1,0.5,1)$ for $N=20$ and $V=(1,1,0.3)$ for $N=200$.
We also add an additional nugget of $10^{-5}I_n$ to all the covariance kernel of GPs to ensure non-degeneracy.
Following the procedure in Section \ref{sec:sample_procedure}, we run MCMC for $1.5\times 10^5$ iterations, burn in the first $5\times 10^4$ and subsample 1 for every 10. 
Based on the resulting $10^4$ posterior samples, we estimate the underlying mean functions and covariance functions and plot the estimates in Figure \ref{fig:periodic_postcntrt}.

Note in Figure \ref{fig:periodic_postcntrt}, 
both $M$ and $N$ have effect on the amount of data information thereafter on the posterior contraction but the contraction rate may depend on them differently.
Both mean and covariance functions have narrower credible bands for more discretization points $N$ (comparing $N=20$ in the first row with $N=200$ for the second row).
On the other hand, both posteriors contract further with more trials $M$ (comparing $M=10$ in the first column agains $M=100$ for the second column).
In general the posterior of mean function contracts to the truth faster than the posterior of covariance function.
With $M=100$ trials and $N=200$ discretization points, both mean and covariance functions are almost recovered by the model \eqref{eq:strt_dyn_model}.

\begin{figure}[t] 
   \centering
   \includegraphics[width=1\textwidth,height=.225\textwidth]{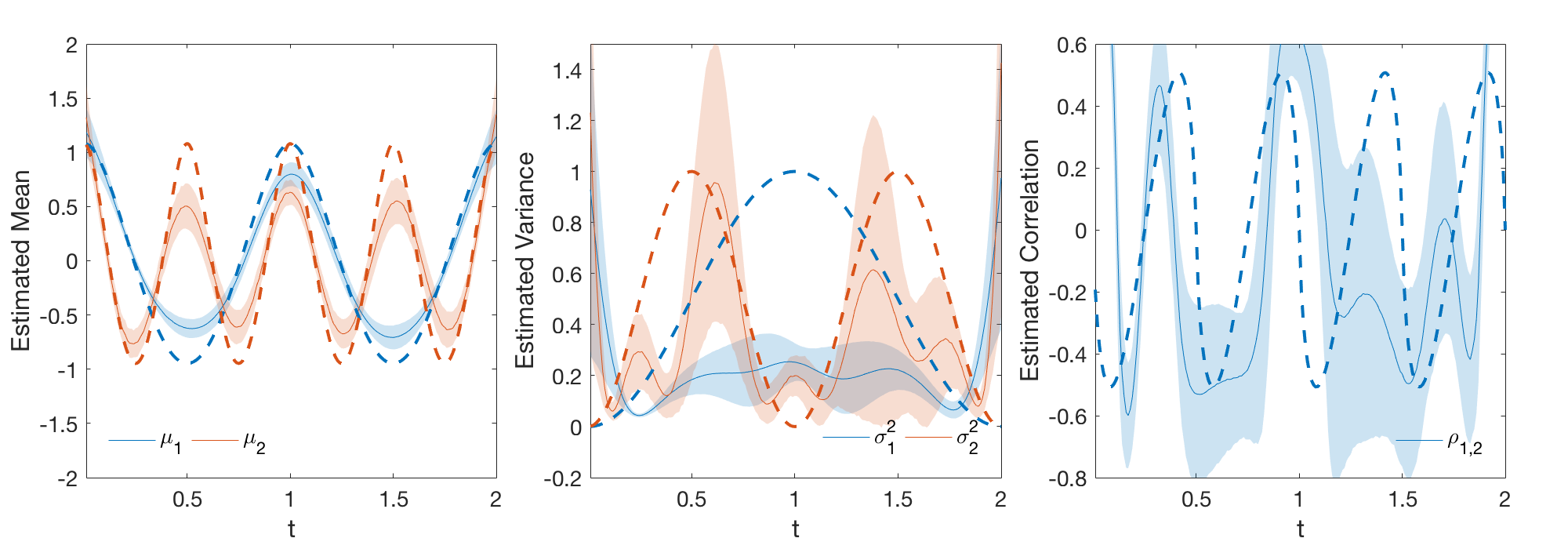}
   \includegraphics[width=1\textwidth,height=.225\textwidth]{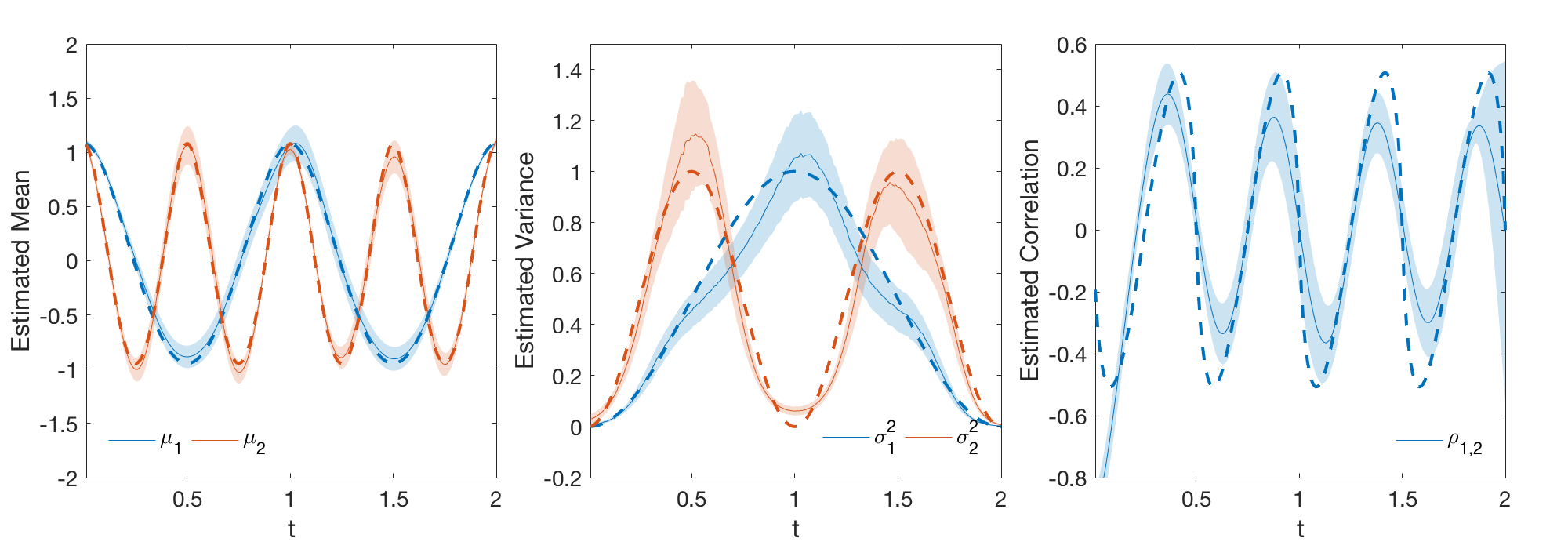}
   \caption{Estimation of the underlying mean functions $\mu_t$ (left column), variance functions $\sigma_t$ (middle column) and correlation function $\rho_t$ (right column) of 2-dimensional periodic processes,
   		using latent factor process model (upper row) and our flexible model (lower row), based on $M=10$ trials of data over $N=200$ evenly spaced points.
   		Dashed lines are true values, solid lines are estimates and shaded regions are $95\%$ credible bands.}
   \label{fig:periodic_constrast}
\end{figure}

\subsubsection{Full Flexibility}\label{sec:flexibility}
Our method \eqref{eq:strt_dyn_model} grants full flexibility because it models mean, variance and correlation processes separately.
This is particularly useful if they behave differently.
It contrasts with latent factor based models that tie mean and covariance processes together. One of the state-of-the-art models of this type is Bayesian nonparametric covariance regression \citep{fox15}:
\begin{equation}\label{eq:lpf_bncr}
y(x) \sim \mathcal N_D(\mu(x),\Sigma(x)), \qquad \mu(x) = \Lambda(x) \psi(x), \quad \Sigma(x) = \Lambda(x) \tp{\Lambda(x)} + \Sigma_0
\end{equation}
We tweak the simulated example \eqref{eq:periodic} for $D=2$ to let mean and correlation processes have higher frequency than variance processes, as shown in the dashed lines in Figure \ref{fig:periodic_constrast}. We generate $M=10$ trials of data over $N=200$ evenly spaced points.
In this case, the true mean processes $\mu(x)$ and true covariance processes $\Sigma(x)$ behave differently but are modeled with a common loading matrix $\Lambda(x)$ in model \eqref{eq:lpf_bncr}. This imposes difficulty on \eqref{eq:lpf_bncr} to have a latent factor process $\psi(x)$ that could properly accommodate the heterogeneity in mean and covariance processes.
Figure \ref{fig:periodic_constrast} shows that 
due to this reason, latent factor based model \eqref{eq:lpf_bncr} (upper row) fails to generate satisfactory fit for all of the mean, variance and correlation processes. Our fully flexible model \eqref{eq:strt_dyn_model} (bottom row), on the contrary, successfully produces more accurate characterization for all of them.
Note that this artificial example is used to demonstrate the flexibility of our dynamic model \eqref{eq:strt_dyn_model}. For cases that are not as extreme, \eqref{eq:lpf_bncr} may performance equally well. 
See more discussion in Section \ref{sec:conclusion} and more details in 
Section 
E.2
in the supplementary file.

\begin{figure}[t] 
   \centering
   \includegraphics[width=1\textwidth,height=.3\textwidth]{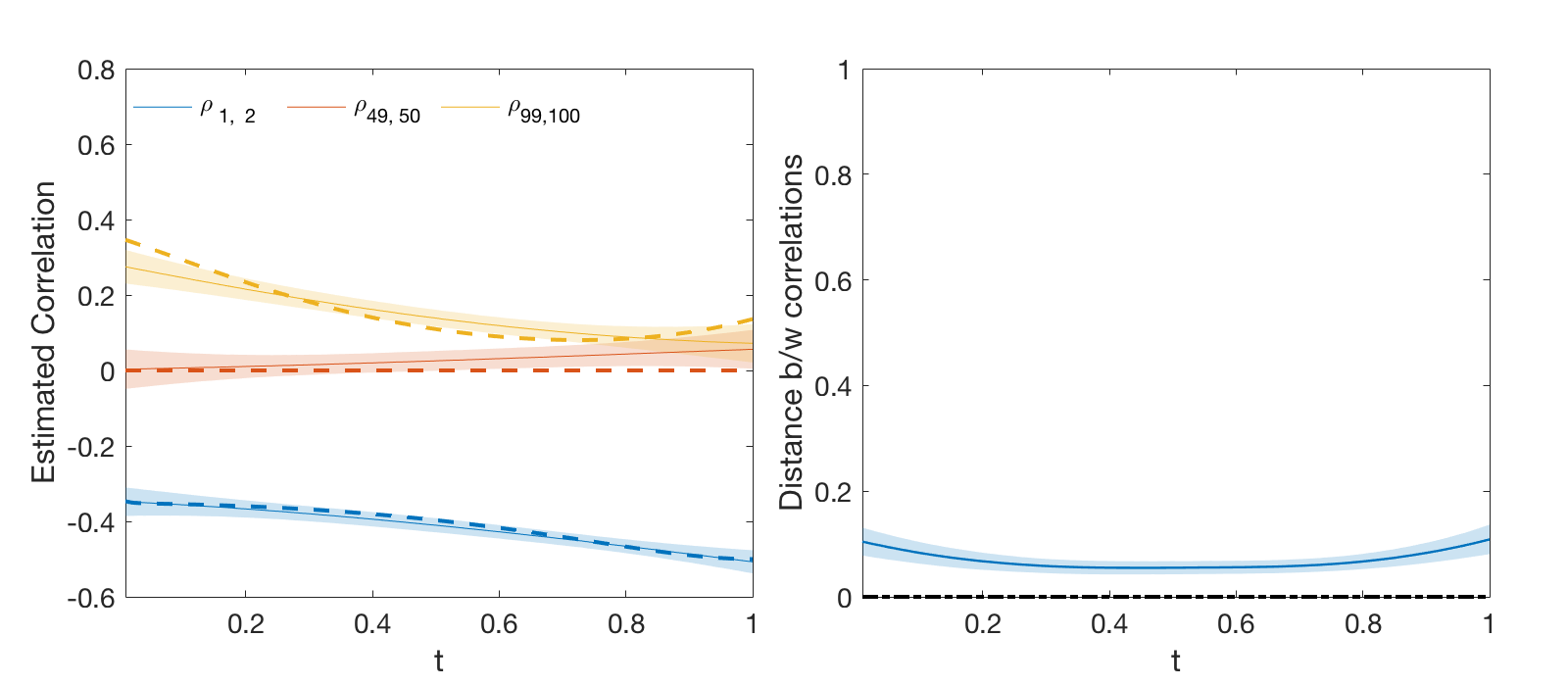}
   \caption{Posterior estimation of the underlying correlation functions $P_t$ (left) and its 2-norm distance to the truth (right) of 100-dimensional periodic processes with 2-band structure,
   		based on $M=100$ trials of data over $N=100$ discretization points.
   		Dashed lines are true values, solid lines are estimates and shaded regions are $95\%$ credible bands.}
   \label{fig:periodic_wbd_post}
\end{figure}

\subsubsection{Scalability}\label{sec:scalability}
Now we use the same simulation model \eqref{eq:periodic} for $D=100$ dimensions to test the scalability of our dynamic model \eqref{eq:strt_dyn_model}.
However instead of the full covariance, we only consider a diagonal covariance matrix plus 4 non-zero off-diagonal entries $\sigma_{1,2}$ ($\sigma_{2,1}$) and $\sigma_{99,100}$ ($\sigma_{100,99}$).
We focus on the correlation process in this example thus set $\mu_t\equiv 0$ and $\sigma_t\equiv 1$, for $t\in[0,1]$.
More specifically when generating data $\{y_t\}$ with \eqref{eq:periodic}, if $i\notin\{2,100\}$ we set $i$-th rows $L_i=S_i=e_i$ with $e_i$ being the $i$-th row of identity matrix.

To apply our dynamical model \eqref{eq:strt_dyn_model} in this setting, we let ${\bf L}_t$ have `$w$-band' structure with $w=2$ at each time $t$.
Setting $s=2$, $a=1$, $b=0.1$, $m=0$ and $V=10^{-3}$, $N=100$ and $M=100$, we repeat the MCMC runs for $7.5\times 10^4$ iterations, burn in the first $2.5\times 10^4$ and subsample 1 for every 10 to obtain $5\times10^3$ posterior samples in the end. Based on those samples, we estimate the underlying correlation functions and only plot $\rho_{1,2}$, $\rho_{49,50}$ and $\rho_{99,100}$ in Figure \ref{fig:periodic_wbd_post}.
With the `$w$-band' structure, we have less entries in the covariance matrix and focus on the `in-group' correlation. 
Our dynamical model \eqref{eq:strt_dyn_model} is sensitive enough to discern the informative non-zero components from the non-informative ones in these correlation functions. 
Unit-vector GP priors provide flexibility for the model to capture the changing pattern of informative correlations.
The left panel of Figure \ref{fig:periodic_wbd_post} shows that the model \eqref{eq:strt_dyn_model} correctly identify the non-zero components $\rho_{1,2}$ and $\rho_{99,100}$ and characterize their evolution.
The right panel shows that the 2-norm distance between the estimated and true correlation matrices, $\Vert \widehat P(t)-P(t)\Vert_2$, is small, indicating that our dynamic model \eqref{eq:strt_dyn_model} performs well with higher dimension in estimating complex dependence structure among multiple stochastic processes.


\section{Analysis of Local Field Potential Activity}\label{sec:lfp}
Now we use the proposed model \eqref{eq:strt_dyn_model} to analyze a local field potential (LFP) activity dataset. The goal of this analysis is
to elucidate how memory encoding, retrieval and decision-making arise from functional interactions among brain regions, by modeling how their dynamic connectivity varies during performance of complex memory tasks. Here we focus on LFP activity data recorded from 24 electrodes spanning the dorsal CA1 subregion of the hippocampus as rats performed a  sequence memory task \citep{allen14,allen16,ng17,holbrook17}. The task involves repeated presentations of a sequence of odors (e.g., ABCDE) at a single port and requires rats to correctly determine whether each odor is presented `in sequence' (InSeq; e.g., ABCDE; by holding their nosepoke response until the signal at 1.2s) or `out of sequence' (OutSeq; e.g., AB\underline{D}DE; by withdrawing their nose before the signal). In previous work using the same dataset, \cite{holbrook16} used a direct MCMC algorithm 
to study the spectral density matrix of LFP from 4 selected channels. However, they did not examine how their correlations varied across time and recording site. These limitations are addressed in this paper. 

\begin{figure}[t] 
   \centering
   \includegraphics[width=1\textwidth,height=.45\textwidth]{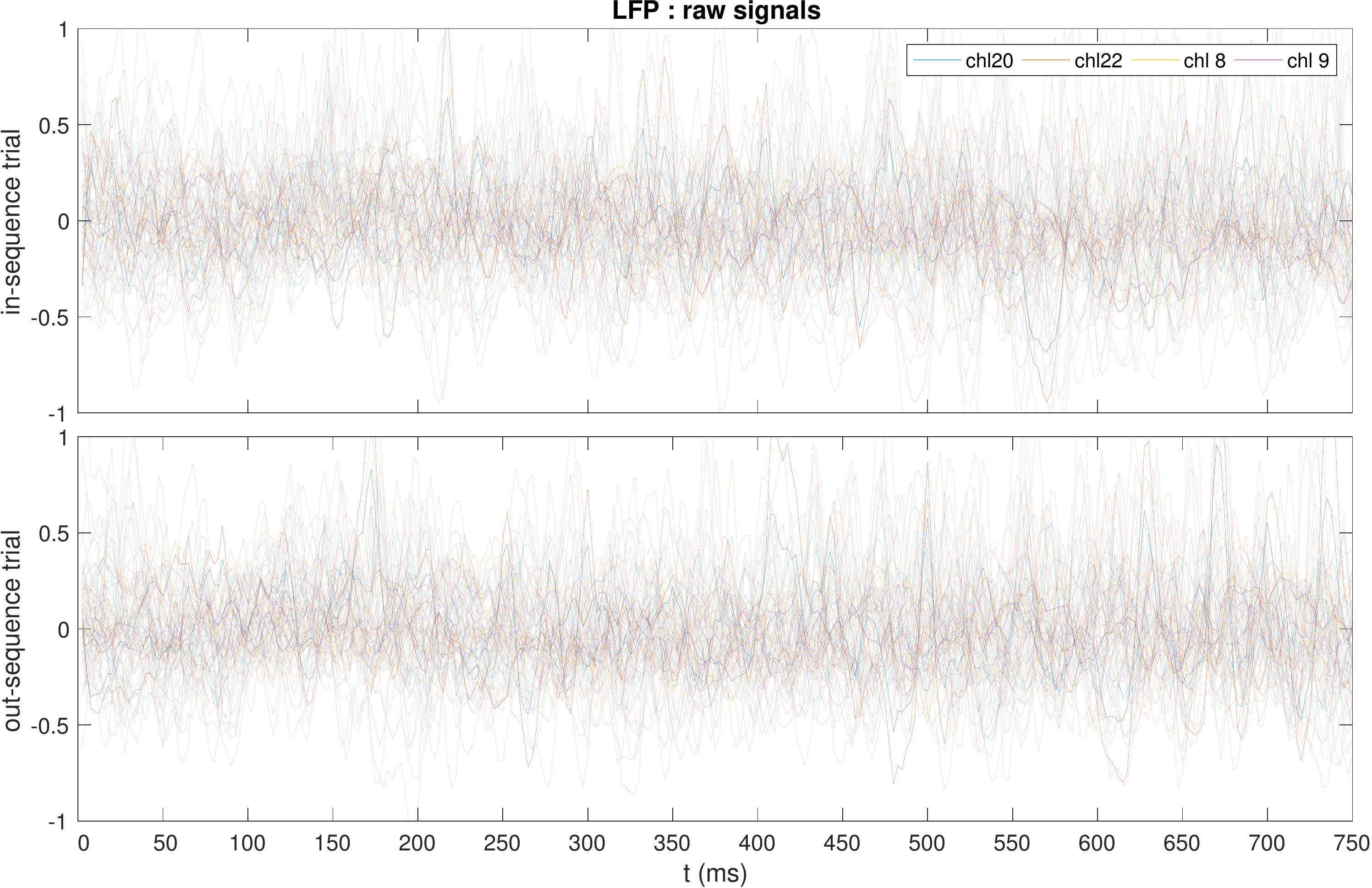} 
   \caption{LFP signals on ``in sequence" and ``out of sequence" trials. It is difficult to identify differences between the two conditions based on a mere visual inspection of the LPFs.}
   \label{fig:LFP_raw_data}
\end{figure}

We focus our analyses on the time window from 0ms to 750ms (with 0 corresponding to when the rat's nose enters the odor port). Critically, this includes a time period during which the behavior of the animal is held constant (0-500ms) so differences in LFP reflect the cognitive processes associated with task performance, and, to serve as a comparison, a time period near 750ms during which the behavioral state of the animal is known to be different (i.e., by 750ms the animal has already withdrawn from the port on the majority of OutSeq trials, but is still in the port on InSeq trials). We also focus our analyses on two sets of adjacent electrodes (electrodes 20 and 22, and electrodes 8 and 9), which allows for comparisons between probes that are near each other ($<$1mm; i.e., 20:22 and 8:9) or more distant from each other ($>$2mm; i.e., 20:8, 20:9, 22:8, and 22:9). Figure \ref{fig:LFP_raw_data} shows $M = 20$ trials of these LFP signals from $D = 4$ channels under both InSeq and OutSeq conditions. 
Our main objective is to quantify how correlations among these LFP channels varied across trial types (InSeq vs OutSeq) and over time (within the first 750ms of trials). To do so, we discretize the time window of $0.75$ seconds into $N = 300$ equally-spaced small intervals. Under each experiment condition (InSeq or OutSeq), we treat all the signals as a 4 dimensional time series and fit them using our proposed dynamic correlation model \eqref{eq:strt_dyn_model} in order to discover the evolution of their relationship. Note that we model the mean, variance, and correlation processes separately but only report findings about the evolution of correlation among those brain signals.

\begin{figure}[t] 
   \centering
   \includegraphics[width=1\textwidth,height=.4\textwidth]{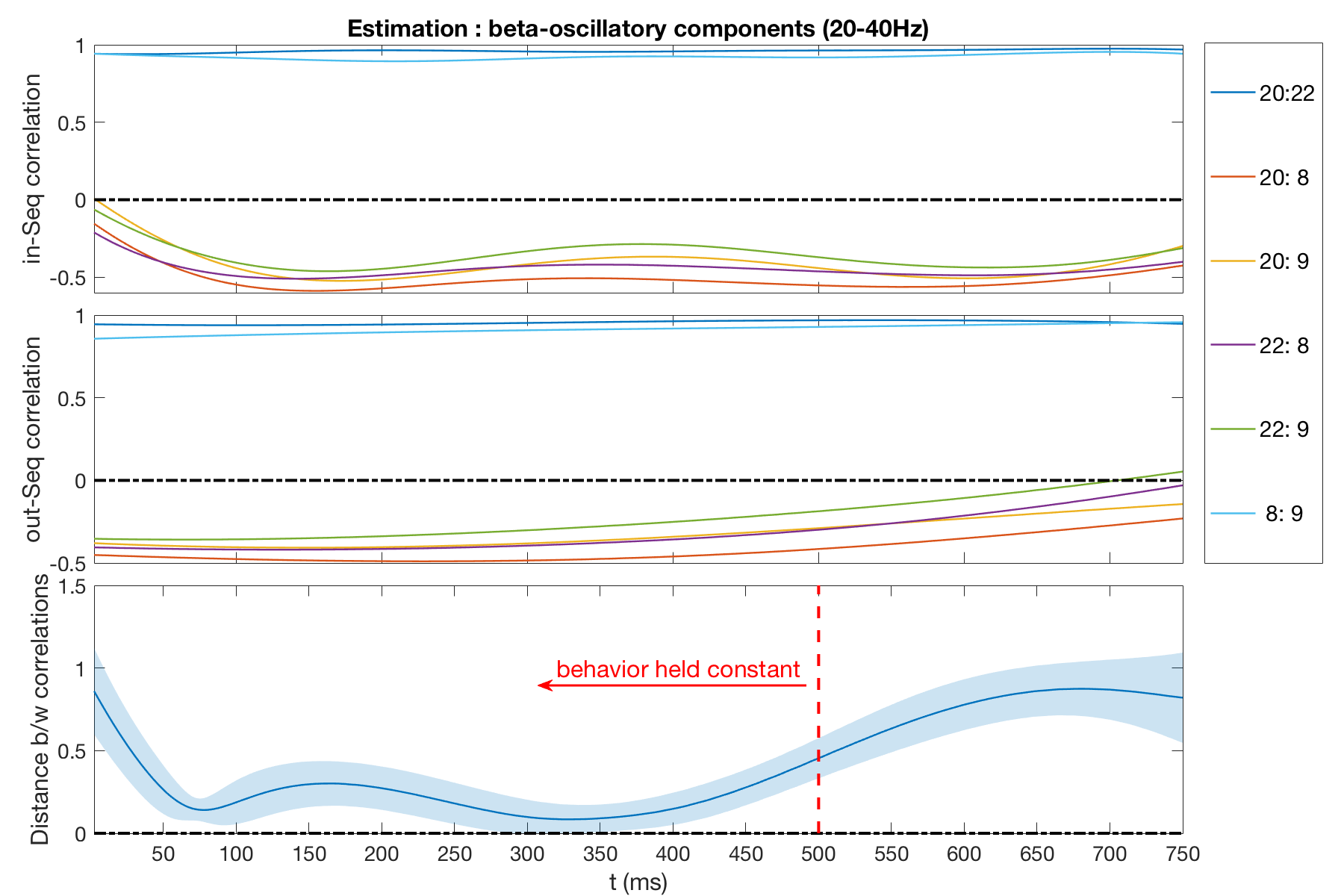}
   \caption{Estimated correlation processes of LFPs (beta) under in-sequence condition (top), out-of-sequence condition (middle) and the (Frobenius) distance between two correlation matrices (bottom).}
   \label{fig:LFP_est_beta}
\end{figure}

We set $s=2$, $a=(1,1,1)$, $b=(1,0.1,0.2)$, $m=(0,0,0)$, $V=(1,1.2,2)$; and the general results are not very sensitive to the choice of these fine-tuning parameters. We also scale the discretized time points into $(0,1]$ and add an additional nugget of $10^{-5}I_n$ to the covariance kernel of GPs. We follow the same procedure in Section \ref{sec:sample_procedure} to collect $7.5\times 10^4$ samples, burn in the first $2.5\times 10^4$ and subsample 1 for every 10. The resulting $10^4$ samples yield estimates of correlation processes as shown in Figure \ref{fig:LFP_est_beta} for beta-filtered traces (20-40Hz) but similar patterns were also observed for theta-filtered traces (4-12Hz; see the supplement). The bottom panel of Figure \ref{fig:LFP_est_beta} shows the dissimilarity between correlation processes under different conditions measured by the Frobenius norm of their difference.

Our approach revealed many important patterns in the data. First, it showed that electrodes near each other (20:22 and 8:9) displayed remarkably high correlations in their LFP activity on InSeq and OutSeq trials, whereas correlations were considerably lower among more distant electrodes (20:8, 20:9, 22:8, and 22:9). Second, it revealed that the correlations between InSeq and OutSeq matrices evolved during the presentation of individual trials. 
These results are consistent with other analyses on learning \citep[see, e.g.,][]{fiecas16}.
As expected, InSeq and OutSeq activity was very similar at the beginning of the time window (e.g., before 350ms), which is before the animal has any information about the InSeq or OutSeq status of the presented odor, but maximally different at the end of the time window, which is after it has made its response on OutSeq trials. Most important, however, is the discovery of InSeq vs OutSeq differences before 500ms, which reveal changes in neural activity associated with the complex cognitive process of identifying if events occurred in their expected order. These findings highlight the sensitivity of our novel approach, as such differences have not been detected with traditional analyses.
Interested readers can find more results about all the 12 channels in Section 
E.3
of the supplementary file.

\section{Conclusion}\label{sec:conclusion}
In this paper, we propose a novel Bayesian framework that grants full flexibility in modeling covariance and correlation matrices.
It extends the separation strategy proposed by \cite{barnard00} and uses the Cholesky decomposition to maintain the positive definiteness of the correlation matrix.
By defining distributions on spheres, a large class of flexible priors can be induced for covariance matrix that go beyond the commonly used but restrictive inverse-Wishart distribution.
Furthermore, the structured models we propose maintain the interpretability of covariance in terms of variance and correlation.
Adaptive $\Delta$-Spherical HMC is introduced to handle the intractability of the resulting posterior.
Furthermore, we extend this structured scheme to dynamical models to capture complex dependence among multiple stochastic processes,
and demonstrate the effectiveness and efficiency in Bayesian modeling covariance and correlation matrices using a normal-inverse-Wishart problem, a simulated periodic process, and an analysis of LFP data.
In addition, we provide both theoretic characterization and empirical investigation of posterior contraction for dynamically covariance modeling, which to our best knowledge, is a first attempt.

In this work, we consider the \emph{marginal (pairwise)} dependence among multiple stochastic processes.
The priors for correlation matrix specified through the sphere-product representation are in general dependent among component variables. For example, the method we use to induce uncorrelated prior between $y_i$ and $y_j$ ($i<j$) by setting $l_{jk} \approx 0$ for $k\leq i$ has a direct consequence that $\Cor(y_{i'}, y_j)\approx 0$ for $i'\leq i$.
In another word, more informative priors (part of the components are correlated) may require careful ordering in $\{y_i\}$.
To avoid this issue, one might consider the inverse of covariance (precision) matrices instead. This leads to modeling the \emph{conditional} dependence, or \emph{Markov network} \citep{dempster72,friedman08}.
Our proposed methodology applies directly to (dynamic) precision matrices/processes, which will be our future direction.

To further scale our method to problems of greater dimensionality in future,
one could explore the low-rank structure of covariance and correlation matrices,
e.g. by adopting the similar factorization as in \citep{fox15} and assuming $\tp\vech({\bf L}_t)\in (\mathcal S^k)^D$ for some $k\ll D$,
or impose some sparse structure on the precision matrices.


We have proved that the posterior of covariance function contracts at a rate given by the general form of concentration function \citep{vanderVaart08a}. 
Empirical evidence (Section \ref{sec:pc_sim}) shows that the posterior of covariance contracts slower than that of mean. More theoretical works is needed to compare their contraction rates.
Also, future research could involve investigating posterior contraction in covariance regression with respect to the optimal rates under different GP priors.

While our research has generated interesting new findings regarding brain signals during memory tasks, one limitation of our current analysis on LFP data is that it is conducted on a single rat. The proposed model can be generalized to account for variation among rats.
In the future, we will apply this sensitive approach to other datasets, including simultaneous LFP recordings from multiple brain regions in rats as well as BOLD fMRI data collected from human subjects performing the same task.


%

\section*{Acknowledgement}
SL is supported by ONR N00014-17-1-2079. 
AH is supported by NIH grant T32 AG000096.
GAE is supported by NIDCD T32 DC010775.
NJF is supported by NSF awards IOS-1150292 and BCS-1439267 and Whitehall Foundation award 2010-05-84.
HO is supported by NSF DMS 1509023 and NSF SES 1461534.
BS is supported by NSF DMS 1622490 and NIH R01 AI107034.
We thank Yulong Lu at Duke for discussions on posterior contraction of covariance regression with Gaussian processes.

\newpage
\bibliography{references}
\bibliographystyle{apa}


\newpage
\begin{center}
{\large\bf SUPPLEMENTARY MATERIAL}
\end{center}

\appendix
\numberwithin{equation}{section}
\numberwithin{figure}{section}
\numberwithin{lem}{section}
\numberwithin{thm}{section}
\numberwithin{dfn}{section}
\numberwithin{rk}{section}

\section{Connection to Known Priors}\label{apx:conn2knownpriors}
The following lemma is essential in proving that our proposed methods \eqref{eq:strt_cov} \eqref{eq:strt_corr} generalize existing methods in specifying priors,
including the inverse-Wishart distribution, and two uniform distributions \citep{barnard00} as well.
\begin{lem}\label{lem:chol_jacobians}
Let $\vect\Sigma = {\bf U} \tp{\bf U}$ be the reversed Cholesky decomposition of $\vect\Sigma$. The Jocobian of the transformation $\vect\Sigma\mapsto {\bf U}$ is
\begin{equation*}
\left| \frac{d_h \vect\Sigma}{d_h \tp{\bf U}} \right| := \left| \frac{\pa \vech \vect\Sigma}{\pa \vech \tp{\bf U}} \right| = 2^D \prod_{i=1}^D |u_{ii}|^i
\end{equation*}
Let $\vect\Rho={\bf L}\tp{\bf L}$ be the Cholesky decomposition of $\vect\Rho$. The Jacobian of the transformation ${\bf L}\mapsto \vect\Rho$ is
\begin{equation*}
\left| \frac{d_h {\bf L}}{d_h \vect\Rho} \right| := \left| \frac{\pa \vech {\bf L}}{\pa \vech \vect\Rho} \right|  = 2^{-D} \prod_{i=1}^D |l_{ii}|^{i-(D+1)}
\end{equation*}
\end{lem}
\begin{proof}
Note we have
\begin{equation*}
d \vect\Sigma = d {\bf U}\tp{\bf U} + {\bf U} d\tp{\bf U}
\end{equation*}
Taking $\VEC$ on both sides and applying its property
\begin{equation*}
d \VEC\vect\Sigma = ({\bf U} \otimes {\bf I}) d \VEC{\bf U} + ({\bf I} \otimes {\bf U}) d\VEC\tp{\bf U}
\end{equation*}
Applying the elimination ${\bf L}_D$ on both sides
\begin{equation*}
\begin{aligned}
d \vech\vect\Sigma &= {\bf L}_D [ ({\bf U} \otimes {\bf I}) {\bf K}_D d \VEC\tp{\bf U} + ({\bf I} \otimes {\bf U}) d\VEC\tp{\bf U} ]
= {\bf L}_D ({\bf K}_D + {\bf I}) ({\bf I} \otimes {\bf U}) d\VEC \tp{\bf U}\\
&= 2 {\bf L}_D {\bf N}_D ({\bf I} \otimes {\bf U}) \tp{\bf L}_D d\vech \tp{\bf U} = 2 {\bf L}_D {\bf N}_D \tp{\bf L}_D \tp{\bf D}_D ({\bf I} \otimes {\bf U}) \tp{\bf L}_D d\vech \tp{\bf U}
\end{aligned}
\end{equation*}
where ${\bf K}_D$ is the \emph{commutation matrix} such that ${\bf K}_D \VEC{\bf A}=\VEC\tp{\bf A}$ for matrix ${\bf A}_{D\times D}$,
${\bf N}_D:=({\bf K}_D+{\bf I})/2$, and ${\bf D}_D$ is the \emph{duplication matrix} which is regarded as the inverse of the elimination matrix ${\bf L}_D$.
The last equation is by ${\bf D}_D {\bf L}_D {\bf N}_D = {\bf N}_D = \tp{\bf N}_D$ \citep[Lemma2.1 and Lemma3.5 in][]{magnus80}.
Thus according to \cite[Lemma3.4 and Lemma4.1 in][]{magnus80} we have
\begin{equation*}
\begin{aligned}
\left| \frac{d_h \vect\Sigma}{d_h \tp{\bf U}} \right| &= \left| \frac{\pa \vech \vect\Sigma}{\pa \vech \tp{\bf U}} \right| 
= |2 {\bf L}_D {\bf N}_D \tp{\bf L}_D \tp{\bf D}_D ({\bf I} \otimes {\bf U}) \tp{\bf L}_D|\\
&= 2^{D(D+1)/2} |{\bf L}_D {\bf N}_D \tp{\bf L}_D| |{\bf L}_D ({\bf I} \otimes \tp{\bf U}) {\bf D}_D| = 2^D \prod_{i=1}^D |u_{ii}^i|
\end{aligned}
\end{equation*}
%
By similar argument, we have
\begin{equation*}
\begin{aligned}
\left| \frac{d_h \vect\Rho}{d_h {\bf L}} \right| &= \left| \frac{\pa \vech \vect\Rho}{\pa \vech {\bf L}} \right| 
= |2 {\bf L}_D {\bf N}_D \tp{\bf L}_D \tp{\bf D}_D ({\bf L} \otimes {\bf I}) \tp{\bf L}_D|\\
&= 2^{D(D+1)/2} |{\bf L}_D {\bf N}_D \tp{\bf L}_D| |{\bf L}_D (\tp{\bf L} \otimes {\bf I}) {\bf D}_D| = 2^D \prod_{i=1}^D |l_{ii}|^{D+1-i} \,.
\end{aligned}
\end{equation*}
Thus it completes the proof.
\end{proof}


\begin{proof}[Proof of Proposition \ref{prop:conn_iWishart}]
We know that the density of $\vect\Sigma\sim \mathcal W^{-1}_D(\vect\Psi,\nu)$ is
\begin{equation*}
p_{\mathcal W^{-1}}(\vect\Sigma\, ;  \vect\Psi,\nu) = \frac{|\vect\Psi|^{\nu/2}}{2^{D\nu/2} \Gamma_D(\nu/2)} |\vect\Sigma|^{-(\nu+D+1)/2} \exp\left(-\half\tr (\vect\Psi {\vect\Sigma}^{-1}) \right)
\end{equation*}
By Lemma \ref{lem:chol_jacobians} we have
\begin{equation*}
\begin{aligned}
p({\bf U}) &= p(\vect\Sigma) \left| \frac{d_h \vect\Sigma}{d_h \tp{\bf U}} \right| = 2^D p_{\mathcal W^{-1}}({\bf U} \tp{\bf U}\, ; \vect\Psi,\nu) \prod_{i=1}^D |u_{ii}^i| \\
&= \frac{|\vect\Psi|^{\nu/2}}{2^{D(\nu-2)/2} \Gamma_D(\nu/2)} |{\bf U}|^{-(\nu+D+1)} \prod_{i=1}^D u_{ii}^i \exp\left(-\half\tr (\vect\Psi {\bf U}^{-\mathsf T} {{\bf U}^{-1}}) \right) \,.
\end{aligned}
\end{equation*}
Then the proof is completed.
\end{proof}


\begin{proof}[Proof of Theorem \ref{thm:conn_unif}]
To prove the first result, we use Lemma \ref{lem:chol_jacobians}
\begin{equation*}
p(\vect\Rho) = p({\bf L}) \left| \frac{d_h {\bf L}}{d_h \vect\Rho} \right| \propto \prod_{i=2}^D |{\bf l}_i|^{2\vect\alpha_i-1} \prod_{i=1}^D |l_{ii}|^{i-(D+1)}
= \prod_{i=1}^D |l_{ii}|^{(i-3)(D+1)}
\end{equation*}
On the other hand, from Equation (8) in \cite{barnard00}, we have the density of marginally uniform distribution:
\begin{equation*}
p(\vect\Rho) \propto |\vect\Rho|^{\frac{D(D-1)}{2}-1} (\prod_{i} \vect\Rho_{ii})^{-\frac{D+1}{2}}
= (\prod_{j=1}^D l_{jj}^2)^{\frac{D(D-1)}{2}-1} (\prod_{j=1}^D \prod_{i=1}^j l_{ii}^2)^{-\frac{D+1}{2}}
= \prod_{j=1}^D |l_{jj}|^{(j-3)(D+1)}
\end{equation*}
where $\vect\Rho_{ii}$ is the $i$-th principal minor of $\vect\Rho$.
Similarly by Lemma \ref{lem:chol_jacobians} we can prove the second result
\begin{equation*}
p(\vect\Rho) = p({\bf L}) \left| \frac{d_h {\bf L}}{d_h \vect\Rho} \right| \propto \prod_{i=2}^D |{\bf l}_i|^{2\vect\alpha_i-1} \prod_{i=1}^D |l_{ii}|^{i-(D+1)}
\propto 1 \,.
\end{equation*}
Therefore we have finished the proof.
\end{proof}


\section{Posterior Contraction}\label{apx:post_contr}
For the Gaussian likelihood models $p_i \sim \mathcal N(\vect\mu_i(t), \vect\Sigma_i(t))$ for $i=0,1$, 
we first bound the Hellinger distance, Kullback-Leibler distance and variance distance $V(p_0,p_1)= \E_0 (\log (p_0/p_1))^2$ with the uniform norm in the following lemma.
Notation $\lesssim$ means ``smaller than or equal to a universal constant times".
\begin{lem}
\label{lem:hKVbd}
For any bounded measurable functions 
$\vect\Sigma_i : \mathcal X\rightarrow \mathbb R^{D^2}$ with Cholesky decompositions $\vect\Sigma_i={\bf L}_i\tp{{\bf L}_i}$, 
assume $\min_{1\leq j\leq D} \inf_{t\in \mathcal X} | l_{i,jj}(t)| \geq c_0>0$, $i=0,1$.
Then we have
\begin{itemize}
\item $h(p_0, p_1) \lesssim \Vert {\bf L}_0 - {\bf L}_1\Vert_\infty^\half$
\item $K(p_0, p_1) \lesssim \Vert {\bf L}_0 - {\bf L}_1\Vert_\infty$
\item $V(p_0, p_1) \lesssim \Vert {\bf L}_0 - {\bf L}_1\Vert_\infty^2$
\end{itemize}
\end{lem}
\begin{proof}
First we calculate 
\begin{equation}\label{eq:logpdf_diff}
\begin{aligned}
\log p_0 -\log p_1 &= \half \left\{ \log \frac{|\vect\Sigma_1|}{|\vect\Sigma_0|} + \tp{({\bf y}-\vect\mu_*)} \vect\Sigma_*^{-1} ({\bf y}-\vect\mu_*) + ** \right\} \\
\vect\Sigma_*^{-1} &= \vect\Sigma_1^{-1} - \vect\Sigma_0^{-1}, \quad \vect\mu_* = \vect\Sigma_* (\vect\Sigma_1^{-1}\vect\mu_1 - \vect\Sigma_0^{-1}\vect\mu_0) \\
** &= - \tp{(\vect\mu_1-\vect\mu_0)} \vect\Sigma_1^{-1} \vect\Sigma_* \vect\Sigma_0^{-1} (\vect\mu_1-\vect\mu_0)
\end{aligned}
\end{equation}
Taking expectation of \eqref{eq:logpdf_diff} with respect to $p_0$ yields the following Kullback-Leibler divergence
\begin{equation*}
K(p_0, p_1) = \half \left\{ \tr (\vect\Sigma_1^{-1}\vect\Sigma_0) + \tp{(\vect\mu_1-\vect\mu_0)} \vect\Sigma_1^{-1} (\vect\mu_1-\vect\mu_0) - D + \log \frac{|\vect\Sigma_1|}{|\vect\Sigma_0|} \right\}
\end{equation*}
Consider $\vect\mu_i\equiv 0$. By the non-negativity of K-L divergence we have for general $\vect\Sigma_i>0$,
\begin{equation}\label{eq:logdetbd}
\log \frac{|\vect\Sigma_0|}{|\vect\Sigma_1|} \leq \tr (\vect\Sigma_1^{-1}\vect\Sigma_0 - {\bf I}) 
\end{equation}
Therefore we can bound K-L divergence
\begin{equation*}
K(p_0, p_1) \leq \half \{ \tr (\vect\Sigma_1^{-1}\vect\Sigma_0 - {\bf I}) + \tr (\vect\Sigma_0^{-1}\vect\Sigma_1 - {\bf I}) \} \leq C(D, c_0) \Vert \vect\Sigma_0 - \vect\Sigma_1\Vert_\infty \lesssim \Vert {\bf L}_0 - {\bf L}_1\Vert_\infty
\end{equation*}
where we bound each term involving trace
\begin{equation}\label{eq:tracebd}
\begin{aligned}
& \tr (\vect\Sigma_1^{-1}\vect\Sigma_0 - {\bf I}) = \tr (\vect\Sigma_1^{-1}(\vect\Sigma_0 - \vect\Sigma_1)) \leq \Vert \vect\Sigma_1^{-1} \Vert_F \Vert \vect\Sigma_0 - \vect\Sigma_1\Vert_F \\
& \leq D^{3/2} \Vert \vect\Sigma_1^{-1} \Vert_2 \Vert \vect\Sigma_0 - \vect\Sigma_1\Vert_\infty 
\lesssim c_0^{-2D} (\Vert {\bf L}_0 \Vert_\infty+\Vert {\bf L}_1 \Vert_\infty) \Vert {\bf L}_0 - {\bf L}_1\Vert_\infty
\end{aligned}
\end{equation}
Note the last inequality holds because
\begin{align*}
c_0^D &\leq \Pi_{j=1}^D | l_{jj} | = \left[ \Pi_{j=1}^D \lambda_j(\vect\Sigma)\right]^\half \leq \lambda_{\min}^\half \left[\frac{\tr (\vect\Sigma)^{D-1}}{D-1}\right]^\half
\leq \lambda_{\min}^\half \left[\frac{\Vert {\bf L} \Vert_F^{2(D-1)}}{D-1}\right]^\half \\
&\leq \lambda_{\min}^\half \left[\frac{D^{2(D-1)}}{D-1}\right]^\half \Vert {\bf L} \Vert_\infty^{D-1}
= \Vert \vect\Sigma_1^{-1} \Vert_2^{-\half} \left[\frac{D^{2(D-1)}}{D-1}\right]^\half \Vert {\bf L} \Vert_\infty^{D-1}
\end{align*}
and we can write
\begin{equation*}
\vect\Sigma_0 - \vect\Sigma_1 = \half \left[ ({\bf L}_0 - {\bf L}_1) \tp{({\bf L}_0 + {\bf L}_1)} + ({\bf L}_0 + {\bf L}_1)\tp{({\bf L}_0 - {\bf L}_1)}\right]
\end{equation*}

Now take expectation of squared \eqref{eq:logpdf_diff} with respect with $p_0$ to get the following variance distance
\begin{equation*}
V(p_0,p_1) = \half \tr ((\vect\Sigma_1^{-1}\vect\Sigma_0 - {\bf I})^2) + \tp{(\vect\mu_1-\vect\mu_0)} \vect\Sigma_1^{-1} \vect\Sigma_0 \vect\Sigma_1^{-1} (\vect\mu_1-\vect\mu_0) + K^2(p_0, p_1)
\end{equation*}
Consider $\vect\mu_i\equiv 0$ and we can bound the variance distance by similar argument as \eqref{eq:tracebd}
\begin{equation*}
V(p_0,p_1) \leq \half \Vert \vect\Sigma_1^{-1}(\vect\Sigma_0 - \vect\Sigma_1) \Vert_F^2 + K^2(p_0, p_1) \leq 
C \Vert \vect\Sigma_1^{-1} \Vert_F^2 \Vert \vect\Sigma_0 - \vect\Sigma_1\Vert_F^2 + K^2(p_0, p_1) \lesssim \Vert {\bf L}_0 - {\bf L}_1\Vert_\infty^2
\end{equation*}
where we use the fact $\Vert {\bf A}{\bf B}\Vert_F \leq \Vert {\bf A}\Vert_F \Vert {\bf B}\Vert_F$.

Lastly, the squared Hellinger distance for multivariate Gaussians can be calculated
\begin{equation*}
h^2(p_0, p_1) = 1 - \frac{|\vect\Sigma_0 \vect\Sigma_1|^{1/4}}{\left|\frac{\vect\Sigma_0+\vect\Sigma_1}{2}\right|^{1/2}} \exp\left\{-\frac18 \tp{(\vect\mu_0-\vect\mu_1)}\left(\frac{\vect\Sigma_0+\vect\Sigma_1}{2}\right)^{-1}(\vect\mu_0-\vect\mu_1)\right\}
\end{equation*}
Consider $\vect\mu_i\equiv 0$. Notice that $1-x\leq -\log x$, and by \eqref{eq:logdetbd} we can bound the squared Hellinger distance using similar argument in \eqref{eq:tracebd}
\begin{equation*}
\begin{aligned}
h^2(p_0, p_1) &\leq \log \frac{\left|\frac{\vect\Sigma_0+\vect\Sigma_1}{2}\right|^{1/2}}{|\vect\Sigma_0 \vect\Sigma_1|^{1/4}}
\leq \half \tr ({\bf L}_0^{-\mathsf T} {\bf L}_1^{-1} (\vect\Sigma_0+\vect\Sigma_1)/2 - {\bf I} ) \\
& = \frac14 \{ \tr ({\bf L}_1^{-1}{\bf L}_0 - {\bf I}) + \tr({\bf L}_0^{-\mathsf T} \tp{\bf L}_1 - {\bf I}) \}
\lesssim c_0^{-D} \Vert {\bf L}_0 - {\bf L}_1\Vert_\infty
\end{aligned}
\end{equation*}
where we use $\Vert{\bf L}^{-1}\Vert_2=\lambda_{\min}^{-1}({\bf L}) \leq c_0^{-D}\left[\frac{D\Vert{\bf L}\Vert_\infty}{D-1}\right]^{D-1}$.
\end{proof}

Define the following \emph{coordinate concentration} function as in \eqref{eq:concent_1}
\begin{equation}\label{eq:concent_coord}
\phi_{l_{0,ij}}(\eps) = \inf_{h_{ij}\in\mathbb H_{ij} : \,\Vert h_{ij}-l_{0,ij}\Vert_{ij} < \eps} \Vert h\Vert_{\mathbb H_{ij}}^2 - \log \Pi(l_{ij} :\, \Vert l_{ij} \Vert_{ij} < \eps)
\end{equation}
It is easy to see that $\phi_{l_{0,ij}}(\eps)\leq \phi_{{\bf L}_0}(\eps)$ for $\forall 1\leq j \leq i\leq D$.
Let $\Vert \cdot\Vert:=\max_{i,j}\Vert\cdot\Vert_{ij}$.
For $\eps>0$, let $N(\eps, B, d)$ denote the minimum number of balls of radius $\eps$ that a cover $B$ in a metric space with metric $d$, which is named \emph {$\eps$-covering number} for $B$.
Now we are ready to prove the theorem of posterior concentration. 

\begin{proof}[Proof of Theorem \ref{thm:post_contr}]
We use Theorem 4 of \cite{ghosal07} and it suffices to verify three conditions (the entropy condition 3.2, the complementary asertion 3.3 and the prior mass condition 3.4) as follows:
\begin{align}
\sup_{\eps>\eps_n} \log N(\eps/36, \{{\bf L}\in \Theta_n: d_n({\bf L},{\bf L}_0)<\eps\},d_n) &\leq n\eps_n^2 \label{eqa:entropy}\\
\frac{\Pi_n(\Theta\backslash \Theta_n)}{\Pi_n(\bar B_n({\bf L}_0, \eps_n))} &= o(e^{-2n\eps_n^2}) \label{eqa:complement}\\
\frac{\Pi_n({\bf L}\in\Theta_n: \kappa\eps_n<d_{n,H}({\bf L},{\bf L}_0)<2\kappa\eps_n)}{\Pi_n(\bar B_n({\bf L}_0, \eps_n))} & \leq e^{n\eps_n^2\kappa^2/4}, \; \text{for\, large}\; \kappa \label{eqa:priormass}
\end{align}
where $\bar B_n({\bf L}_0, \eps):=\{{\bf L}\in\Theta: \frac1n\sum_{i=1}^n K_i({\bf L}_0,{\bf L})\leq \eps^2, \frac1n\sum_{i=1}^n V_i({\bf L}_0,{\bf L})\leq \eps^2\}$,
with $\Theta=L^\infty(\mathcal X)^{D(D+1)/2}$, $K_i({\bf L}_0,{\bf L})=K(P_{{\bf L}_0,i}, P_{{\bf L},i})$ and $V_i({\bf L}_0,{\bf L})=V(P_{{\bf L}_0,i}, P_{{\bf L},i})$.

Applying Theorem 2.1 of \cite{vanderVaart08a} to each Gaussian random element $l_{ij}$ in $\mathbb B_{ij}=L^\infty(\mathcal X)$ for $1\leq j\leq i\leq D$, with $l_{0,ij}\in\bar{\mathbb H}_{ij}$, we have $C>2$ and the measurable set $B_{n,ij}\subset \mathbb B_{ij}$ such that
\begin{align}
\log N(3\eps_n, B_{n,ij}, \Vert \cdot \Vert_\infty) &\leq 6 Cn\eps_n^2 \label{eqa:gp_entropy}\\
\Pi(l_{ij}\notin B_{n,ij}) &\leq e^{-Cn\eps_n^2} \label{eqa:gp_complement}\\
\Pi(\Vert l_{ij}-l_{0,ij}\Vert_\infty <2\eps_n) &\geq e^{-n\eps_n^2}  \label{eqa:gp_priormass}
\end{align}
Now set $\Theta_n=\{{\bf L}: l_{ij} \in B_{n,ij}\}$, and $N(\eps_n, \Theta_n, d_n) = \max_{1\leq j\leq i\leq D} N(3\eps_{n,ij}, B_{n,ij}, \Vert \cdot \Vert_\infty)$.
By Lemma \ref{lem:hKVbd} and \eqref{eqa:gp_entropy}, we have the following global entropy bound because $d_n^2({\bf L},{\bf L}')\leq \Vert {\bf L}-{\bf L}'\Vert_\infty\leq \eps_n^2$ for $\forall {\bf L}, {\bf L}'\in\Theta_n$.
\begin{equation*}
\log N(\eps_n, \Theta_n, d_n) \leq 6Cn(\eps_n^2)^2 \leq C'n\eps_n^4 \leq n\eps^2
\end{equation*}
which is stronger than the local entropy condition \eqref{eqa:entropy}.
Now by Lemma \ref{lem:hKVbd} and \eqref{eqa:gp_priormass} we have
\begin{align*}
\Pi_n(\bar B_n({\bf L}_0, \eps_n)) & \geq \Pi_n(\Vert {\bf L}_0-{\bf L}\Vert_\infty\leq \eps_n^2, \Vert {\bf L}_0-{\bf L}\Vert_\infty^2\leq \eps_n^2) 
= \Pi_n(\Vert {\bf L}_0-{\bf L}\Vert_\infty\leq \eps_n^2) \\
&= \Pi(\Vert l_\mathrm{argmax}-l_{0,\mathrm{argmax}}\Vert_\infty <\eps_n^2) 
\geq e^{-n(\eps_n^2/2)^2} = e^{-n\eps_n^4/4} \geq e^{-n\eps_n^2\kappa^2/4}
\end{align*}
Then \eqref{eqa:priormass} is immediately satisfied because the numerator is bounded by 1.
Finally, by \eqref{eqa:gp_complement} we have
\begin{align*}
\Pi_n(\Theta\backslash \Theta_n) \leq \sum_{i=1}^D\sum_{j=1}^i \Pi(l_{ij}\notin B_{n,ij}) \leq \frac{D(D+1)}2 e^{-Cn\eps_n^2} = o(e^{-2n\eps_n^2})
\end{align*}
Then \eqref{eqa:complement} holds because the denominator is bounded below by a term ($e^{-n\eps_n^4/4}$) of smaller order.
Therefore the proof is completed.
\end{proof}


\section{Spherical Hamiltonian Monte Carlo}\label{apx:sphHMC}
\subsection{Derivation of the geometric integrator for SphHMC}
The Lagrangian dynamics \eqref{eq:LD} on the sphere $\mathcal S^{D-1}(r)$ with the first $(D-1)$ coordinates
can be split into the following two smaller dynamics:

\noindent
\begin{subequations}
\begin{minipage}{.5\textwidth}
\begin{equation}\label{eq:geodesic}
\begin{dcases}
\begin{aligned}
\dot \bq_{-D} & = \bv_{-D}\\
\dot \bv_{-D} & = - \tp \bv_{-D} \bGamma(\bq_{-D}) \bv_{-D}
\end{aligned}
\end{dcases}
\end{equation}
\end{minipage}
\begin{minipage}{.5\textwidth}
\begin{equation}\label{eq:residual}
\begin{dcases}
\begin{aligned}
\dot \bq_{-D} & = {\bf 0}\\
\dot \bv_{-D} & = - \nabla^{-1}_{\bq_{-D}} \tilde U(\bq_{-D})
\end{aligned}
\end{dcases}
\end{equation}
\end{minipage}\par\medskip\noindent
\end{subequations}
where \eqref{eq:geodesic} is the equation of geodesic on manifold $\mathcal S^D$ which has analytical solution;
and \eqref{eq:residual} has analytical solution. Both define volume preserving maps.

The mapping $\mathcal I:\, \bq_{-D} \mapsto \bq=(\bq_{-D}, q_D)$ can be viewed as an imbedding of $\mathcal S^{D-1}_+$ into $\mathbb R^{D}$.
Denote its Jacobian as $d \mathcal I(\bq):=\begin{bmatrix} {\bf I}_{D-1}\\ -\frac{\tp{\bq}_{-D}}{q_D}\end{bmatrix}$.
Then we have
\begin{equation*}
\begin{aligned}
\tp{d \mathcal I(\bq)} d \mathcal I(\bq) &= \bG(\bq_{-D}), & d \mathcal I(\bq) \bG(\bq_{-D})^{-1} \tp{d \mathcal I(\bq)} &=\mathcal P(\bq) = {\bf I} - r^{-2} \bq\tp \bq \\
\nabla_{\bq_{-D}} \tilde U(\bq) &= \tp{d \mathcal I(\bq)} \nabla_{\bq} \tilde U(\bq), & \bv = d \mathcal I(\bq) \bv_{-D}, \qquad & \tp \bv \bv = \tp \bv_{-D} \bG(\bq_{-D}) \bv_{-D}
\end{aligned}
\end{equation*}

Then Equation \eqref{eq:geodesic} has the following solution with full coordinates
\begin{equation}\label{eq:geod}\small
\begin{aligned}
\begin{bmatrix}
\bq(t) \\ \bv(t)
\end{bmatrix}
&= \begin{bmatrix}
{\bf I} & {\bf 0}\\
\tp{\bf 0} & r^{-1}\Vert \bv(0)\Vert_2
\end{bmatrix}
\begin{bmatrix}
\cos(r^{-1}\Vert \bv(0)\Vert_2 t) & \sin(r^{-1}\Vert \bv(0)\Vert_2 t)\\
-\sin(r^{-1}\Vert \bv(0)\Vert_2 t) & \cos(r^{-1}\Vert \bv(0)\Vert_2 t)
\end{bmatrix}
\begin{bmatrix}
{\bf I} & {\bf 0}\\
\tp{\bf 0} & r\Vert \bv(0)\Vert_2^{-1}
\end{bmatrix}
\begin{bmatrix}
\bq(0) \\ \bv(0)
\end{bmatrix}\\
&= \begin{bmatrix}
\bq(0) \cos(r^{-1}\Vert \bv(0)\Vert_2 t) + r \bv(0) \Vert \bv(0)\Vert_2^{-1} \sin(r^{-1}\Vert \bv(0)\Vert_2 t)\\
- r^{-1} \bq(0) \Vert \bv(0)\Vert_2 \sin(r^{-1}\Vert \bv(0)\Vert_2 t) + \bv(0) \cos(r^{-1}\Vert \bv(0)\Vert_2 t)
\end{bmatrix}
\end{aligned}
\end{equation}
and Equation \eqref{eq:residual} has the following solution in full coordinates
\begin{equation}\label{eq:resid}
\begin{aligned}
\bq(t) &= \bq(0)\\
\bv(t) &= \bv(0) - \frac{t}{2}
d \mathcal I(\bq(0)) \nabla^{-1}_{\bq_{-D}} \tilde U(\bq(0))
= \bv(0) - \frac{t}{2} \mathcal P(\bq) \nabla_{\bq} \tilde U(\bq(0))
\end{aligned}
\end{equation}

So numerically updating \eqref{eq:resid} for $h/2$, updating \eqref{eq:geod} for $h$ and updating \eqref{eq:resid} for another $h/2$ 
yield the integrator \eqref{eq:sphHMC_prop}.

\subsection{Reformulating Acceptance}
At the end of the numerical simulation, a proposal $(\bq_T, \bv_T)$ is accepted according to the following probability
\begin{equation}\label{eq:sphHMC_acpt_classic}
a_{sphHMC} = 1 \wedge \exp(-\Delta E),\qquad \Delta E = E(\bq_T, \bv_T) - E(\bq_0, \bv_0)
\end{equation}
Such classic definition of acceptance probability can be reformulated by replacing $\Delta E$ in \eqref{eq:sphHMC_acpt_classic} with
\begin{equation*}\label{eq:sphHMC_acpt_reform}
\Delta E=\sum_{\tau=1}^T \Delta E_{\tau} \qquad \Delta E_{\tau} = E(\bq_{\tau}, \bv_{\tau}) - E(\bq_{\tau-1}, \bv_{\tau-1}) 
\end{equation*}
With \eqref{eq:sphHMC_prop} we can write
\[\small
\begin{split}
\Delta E' =& E(\bq',\bv') - E(\bq,\bv)\\
=& \tilde U(\bq') - \tilde U(\bq) + \half \tp{\bv'}_{-D} \bG(\bq'_{-D}) \bv'_{-D} - \half \tp \bv_{-D} \bG(\bq_{-D}) \bv_{-D} \\
=& \Delta \tilde U - \half \Vert \bv \Vert_2^2
     + \half \left\Vert \bv^+ - \frac{h}{2} \mathcal P(\bq') \nabla_{\bq}\tilde U(\bq') \right\Vert_2^2\\
=& \Delta \tilde U - \half \Vert \bv \Vert_2^2 + \half \tp{\bv^+} \bv^+ - \frac{h}{2} \tp{\bv^+} \mathcal P(\bq') \nabla_{\bq}\tilde U(\bq') + \frac{h^2}{8} \tp{\nabla_{\bq} \tilde U(\bq')} \mathcal P(\bq') \nabla_{\bq} \tilde U(\bq') \\
=& \Delta \tilde U - \half \Vert \bv \Vert_2^2 + \half \Vert \bv^-\Vert_2^2  - \frac{h}{2} \tp{\bv^+} \nabla_{\bq}\tilde U(\bq') + \frac{h^2}{8} \Vert \nabla_{\bq}\tilde U(\bq')\Vert_{\mathcal P(\bq')}^2\\
=& \Delta \tilde U - \half \Vert \bv \Vert_2^2 - \frac{h}{2} \tp{\bv^+} \nabla_{\bq}\tilde U(\bq') + \frac{h^2}{8} \Vert \nabla_{\bq}\tilde U(\bq')\Vert_{\mathcal P(\bq')}^2 
  + \half \Vert \bv \Vert_2^2 - \frac{h}{2} \tp \bv \nabla_{\bq}\tilde U(\bq) + \frac{h^2}{8} \Vert \nabla_{\bq}\tilde U(\bq)\Vert_{\mathcal P(\bq)}^2 \\
=& \Delta \tilde U - \frac{h}{2} \left[ \tp{\bv'} \nabla_{\bq}\tilde U(\bq') + \tp \bv \nabla_{\bq}\tilde U(\bq) \right] - \frac{h^2}{8} \left[ \Vert \nabla_{\bq}\tilde U(\bq')\Vert_{\mathcal P(\bq')}^2 - \Vert \nabla_{\bq}\tilde U(\bq)\Vert_{\mathcal P(\bq)}^2 \right]
\end{split}
\]
where $\mathcal P(\bq') \bv^+ = \bv^+$, $\mathcal P(\bq) \bv^-=\bv^-$, and $\Vert \bv^+\Vert_2^2 = \Vert \bv^-\Vert_2^2$.
Accumulating the above terms over $\tau=1,\cdots,T$ yields the reformulated acceptance probability \eqref{eq:sphHMC_acpt}.

We now prove the energy conservation theorem \ref{thm:energy_conserv} \citep{beskos11}.
\begin{proof}[Proof of Theorem \ref{thm:energy_conserv}]
With the second equation of Lagrangian dynamics \eqref{eq:LD} we have
\begin{equation*}
\begin{aligned}
- \langle \bv(t), \bg(\bq(t))\rangle &= \tp{\bv(t)} \nabla_{\bq}\tilde U(\bq(t)) = \tp{\bv_{-D}(t)} \tp{d \mathcal I(\bq)} \nabla_{\bq}\tilde U(\bq(t)) = \tp{\bv_{-D}(t)} \nabla_{\bq_{-D}}\tilde U(\bq(t)) \\
&= \tp{\bv_{-D}(t)} \bG(\bq_{-D}(t)) \left[ \dot \bv_{-D}(t) + \tp \bv_{-D}(t)\bGamma(\bq_{-D}(t)) \bv_{-D}(t) \right] \\
&= \tp{\bv_{-D}(t)} \bG(\bq_{-D}(t)) \dot \bv_{-D}(t) + \half \tp{\bv_{-D}(t)} d\bG(\bq_{-D}(t)) \bv_{-D}(t) \\
&= \frac{d}{dt} \half \tp{\bv_{-D}(t)} \bG(\bq_{-D}(t)) \bv_{-D}(t) = \frac{d}{dt} \half \Vert \bv(t) \Vert_2^2
\end{aligned}
\end{equation*}
Then we have the first equality hold because
\begin{equation*}
- \int_0^{T} \langle \bv(t), \bg(\bq(t))\rangle dt = \half \Vert \bv(T) \Vert_2^2 - \half \Vert \bv(0) \Vert_2^2
\end{equation*}
Lastly, from the first equation of Lagrangian dynamics \eqref{eq:LD}
\begin{equation*}
\tilde U(\bq(T)) - \tilde U(\bq(0)) = \int_0^T \dot{\tilde U}(\bq(t)) = \int_0^T \langle \dot \bq(t), \nabla_{\bq} \tilde U(\bq(t)) \rangle dt = \int_0^{T} \langle \bv(t), \bg(\bq(t))\rangle dt
\end{equation*}
Therefore the second equality is proved.
\end{proof}

\begin{algorithm}[t]
\caption{Adaptive Spherical HMC (adp-SphHMC)}
\label{alg:adp-sphHMC}
\begin{algorithmic}\small
\STATE Given $\bq_0, a_0, N, N^{\textrm adapt}$.
\STATE Set $h_0=1$ or using Algorithm 4 of \cite{hoffman14}, $\mu=\log(10h_0), \bar h_0=1, \bar A_0=0, \gamma=0.05, n_0=10, \kappa=0.75$.
\FOR{$n=1$ to $N$}
\STATE Sample a new velocity $\bv_{n-1} \sim \mathcal N({\bf 0},{\bf I}_D)$, and set $\bv_{n-1} = \mathcal P(\bq_{n-1}) \bv_{n-1}$.
\STATE Set $\bq^{(0)} = \bq_{n-1}$, $\bv^{(0)} = \bv_{n-1}$.
\FOR{$\tau=0$ to $T-1$ ($T =T_{2orth}\; \textrm{or}\; T_{stoch}$)}
\STATE Run leapfrog step \eqref{eq:sphHMC_prop} to update $(\bq^{(\tau+1)}, \bv^{(\tau+1)}) \leftarrow \mathcal T_{h_{n-1}}(\bq^{(\tau)}, \bv^{(\tau)})$.
\IF{Stopping criterion \eqref{eq:sphHMC_2orth} (or \eqref{eq:sphHMC_stoch}) is satisfied}
\STATE Break
\ENDIF
\ENDFOR
\STATE Accept the proposal $(\bq^{(T)},\bv^{(T)})$ with probability $a^{sphHMC}_n$ in \eqref{eq:sphHMC_acpt} and set $\bq_n = \bq^{(T)}$; otherwise set $\bq_n = \bq_{n-1}$.
\IF{$n\leq N^{\textrm adapt}$}
\STATE Set $\bar A_n = \left(1-\frac{1}{n+n_0}\right) \bar A_{n-1} + \frac{1}{n+n_0} (a_0-a_n)$.
\STATE Set $\log h_n = \mu - \frac{\sqrt n}{\gamma} \bar A_n$, and $\log \bar h_n = n^{-\kappa} \log h_n + (1- n^{-\kappa}) \log \bar h_{n-1}$.
\ELSE
\STATE Set $h_n=\bar h_{N^{\textrm adpat}}$.
\ENDIF
\ENDFOR
\end{algorithmic}
\end{algorithm}


\section{Gradient Calculation in Normal-inverse-Wishart Problem}\label{apx:calculations}
We use the representation \eqref{eq:strt_corr} and derive log-posterior (log-likelihood and log-prior) and the corresponding gradients for \eqref{eq:conj-iwishart} using matrix calculus.
\subsection{Gradients of log-likelihood}
Denote ${\bf y}_n^*:=({\bf y}_n-\vect\mu_0)/\vect\sigma$. Then the log-likelihood becomes
\begin{equation*}
\ell ({\bf y}^*; \vect\sigma, \vect\Rho) = - N \tp{\bm 1}_D \log\vect\sigma -\frac{N}{2} \log|\vect\Rho| -\half \sum_{n=1}^N \tp{{\bf y}_n^*} \vect\Rho^{-1} {\bf y}_n^*
\end{equation*}

\noindent $\left[\frac{\pa \ell}{\pa \vect\tau}\right]$. \quad
We calculate the gradient of log-likelihood with respect to $\vect\sigma$
\begin{equation*}
\frac{\pa \ell}{\pa \sigma_k} = -N\sigma_k^{-1} + \sum_{n=1}^N \sum_i \frac{y^*_{ni}}{\sigma_i} \delta_{ik} (\vect\Rho^{-1} {\bf y}_n^*)_i, \quad i.\,e.\; 
\frac{\pa \ell}{\pa \vect\sigma} = -N \vect\sigma^{-1} + \sum_{n=1}^N \diag({\bf y}^*_n/\vect\sigma)  (\vect\Rho^{-1} {\bf y}_n^*)
\end{equation*}
And with the transformation $\vect\tau=\log(\vect\sigma)$ it becomes
\begin{equation*}
\frac{\pa \ell}{\pa \vect\tau} = \frac{\tp{d\vect\sigma}}{d\vect\tau} \frac{\pa \ell}{\pa \vect\sigma} 
=  \diag(\vect\sigma) \left[ -\frac{N}{\vect\sigma} + \sum_{n=1}^N \diag({\bf y}^*_n/\vect\sigma) (\vect\Rho^{-1} {\bf y}_n^*) \right] = -N {\bm 1}_D + \sum_{n=1}^N \diag({\bf y}^*_n)   (\vect\Rho^{-1} {\bf y}_n^*)
\end{equation*}

\noindent $\left[\frac{\pa \ell}{\pa {\bf U}^*} \; \left( \frac{\pa \ell}{\pa {\bf L}} \right)\right]$. \quad
When $\vect\Rho = {\bf U}^* \tp{({\bf U}^*)}$, $\half \log|\vect\Rho| = \log|{\bf U}^*| = \tp{\bm 1}_D \log|\diag({\bf U}^*)|$ and thus we have
\begin{equation*}
\frac{\pa \ell}{\pa {\bf U}^*} = - \frac{N{\bf I}_D}{{\bf U}^*} + \sum_{n=1}^N \frac{d g_n(\widetilde{\bf U})}{d{\bf U}^*}
\end{equation*}
where $\frac{{\bf I}_D}{{\bf U}^*}=\diag(\{(u^*_{ii})^{-1}\})$ is a diagonal matrix formed by element-wise division, $\widetilde{\bf U}:=({\bf U}^*)^{-1}$ and $g_n(\widetilde{\bf U}):= -\half \tp{{\bf y}_n^*} \tp{\widetilde{\bf U}} \widetilde{\bf U} {\bf y}_n^*$.
Taking differential directly on $g_n({\bf U}^*):= -\half \tp{{\bf y}_n^*} ({\bf U}^*)^{-\mathsf T} ({\bf U}^*)^{-1} {\bf y}_n^*$, and
noting that differential and trace operators are exchangeable, we have
\begin{equation*}
\begin{aligned}
d g_n({\bf U}^*) &= -\half \tr ( \tp{{\bf y}_n^*} d ({\bf U}^*)^{-\mathsf T} ({\bf U}^*)^{-1} {\bf y}_n^* + \tp{{\bf y}_n^*} ({\bf U}^*)^{-\mathsf T} d ({\bf U}^*)^{-1} {\bf y}_n^* ) \\
&= \half \left[ \tr (\tp{{\bf y}_n^*} ({\bf U}^*)^{-\mathsf T} d \tp{({\bf U}^*)} \vect\Rho^{-1} {\bf y}_n^* ) + \tr ( \tp{{\bf y}_n^*} \vect\Rho^{-1} d {\bf U}^* ({\bf U}^*)^{-1} {\bf y}_n^* ) \right] \\
&= \tr ( \tp{{\bf y}_n^*} \vect\Rho^{-1} d {\bf U}^* ({\bf U}^*)^{-1} {\bf y}_n^* ) = \tr ( ({\bf U}^*)^{-1} {\bf y}_n^* \tp{{\bf y}_n^*} \vect\Rho^{-1} d {\bf U}^* )
\end{aligned}
\end{equation*}
Conversion from differential to normal derivative form in the numerator layout \citep{minka00} yields
\begin{equation*}
\frac{\pa g_n({\bf U}^*)}{\pa \tp{({\bf U}^*)}} = \mathrm{tril} ( ({\bf U}^*)^{-1} {\bf y}_n^* \tp{{\bf y}_n^*} \vect\Rho^{-1} ) , \quad i.e.,\;
\frac{\pa g_n({\bf U}^*)}{\pa {\bf U}^*} = \mathrm{triu} ( \vect\Rho^{-1} {\bf y}_n^* \tp{{\bf y}_n^*} ({\bf U}^*)^{-\mathsf T} )
\end{equation*}
Finally, we have
\begin{equation*}
\frac{\pa \ell}{\pa {\bf U}^*} = - \frac{N{\bf I}_D}{{\bf U}^*} + \mathrm{triu} ( \vect\Rho^{-1} \sum_{n=1}^N {\bf y}_n^* \tp{{\bf y}_n^*} ({\bf U}^*)^{-\mathsf T} )
\end{equation*}

When $\vect\Rho = {\bf L} \tp{\bf L}$, 
by similar argument as above, we have
\begin{equation*}
\frac{\pa \ell}{\pa {\bf L}} = - \frac{N{\bf I}_D}{{\bf L}} + \mathrm{tril} ( \vect\Rho^{-1} \sum_{n=1}^N {\bf y}_n^* \tp{{\bf y}_n^*} {\bf L}^{-\mathsf T} )
\end{equation*}

\subsection{Gradients of log-priors}
The logarithm of conditional prior $p(\vect\sigma| {\bf U}^*)$ after transformation $\vect\tau=\log(\vect\sigma)$ becomes
\begin{equation*}
\begin{aligned}
\log p(\vect\tau| {\bf U}^*) = \log p(\vect\sigma| {\bf U}^*) +  \log \left| \frac{d \vect\sigma}{d \vect\tau}\right| 
&= \sum_{i=1}^D (i-(\nu+D))\tau_i  -\half\tr (\vect\Psi \diag(e^{-\vect\tau}) \vect\Rho^{-1} \diag(e^{-\vect\tau}))
\end{aligned}
\end{equation*}

\noindent $\left[\frac{d}{d\vect\tau} \log p(\vect\tau| {\bf U}^*)\right]$. \quad
We calculate the derivative of $\log p(\vect\tau| {\bf U}^*)$ with respect to $\vect\tau$
\begin{equation*}
\frac{d}{d\vect\tau} \log p(\vect\tau| {\bf U}^*) = {\bf i} -(\nu+D) + \frac{d g(\vect\tau)}{d \vect\tau}
\end{equation*}
where ${\bf i}=\tp{[1,\cdots,D]}$, and $g(\vect\tau) = -\half\tr (\vect\Psi \diag(e^{-\vect\tau}) \vect\Rho^{-1} \diag(e^{-\vect\tau}))$.

Noting that differential and trace operators are exchangeable, we have
\begin{equation*}
\begin{aligned}
d g(\vect\tau) &= -\half\tr (\vect\Psi d \diag(e^{-\vect\tau}) \vect\Rho^{-1} \diag(e^{-\vect\tau}) + \vect\Psi \diag(e^{-\vect\tau}) \vect\Rho^{-1} d \diag(e^{-\vect\tau})) \\
&= \half \left[ \tr (\vect\Rho^{-1} \diag(e^{-\vect\tau}) \vect\Psi \diag(e^{-\vect\tau}) \diag(d \vect\tau) ) + \tr ( \vect\Psi \diag(e^{-\vect\tau}) \vect\Rho^{-1} \diag(e^{-\vect\tau}) \diag(d \vect\tau) ) \right] \\
&= \sum_{i=1}^D d \tau_i \sum_{j=1}^D \psi_{ij} e^{-\tau_j} \rho^{ji} e^{-\tau_i}
\end{aligned}
\end{equation*}
Thus
\begin{equation*}
\frac{d g(\vect\tau)}{d \vect\tau} = \diag (\vect\Psi \diag(e^{-\vect\tau}) \vect\Rho^{-1}) \diag(e^{-\vect\tau}) 
= \diag (\vect\Psi \diag(e^{-\vect\tau}) \vect\Rho^{-1}) \circ e^{-\vect\tau}
\end{equation*}
where $\diag$ acting on a vector forms a diagonal matrix while the action a matrix means extracting the diagonal vector. $\circ$ is the Hadamard product (a.k.a. Schur product), i.e. the entrywise product.

\noindent $\left[\frac{d}{d{\bf U}^*} \log p({\bf U}^*|\vect\tau)\right]$. \quad
Now consider the derivative of $\log p({\bf U}^*|\vect\tau)$ 
with respect to the matrix ${\bf U}^*$.
We have
\begin{equation*}
\frac{d}{d{\bf U}^*} \log p({\bf U}^*|\vect\tau) = \frac{\diag({\bf i}-(\nu+D+1))}{{\bf U}^*} + \frac{d g({\bf U}^*)}{d {\bf U}^*}
\end{equation*}
where $g({\bf U}^*)=-\half\tr (\vect\Psi \diag(e^{-\vect\tau}) ({\bf U}^*)^{-\mathsf T} ({\bf U}^*)^{-1} \diag(e^{-\vect\tau}))$, 
and $\frac{\diag({\bf i})}{{\bf U}^*}$ is a diagonal matrix formed by element-wise division.

Again by the exchangeability between differential and trace, we have
\begin{equation*}\small
\begin{aligned}
&d g({\bf U}^*)\\ &= -\half\tr (\vect\Psi \diag(e^{-\vect\tau}) d ({\bf U}^*)^{-\mathsf T} ({\bf U}^*)^{-1} \diag(e^{-\vect\tau}) + \vect\Psi \diag(e^{-\vect\tau}) ({\bf U}^*)^{-\mathsf T} d ({\bf U}^*)^{-1} \diag(e^{-\vect\tau})) \\
&= \half \left[ \tr (\vect\Psi \diag(e^{-\vect\tau}) ({\bf U}^*)^{-\mathsf T} d \tp{({\bf U}^*)} \vect\Rho^{-1} \diag(e^{-\vect\tau}) ) + \tr ( \vect\Psi \diag(e^{-\vect\tau}) \vect\Rho^{-1} d {\bf U}^* ({\bf U}^*)^{-1} \diag(e^{-\vect\tau}) ) \right] \\
&= \half \left[ \tr (\diag(e^{-\vect\tau}) \vect\Rho^{-1} d {\bf U}^* ({\bf U}^*)^{-1} \diag(e^{-\vect\tau}) \vect\Psi  ) + \tr ( \vect\Psi \diag(e^{-\vect\tau}) \vect\Rho^{-1} d {\bf U}^* ({\bf U}^*)^{-1} \diag(e^{-\vect\tau}) ) \right] \\
&= \tr ( ({\bf U}^*)^{-1} \diag(e^{-\vect\tau}) \vect\Psi \diag(e^{-\vect\tau}) \vect\Rho^{-1} d {\bf U}^*   )
\end{aligned}
\end{equation*}
Therefore we have
\begin{equation*}
\frac{d g({\bf U}^*)}{d \tp{({\bf U}^*)}} = \mathrm{tril}( ({\bf U}^*)^{-1} \diag(e^{-\vect\tau}) \vect\Psi \diag(e^{-\vect\tau}) \vect\Rho^{-1} ) \,,
\end{equation*}
that is,
\begin{equation*}
\frac{d g({\bf U}^*)}{d {\bf U}^*} = \mathrm{triu}( \vect\Rho^{-1} \diag(e^{-\vect\tau}) \vect\Psi \diag(e^{-\vect\tau}) ({\bf U}^*)^{-\mathsf T} ) \,,
\end{equation*}

\noindent $\left[\frac{d}{d\vect\tau} \log p(\vect\tau),\; \frac{d}{d{\bf L}} \log p({\bf L})\right]$. \quad
Lastly, the log-priors for \eqref{eq:nonconj-iwishart} and their gradients after transformation $\vect\tau:=\log(\vect\sigma)$ are calculated
\begin{equation*}
\begin{aligned}
\log p(\vect\tau) &= -\half \tp{\vect\tau} \vect\tau, & \frac{d}{d\vect\tau} \log p(\vect\tau) &= -\vect\tau \\
\log p({\bf l}_i) &= \log p({\bf l}_i^2) + \tp{\bm 1}_i\log|2{\bf l}_i|  = \tp{(2(\vect\alpha_i-1)+{\bm 1}_i)} \log |{\bf l}_i|,  
& \frac{d}{d{\bf l}_i} \log p({\bf l}_i) &= \frac{2(\vect\alpha_i-1)+{\bm 1}_i}{{\bf l}_i}
\end{aligned}
\end{equation*}
The bottom row can be written as
\begin{equation*}
\begin{aligned}
\log p({\bf L}) &= \sum_{i=1}^D \tp{(2\vect\alpha_i-1)} \log |{\bf l}_i|, &
\frac{d}{d\bf L} \log p({\bf L}) &= \frac{2\vect\alpha-1}{\bf L}
\end{aligned}
\end{equation*}
where $\frac{1}{\bf L}$ denotes a lower-triangular matrix with $l_{ij}^{-1}$ being its $(i,j)$ entry $(i\ge j)$.

\section{More Numerical Results}

\subsection{Flexibility of von Mises-Fisher Prior and Bingham Prior}\label{apx:morepriors}

\begin{figure}[t] 
   \centering
   \includegraphics[width=1\textwidth,height=.45\textwidth]{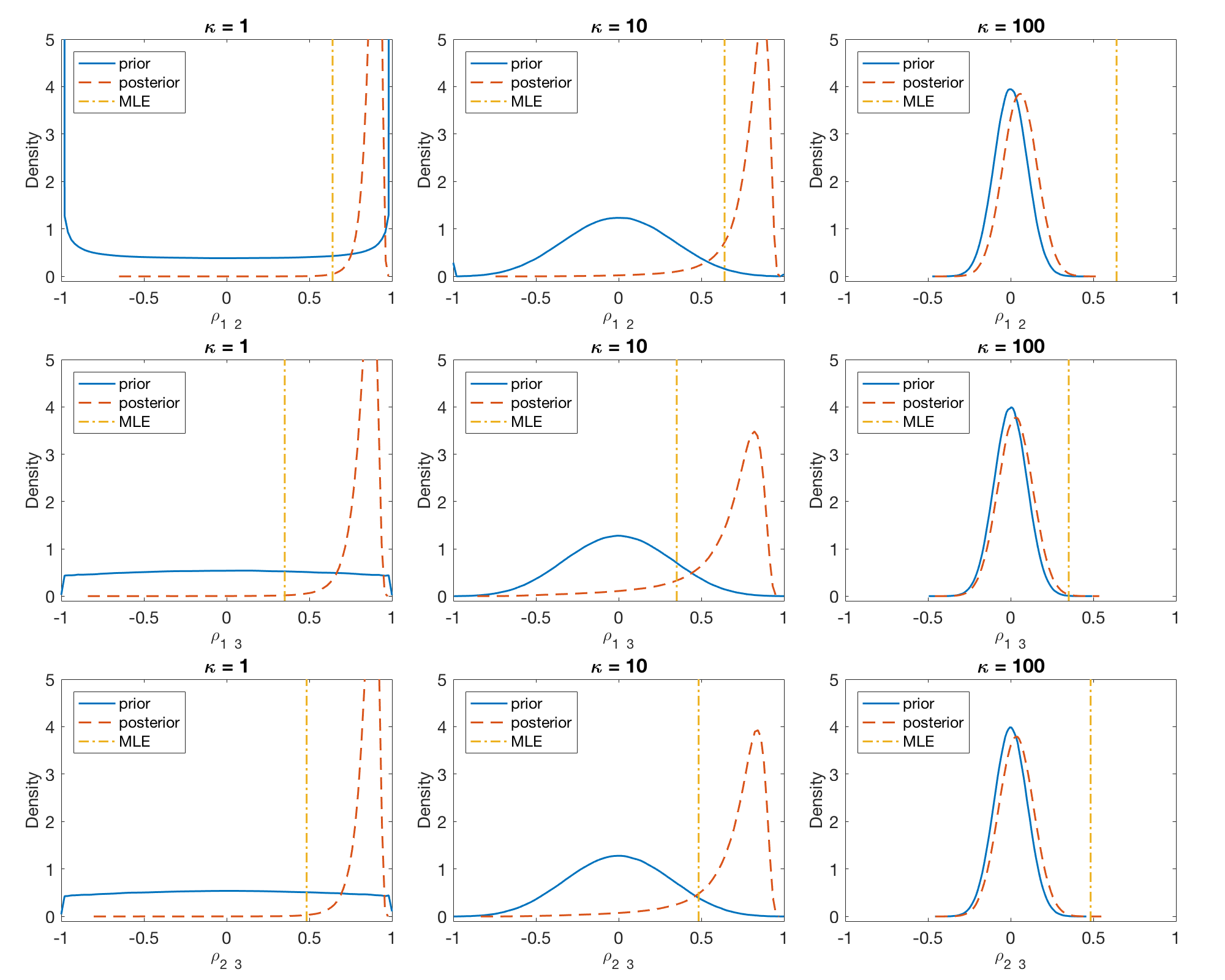} 
   \caption{Marginal posterior, prior (induced from von Mises-Fisher distribution) densities of correlations and MLEs with different settings for concentration parameter $\kappa$, estimated with $10^6$ samples.}
   \label{fig:post_corr_vmf}
\end{figure}

\begin{dfn}[Fisher-Bingham / Kent distribution]
The probability density function of the Kent distribution for the random vector ${\bf l}_i\in\mathcal S^{i-1}$ is given by
\begin{equation*}\label{eq:Kent_density}
p({\bf l}_i) \propto \exp\left\{\kappa \tp{\vect\gamma_1} {\bf l}_i + \sum_{k=2}^i \beta_k (\tp{\vect\gamma_k} {\bf l}_i)^2\right\}
\end{equation*}
where $\sum_{k=2}^i \beta_k=0$ and $0\leq 2|\beta_k|<\kappa$ and the vectors $\{\vect\gamma_k\}_{k=1}^i$ are orthonormal.
\end{dfn}
\begin{rk}
The parameters $\kappa$ and $\vect\gamma_1$ are called the \emph{concentration} and the \emph{mean direction} parameter, respectively.
The greater the value of $\kappa$, the higher the concentration of the distribution around the mean direction $\vect\gamma_1$.
The choice of $\vect\gamma_1$ could impact our priors when modeling correlations. 
Parameters $\{\beta_k\}_{k=2}^i$ determine the ellipticity of the contours of equal probability.
The vectors $\{\vect\gamma_k\}_{k=2}^i$ determine the orientation of the equal probability contours on the sphere.
\end{rk}
\begin{rk}
If $\beta_k=0$ for $k=2,\cdots, i$, then this distribution reduces to \emph{von Mises-Fisher} distribution \citep{fisher53,mardia09}, denoted as $\vmf(\kappa, \vect\gamma_1)$.
If $\kappa=0$, then it defines an antipodally symmetric distribution, named \emph{Bingham} distribution \citep{bingham74}, denoted as $\bing({\bf A})$, 
with $\tp{\bf l}_i {\bf A} {\bf l}_i=\sum_{k=2}^i \beta_k (\tp{\vect\gamma_k} {\bf l}_i)^2$.
\end{rk}

As before, to induce smaller correlations, one can put higher prior probabilities for ${\bf l}_i$ on the poles of $\mathcal S^{i-1}$. For example, we might consider ${\bf l}_i \sim \vmf(\kappa, {\bf n}_i)$,
or 
${\bf l}_i \sim \bing(\zeta\diag({\bf n}_i))$, where ${\bf n}_i:=\tp{(0,\cdots,0,1)}$ is denoted as the north pole. 
Now let's consider the following von Mises-Fisher prior \citep{fisher87,fisher53,mardia09} for ${\bf l}_i$, the $i$-th row of the Cholesky factor ${\bf L}$ of correlation matrix $\vect\Rho$ in the structured model \eqref{eq:strt_corr}.
\begin{dfn}[Von Mises-Fisher distribution]
The probability density function of the von Mises-Fisher distribution for the random vector ${\bf l}_i\in\mathcal S^{i-1}$ is given by
\begin{equation*}\label{eq:vMF_density}
p({\bf l}_i) = C_i(\kappa) \exp(\kappa \tp{\vect\mu} {\bf l}_i)
\end{equation*}
where $\kappa\geq 0$, $\Vert \vect\mu\Vert=1$ and the normalization constant $C_i(\kappa)$ is equal to
\begin{equation*}
C_i(\kappa) = \frac{\kappa^{i/2-1}}{(2\pi)^{i/2} I_{i/2-1}(\kappa)}
\end{equation*}
where $I_v$ denotes the modified \emph{Bessel} function of the first kind at order $v$.
Denote ${\bf l}_i \sim \vmf(\kappa, \vect\mu)$.
\end{dfn}

Since we have no prior knowledge about the mean direction $\vect\mu$, we choose $\vect\mu={\bf n}_i=\tp{({\bm 0}_{i-1},1)}$ that favors the polar direction, i.e.
\begin{equation*}\label{eq:nonconj-iwishart2}
{\bf l}_i \sim \vmf(\kappa, {\bf n}_i), \qquad p({\bf l}_i) \propto \exp(\kappa l_{ii}), \quad i=2,\cdots, D
\end{equation*}
where we consider i) $\kappa=1$; ii) $\kappa=10$; iii) $\kappa=100$.
With the von Mises-Fisher prior, we have
\begin{equation*}
\log p({\bf L}) = \sum_{i=1}^D \kappa l_{ii} = \kappa \tr({\bf L}), \quad \frac{d}{d\bf L} \log p({\bf L}) = \kappa {\bf I}
\end{equation*}
We repeat the experiment in Section \ref{sec:prieff} with the von Mises-Fisher prior for ${\bf l}_i$.
The posteriors, priors and maximal likelihood estimates (MLE) of correlations with different $\kappa$'s are plotted in Figure \ref{fig:post_corr_vmf} respectively.
With larger concentration parameter $\kappa$, the posterior is pulled more towards $0$.

\begin{figure}[t] 
   \centering
   \includegraphics[width=1\textwidth,height=.45\textwidth]{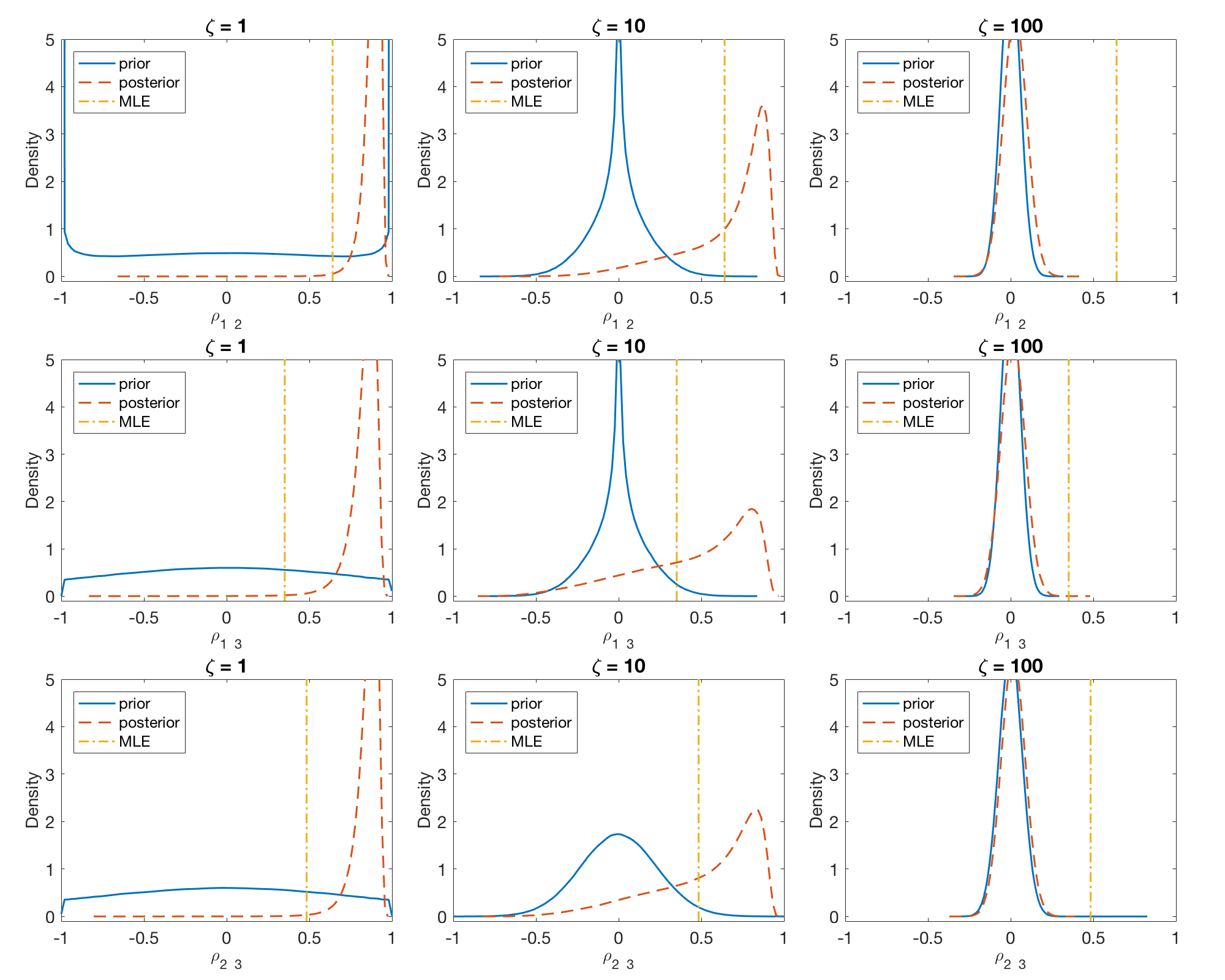} 
   \caption{Marginal posterior, prior (induced from Bingham distribution) densities of correlations and MLEs with different settings for concentration parameter $\zeta$, estimated with $10^6$ samples.}
   \label{fig:post_corr_bing}
\end{figure}

Finally, we consider the following Bingham prior \citep{bingham74,onstott80} for ${\bf l}_i$ in the structured model \eqref{eq:strt_corr}.
\begin{dfn}[Bingham distribution]
The probability density function of the Bingham distribution for the random vector ${\bf l}_i\in\mathcal S^{i-1}$ is given by
\begin{equation*}\label{eq:Bingham_density}
p({\bf l}_i) = _1\!\!F_1(\half;\,\frac{n}{2};\,{\bf Z})^{-1} \exp (\tp{\bf l}_i {\bf M} {\bf Z} \tp{\bf M} {\bf l}_i)
\end{equation*}
where ${\bf M}$ is an orthogonal orientation matrix, ${\bf Z}$ is a diagonal concentration matrix, and $_1F_1(\cdot;\,\cdot;\,\cdot)$ is a \emph{confluent hypergeometric function} of matrix argument.
Denote ${\bf l}_i \sim \bing({\bf M}, {\bf Z})$.
\end{dfn}

Note, according to \cite{bingham74}, this distribution is defined for ${\bf Z}$ up to an arbitrary scalar matrix $\zeta_0{\bf I}$.
Therefore, we consider ${\bf M}={\bf I}$ and ${\bf Z}=\zeta\diag({\bf n}_i)$ that favors the polar direction, i.e.
\begin{equation*}\label{eq:nonconj-iwishart3}
{\bf l}_i \sim \bing({\bf I}, \zeta\diag({\bf n}_i)), \qquad p({\bf l}_i) \propto \exp(\zeta l_{ii}^2), \quad i=2,\cdots, D
\end{equation*}
where we consider i) $\zeta=1$; ii) $\zeta=10$; iii) $\zeta=100$.
The log-prior and its gradient are calculated as follows
\begin{equation*}
\log p({\bf L}) = \sum_{i=1}^D \zeta l_{ii}^2 = \zeta \Vert\diag({\bf L})\Vert^2, \quad \frac{d}{d\bf L} \log p({\bf L}) = 2\zeta \diag({\bf L})
\end{equation*}
We repeat the above experiment with the Bingham prior for ${\bf l}_i$.
The posteriors, priors and maximal likelihood estimates (MLE) of correlations with different $\zeta$'s are plotted in Figure \ref{fig:post_corr_bing} respectively.
With larger concentration parameter $\zeta$, the posteriors are pulled more towards the induced priors and concentrate on 0.


\subsection{More Comparison to Latent Factor Process Model}\label{apx:more_latentfactor}
The example of simulated periodic process in Section \ref{sec:flexibility} is consider for $D=2$ for simplicity and convenience of visualization.
Here we consider higher dimension $D=10$. The purpose here is not to show the scalability, but rather to investigate the robustness of our dynamic model \eqref{eq:strt_dyn_model} in terms of full flexibility.

\begin{figure}[t] 
   \centering
   \includegraphics[width=1\textwidth,height=.3\textwidth]{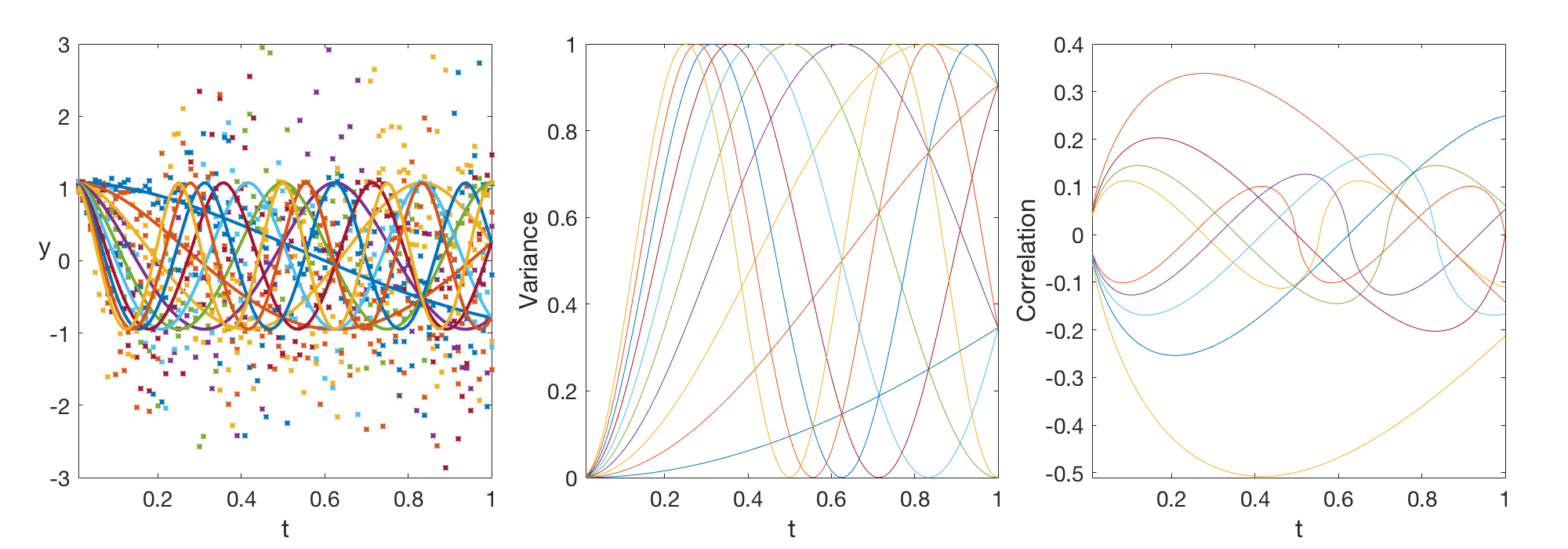} 
   \caption{Simulated data $y$ over the underlying mean functions $\mu_t$ (left), the variance functions $\Sigma_t$, and the correlation functions $\Rho_t$ (right) of 10-dimension periodic processes.}
   \label{fig:periodic_process_D10}
\end{figure}

We generate $M=20$ trials of data over $N=100$ evenly spaced points over $[0,1]$.
The true mean, variance and correlation functions are modified from the example \eqref{eq:periodic} using the Clausen functions \citep{clausen1832}.
Seen from Figure \ref{fig:periodic_process_D10},
they behave more intricately with higher heterogeneity among those processes. This could impose further challenge for latent factor based models like \eqref{eq:lpf_bncr} compared to $D=2$.
We repeat the experiments in Section \ref{sec:flexibility} and compare our dynamic model \eqref{eq:strt_dyn_model} with the latent factor process model \eqref{eq:lpf_bncr} by \cite{fox15}.
To aid the visualization, we subtract the estimated process from their true values and plot the error functions in Figure \ref{fig:periodic_constrast_more}.
Even if we have tried our best to tune the parameters, e.g. $L$, the number of basis functions, and $k$, the size of latent factors, the latent factor process model \citep{fox15} is outperformed by our flexible dynamic model \eqref{eq:strt_dyn_model} in reducing estimation errors.

\begin{figure}[t] 
   \centering
   \includegraphics[width=1\textwidth,height=.225\textwidth]{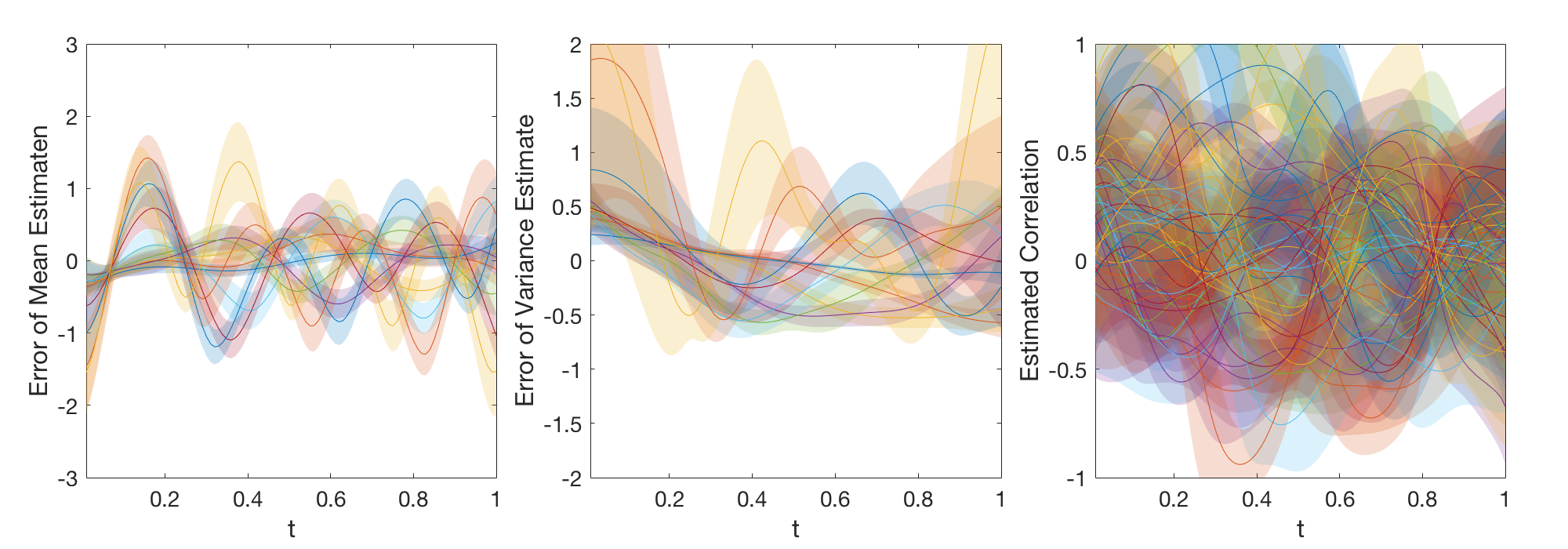}
   \includegraphics[width=1\textwidth,height=.225\textwidth]{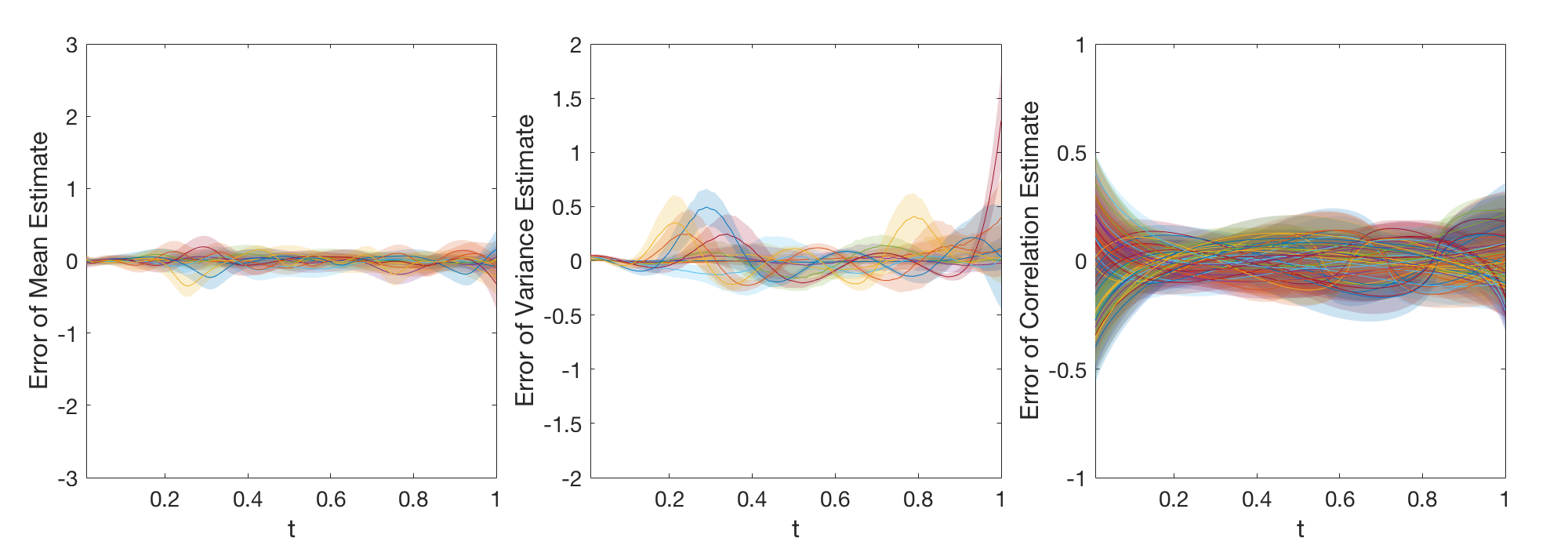}
   \caption{Estimated error functions of the underlying mean $\mu_t$ (left column), variance $\sigma_t$ (middle column) and correlation $\rho_t$ (right column) of 10-dimensional periodic processes,
   		using latent factor process model (upper row) and our flexible model (lower row), based on $M=20$ trials of data over $N=100$ evenly spaced points.
   		Solid lines are estimated errors and shaded regions are $95\%$ credible bands.}
   \label{fig:periodic_constrast_more}
\end{figure}


\begin{figure}[htbp] 
   \centering
   \includegraphics[width=1\textwidth,height=.45\textwidth]{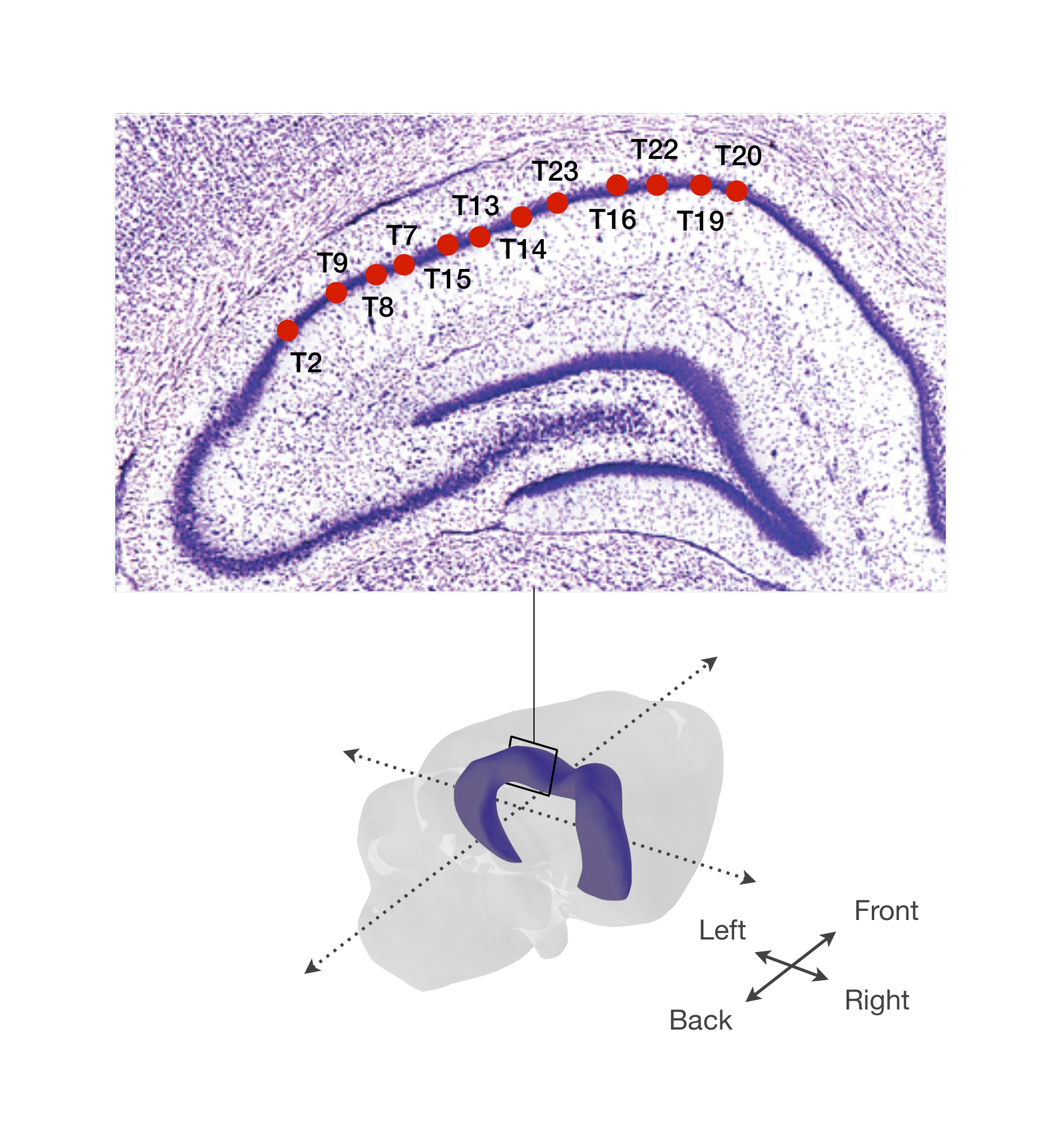} 
   \caption{Locations of recorded LFP signals in CA1 subregion of the rat's hippocampus.}
   \label{fig:LFP_tetrode_locations}
\end{figure}

\subsection{More Results on the Analysis of LFP data}\label{apx:moreLFP}
In Section \ref{sec:lfp}, we studied the LFP data collected from the hippocampus of rats performing a complex sequence memory task. 
Figure \ref{fig:LFP_tetrode_locations} shows 12 locations from CA1 subregion of the hippocampus of the rat where LFP signals are recorded.
Figure \ref{fig:LFP_theta} shows the theta-filtered traces (4-12Hz; left panel) and the estimated correlation processes under different experiment conditions (InSeq vs OutSeq; right panel).
Here we observe a similar dynamic pattern of correlation matrices under two conditions that diverge after 500ms, indicating the neural activity associated with the cognitive process of identifying whether events occurred in their expected order.

\begin{figure}[htbp] 
   \centering
   \includegraphics[width=.49\textwidth,height=.4\textwidth]{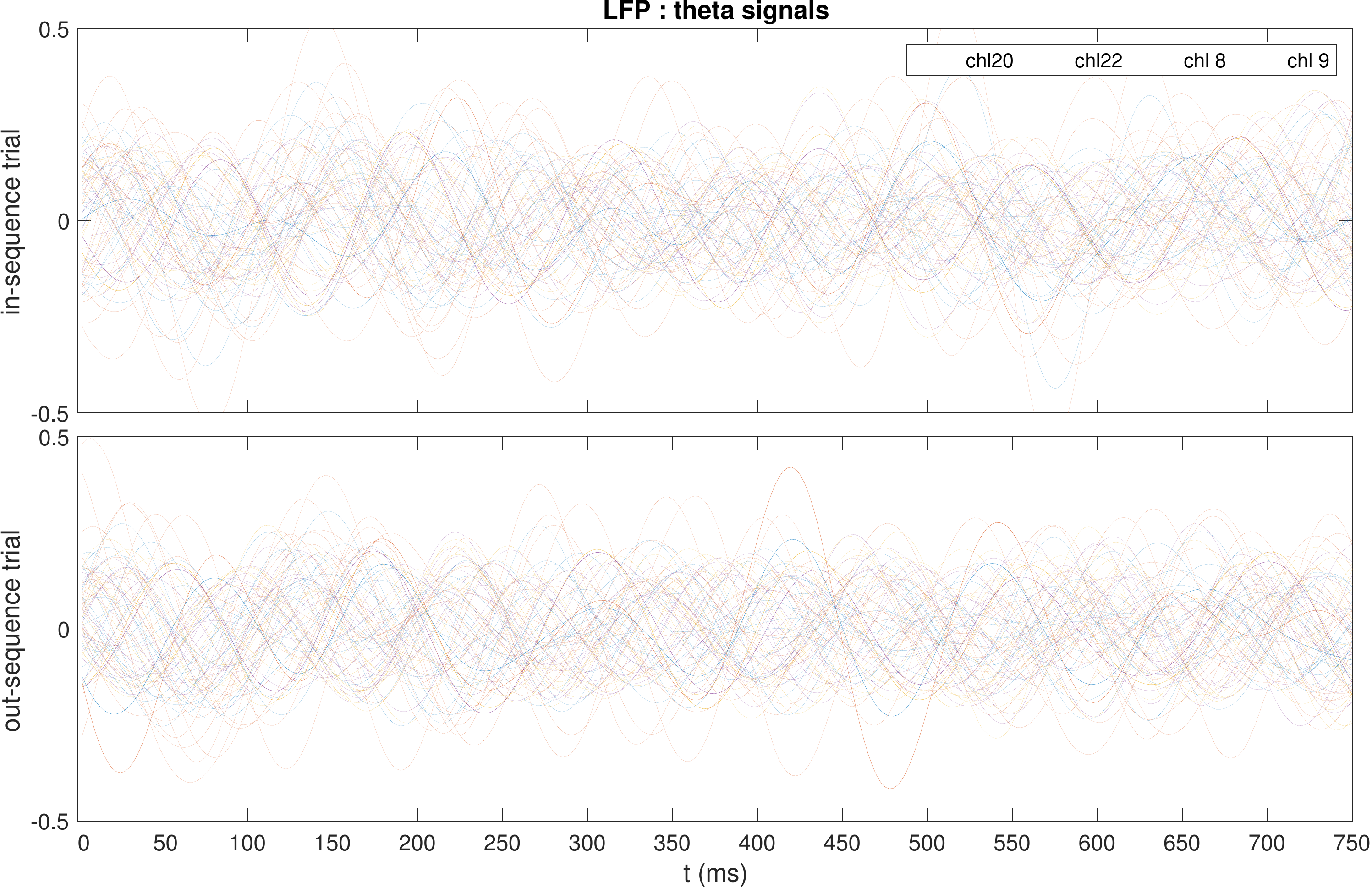}
   \includegraphics[width=.49\textwidth,height=.4\textwidth]{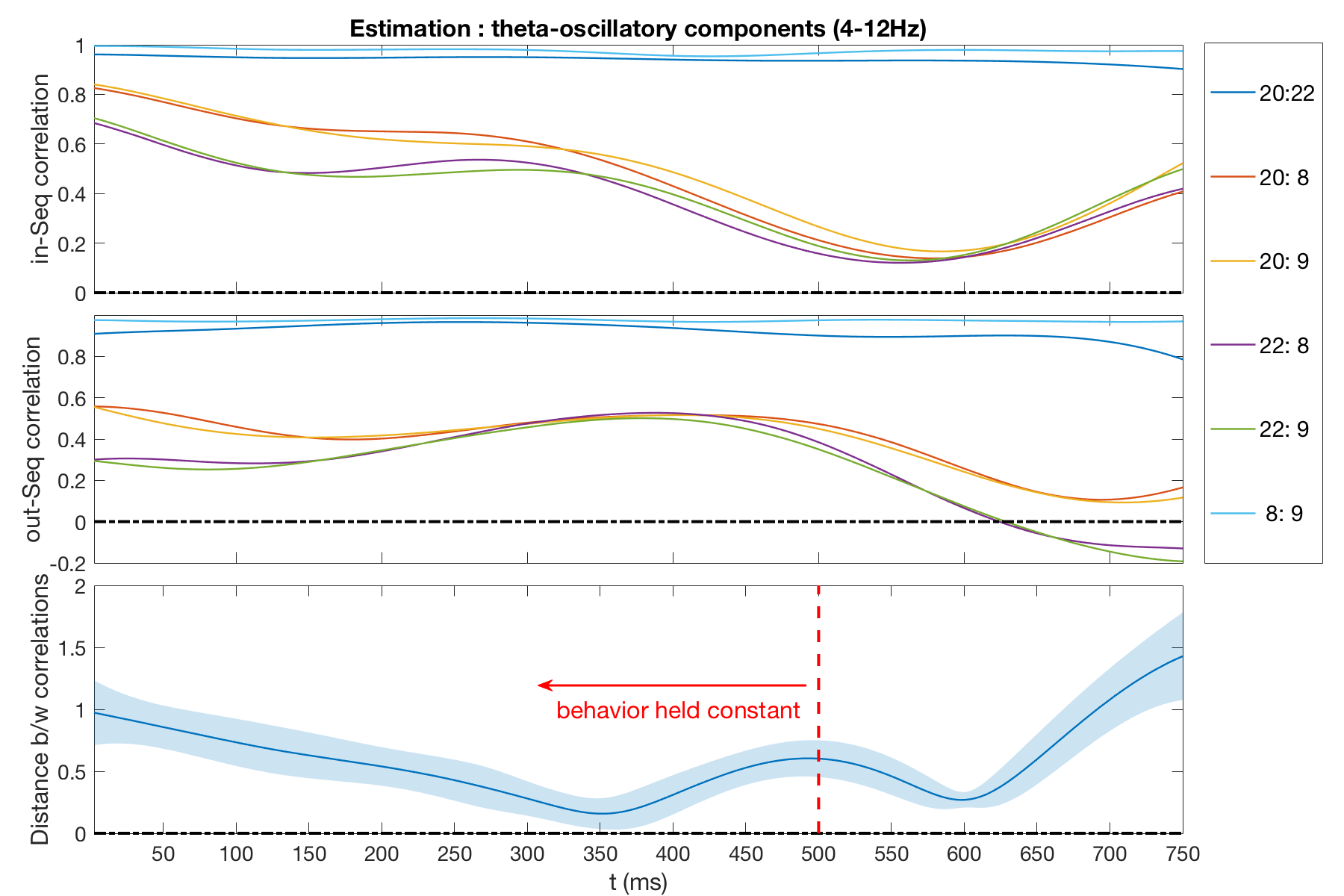}
   \caption{Results of LFP theta signals: data (left), estimation of correlations (right).}
   \label{fig:LFP_theta}
\end{figure}

We also did a study of the correlation evolution on the full 12 channels and revealed the block structure of those channels and the same changing pattern under different experimental conditions discovered with the chosen 4 channels in Section \ref{sec:lfp}.
This video demonstrates the result of 12 channels \url{https://www.youtube.com/watch?v=NMUUic0IDsM}.


\end{document}